\documentclass{article} % For LaTeX2e
\usepackage{silence}
\usepackage{iclr2027_arxiv,times}

\usepackage{amsmath,amsfonts,bm}

\def\eqref#1{equation~\ref{#1}}
\def\Eqref#1{Equation~\ref{#1}}
\def\1{\bm{1}}

\DeclareMathAlphabet{\mathsfit}{\encodingdefault}{\sfdefault}{m}{sl}
\SetMathAlphabet{\mathsfit}{bold}{\encodingdefault}{\sfdefault}{bx}{n}

\usepackage[table]{xcolor}
\usepackage{array}
\usepackage{microtype}
\usepackage{graphicx}
\usepackage{booktabs}
\usepackage{cancel}
\usepackage{multirow}
\usepackage{adjustbox}
\usepackage{makecell}
\usepackage{siunitx}
\usepackage{subcaption}
\usepackage[T1]{fontenc}

\definecolor{RankBase}{HTML}{F4F5F7}
\colorlet{RankBest}{red!15}
\definecolor{RankBestText}{HTML}{2D4354}
\colorlet{RankSecond}{blue!10}
\definecolor{RankSecondText}{HTML}{334452}
\definecolor{RankThird}{HTML}{C3CEB8}
\definecolor{RankThirdText}{HTML}{35422D}
\definecolor{OTGroupShade}{HTML}{E7EAED}
\definecolor{OTGroupText}{HTML}{4B5157}

\newcommand{\uotrankfirst}[1]{\cellcolor{RankBest}\color{RankBestText}\bfseries #1}
\newcommand{\uotranksecond}[1]{\cellcolor{RankSecond}\color{RankSecondText}\bfseries #1}
\newcommand{\uotrankthird}[1]{\cellcolor{RankThird}\color{RankThirdText}\bfseries #1}
\providecommand{\setting}[2]{%
  \textbf{#1}\nobreak\hspace{0.3em}%
  {\scriptsize\color{black!55}\mbox{#2}}%
}
\newcommand{\uotgroup}[2]{%
  \rowcolor{OTGroupShade}%
  \multicolumn{13}{@{}l}{%
    \hspace{0.45em}\strut
    {\footnotesize\color{OTGroupText}\textsc{#1}}%
    \enspace{\scriptsize\color{black!55}\textit{(#2)}}%
  }\\
}
\newcolumntype{Z}[1]{>{\columncolor{RankBase}}S[table-format=#1]}

\usepackage[table]{xcolor}

\newcommand{\bestcfg}[1]{%
    \cellcolor{blue!20}\textbf{#1}%
}

\usepackage{pgfplots}
\pgfplotsset{compat=1.18}
\usepackage{tikz}
\usepackage{wrapfig}
\usepackage[ruled,vlined]{algorithm2e}
\usepackage{graphicx}
\usepackage{booktabs}
\usepackage{capt-of}
\usepackage{colortbl}
\usepackage{xcolor}

\usepackage{amssymb}
\usepackage{amsthm}
\newtheorem{definition}{Definition}
\newtheorem{theorem}{Theorem}
\newtheorem{lemma}{Lemma}
\newtheorem{proposition}{Proposition}
\newtheorem{corollary}{Corollary}
\newtheorem{assumption}{Assumption}
\theoremstyle{remark}
\newtheorem{remark}{Remark}

\usepackage{amsmath}
\usepackage{url}
\usepackage{hyperref}
\usepackage{comment}
\title{One-Step Generative Modeling via Unbalanced Optimal Transport}

\author{
Yirong Shen$^{1,2}$ \quad
Mengfei Xia$^2$\thanks{Corresponding author.} \quad
Junpeng Jing$^1$ \quad
Lu Gan$^3$ \quad
Cong Ling$^1$
\And
$^1$Imperial College London \quad
$^2$Ant Group \quad
$^3$Brunel University of London
}

\usepackage{float}
\usepackage{longtable}
\newcommand{\cls}{\mathsf{c}}
\iclrfinalcopy % Uncomment for camera-ready version, but NOT for submission.
\begin{document}

\maketitle
\fancyhead[L]{Preprint --- Work in Progress}

\begin{abstract}

Drifting models enable one-step generation by amortizing distribution transport into training, but this efficiency places greater demands on the transport field estimated at each update. In large-scale training, the field is computed from finite mini-batches of generated and real samples, which provide only imperfect approximations to the underlying distributions. Balanced optimal transport enforces exact mass matching within every mini-batch, making the estimated field sensitive to the particular composition of the real-data batch. We find that generated and real samples should be treated asymmetrically: letting the mass assigned to real samples adapt while keeping every generated sample fully transported improves generation across six feature-space metrics in controlled ablations, and is more robust to the relaxation strength than relaxing both marginals simultaneously, which falls below balanced transport under stronger relaxation. Motivated by this observation, we propose Unbalanced Optimal Transport Gradient Flow (UOT-GF), which keeps the generated-sample marginal fixed and relaxes only the real-data marginal.  Under identical settings at DiT-B/2 on ImageNet-256, UOT-GF improves Fréchet Inception Distance (FID) from 1.53 to 1.46 over the balanced W-Flow baseline; scaling the same recipe yields 1.34 and 1.22 FID at L/2 and XL/2, the best FID among the one-step models we compare. We further derive the induced UOT transport force, establish a kinetic Vlasov--Fokker--Planck formulation whose overdamped zero-temperature limit recovers the drifting dynamics, and characterize non-target stationary states together with sufficient conditions for convergence.

\end{abstract}
\section{Introduction}

\begin{figure}[t]
    \centering
    \includegraphics[width=\linewidth]{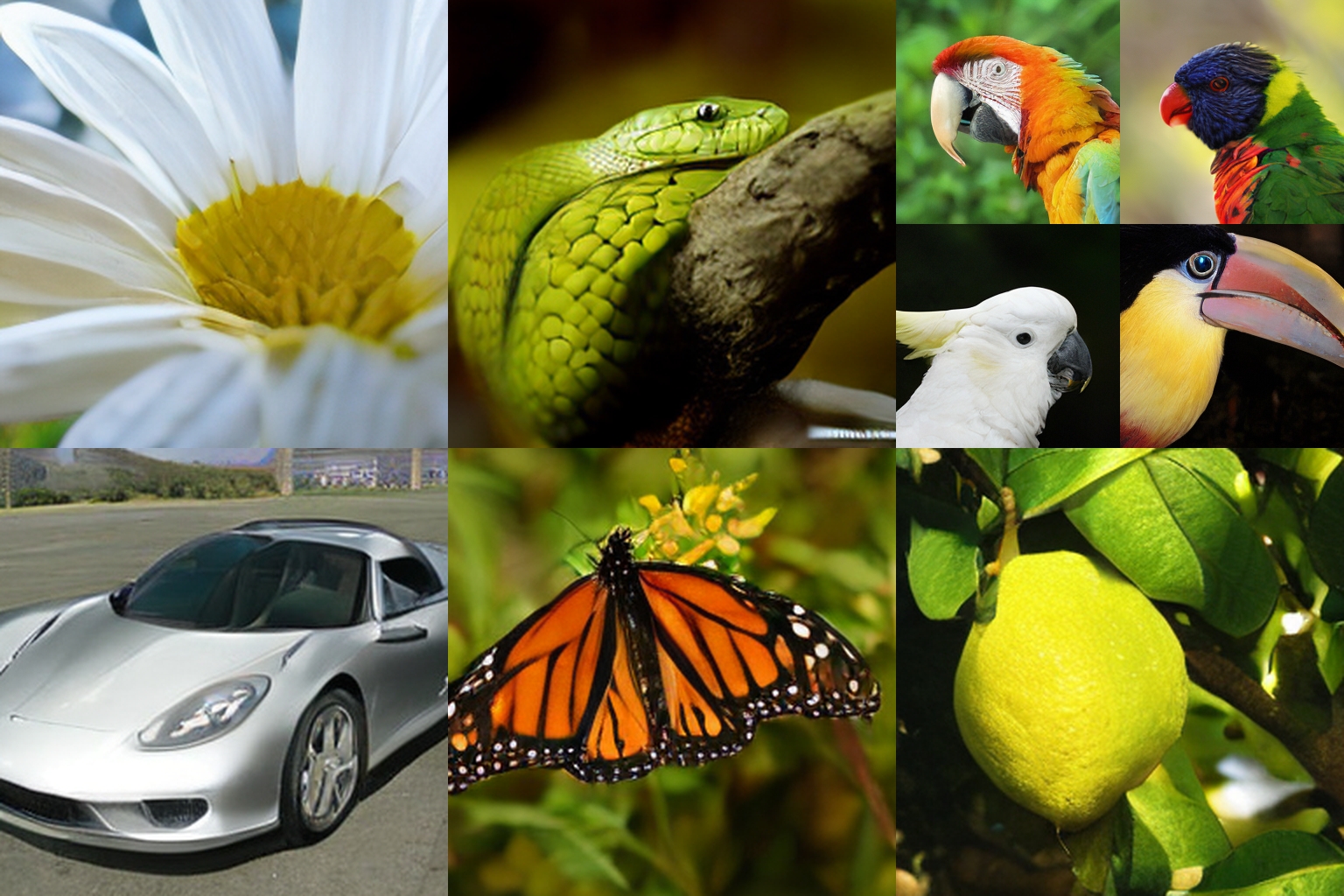}
    \caption{Selected 1-NFE samples from UOT-GF-XL/2 trained from scratch on
    ImageNet-256 (guidance scale $1.14$).}
    \label{fig:samples}
\end{figure}

Generative modeling seeks to transform samples from a simple reference
distribution into samples from the data distribution. Diffusion and flow-based
models achieve high generation quality but often require repeated network
evaluations at inference time
\citep{ddpm,song2021scorebased,lipman2023flow,liu2023flow}. GANs reduce sampling
to a single forward pass, but introduce adversarial optimization and its associated
training difficulties \citep{GAN,styleGAN}. Autoregressive and bridge-based
approaches offer further alternatives, with different trade-offs between generation
speed, modeling flexibility, and sample quality
\citep{li2024autoregressive,var,shi2023diffusion,bortoli2021diffusion,
pariset2023unbalanced}. Drifting models use a different strategy: rather than
integrating a long trajectory at test time, they learn how generated
particles should move toward the data distribution during training, enabling one-step sampling
\citep{deng2026generative}. W-Flow \citep{han2026one} interprets these dynamics as
a Wasserstein gradient flow induced by the balanced Sinkhorn divergence and reports
strong one-step performance on ImageNet-256.

The practical accuracy of such a drifting update depends on how well its transport
direction can be estimated during training. At each iteration, only a limited set of
generated and real examples participates in the transport problem
\citep{han2026one}. For large-scale datasets with substantial intra-class or local
variation, the samples present in one batch need not reflect the relative prevalence
of the underlying data patterns: some may be over-represented, while others may be
missing altogether. ImageNet provides a representative example, where individual
classes can exhibit considerable variation in appearance, pose, viewpoint, and
background \citep{imagenet,russakovsky2015imagenet,3524938.3525830}. Balanced transport still
assigns a prescribed amount of mass to every sample in the batch. Consequently,
the resulting coupling may be influenced by incidental batch composition \citep{pmlr-v139-fatras21a} rather than only by the population-level structure one ultimately wishes to learn. This raises a natural question: should the generated
and real samples be subject to the same marginal constraints in drifting-based
training?

To address this issue, we propose Unbalanced Optimal Transport Gradient Flow (UOT-GF). Unlike semi-dual,
adversarial, or JKO-type uses of unbalanced transport
\citep{NEURIPS2020_9719a00e,3666122.3667962,3692070.3692414}, it enters as the
geometry of an explicit, simulation-free drifting update, making the marginal
relaxation a design axis: it defines a family of transport energies, each inducing
its own flow, with balanced Sinkhorn drifting recovered as the $\tau\rightarrow1$
member. 
Our central finding is that \emph{which marginal is relaxed matters}: relaxing the real-data marginal alone is at least as effective as relaxing both, and is more robust to the relaxation strength---at stronger relaxation, relaxing both marginals degrades to worse than balanced transport while symmetrized source-fixed transport still improves on it. Each directed transport problem fixes its source marginal and relaxes its target marginal, and the full interaction averages contributions from both directions. Generated and real samples therefore each appear once on the fixed side and once on the relaxed side. The effective cross-interaction plan permits marginal reweighting while retaining at least half of each distribution's reference mass, as shown in Section~\ref{sec:energy}. This structural property does not by itself guarantee target coverage, which we evaluate empirically in Section~\ref{sec:exp}.% Our central finding is that how the marginals are relaxed matters more than
% whether: the standard both-sided relaxation, which the added freedom should favor,
% performs worse than balanced transport, while fixing the source and relaxing only the
% target improves on it consistently. The marginals are not interchangeable. The source
% is over generated particles the flow must keep moving; holding it exact keeps every
% particle under full-weight attraction, whereas relaxing it kills that force precisely
% on the particles furthest from the data. The target is an empirical batch of
% unreliable proportions; relaxing it only moves the barycentric target. This may
% appear to invite mode dropping, but every generated particle is still assigned in
% full at every step, and recall is unchanged from the balanced baseline
% (Section~\ref{sec:exp}). 
Our contributions are as follows:

\begin{itemize}\itemsep2pt
\item We study marginal relaxation in gradient-flow drifting and show empirically that symmetrized source-fixed transport improves generation over balanced transport across six feature-space metrics, and is markedly more robust to the relaxation strength than relaxing both marginals, which falls below balanced transport under stronger relaxation.

\item We develop a unified kinetic theory connecting stochastic dynamics to deterministic drifting, with balanced Sinkhorn drifting as a special case. We establish free-energy dissipation, reveal non-target stationary states inherent to the drifting paradigm, and give sufficient conditions for reaching the target.

\item On ImageNet-256, relaxation improves balanced transport \citep{han2026one} at B/2 under identical settings (1.53 $\rightarrow$ 1.46 FID), and scaling one fixed recipe yields 1.22 FID at XL/2, a new one-step state of the art among models trained from scratch (Figure~\ref{fig:benchmark_samples}). %On ImageNet-256, UOT-GF achieves one-step FIDs of $1.22$, $1.34$, and $1.46$ at XL, L, and B scales, respectively, outperforming balanced Sinkhorn drifting \citep{han2026one} at every scale and setting a new state of the art among one-step generators (Figure~\ref{fig:benchmark_samples}).
\end{itemize}

% Overall, our contributions are summarized as follows:
% \begin{itemize}
%     \item We establish a kinetic variational foundation for UOT-GF by constructing nonlinear VFP dynamics, proving free-energy
%     dissipation, and recovering Drifting model through an
%     $\mathcal{O}(\sqrt{m})$ overdamped limit followed by a zero-temperature
%     limit.
%     \item We characterize the long-time behavior of UOT-GF by identifying non-target stationary states that obstruct global convergence and proving conditional $\mathcal{O}(t^{-1})$ energy decay together with Wasserstein distance $\mathcal W_2$ convergence.
%     \item Extensive experiments demonstrate that UOT-GF achieves state-of-the-art one-step generation performance on ImageNet-256, obtaining 1.22 FID at XL scale, 1.34 at L scale and 1.46 at B scale, significantly improving the Drifting model.
% \end{itemize}

\begin{figure}[t]
    \centering
    \captionsetup{font=small, skip=2pt}
    \includegraphics[width=\linewidth]{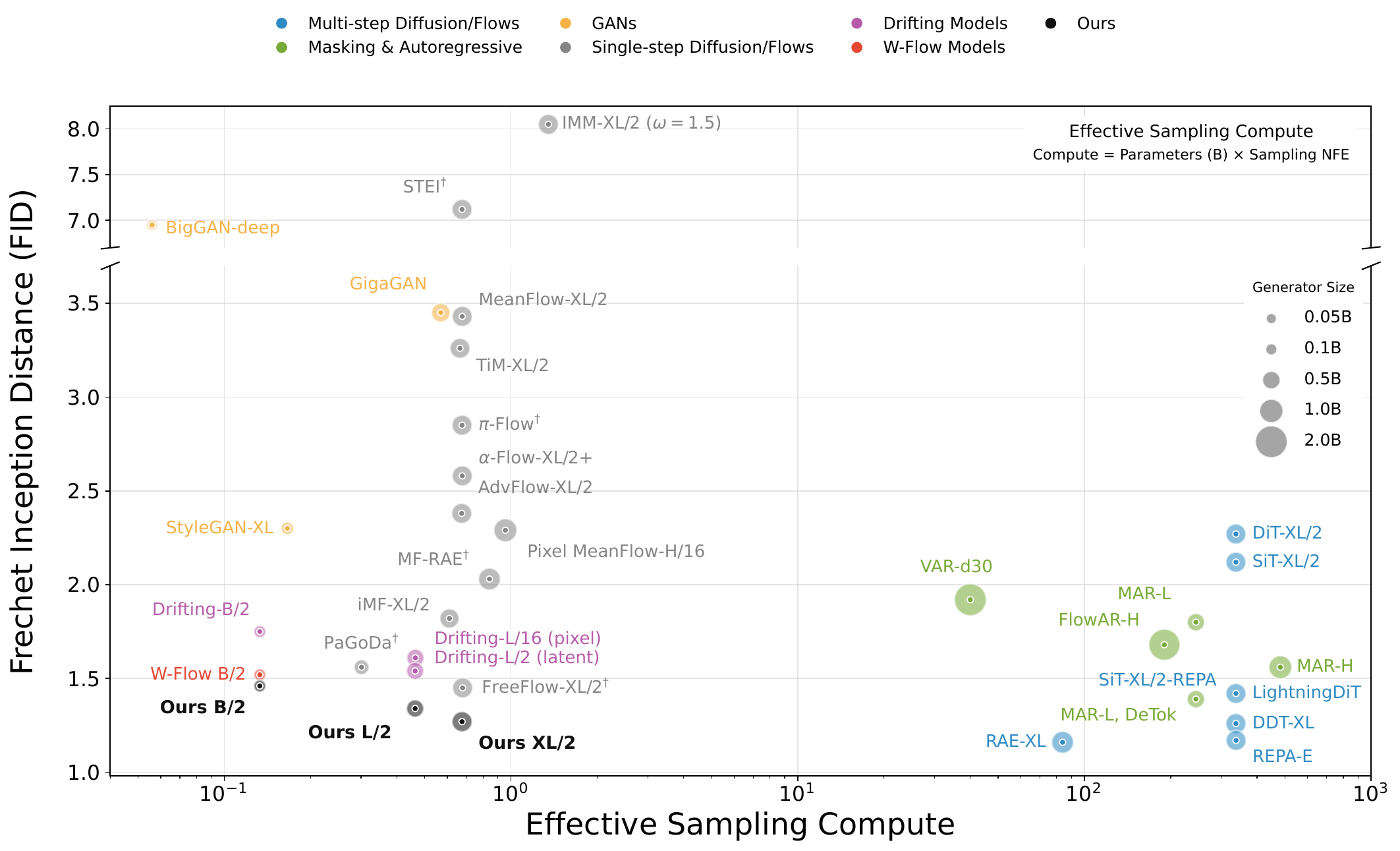}
    \caption{Fréchet Inception Distance (FID) versus effective sampling compute (parameters $\times$ NFE \citep{3618408.3619743}) on ImageNet-256. 1-NFE samples are shown in Appendix~\ref{app:samples}.}
    \label{fig:benchmark_samples}
    \vspace{-4mm}
\end{figure}

\begin{figure}[t]
    \centering
    \begin{subfigure}[b]{0.48\linewidth}
        \centering
        \includegraphics[width=\linewidth]{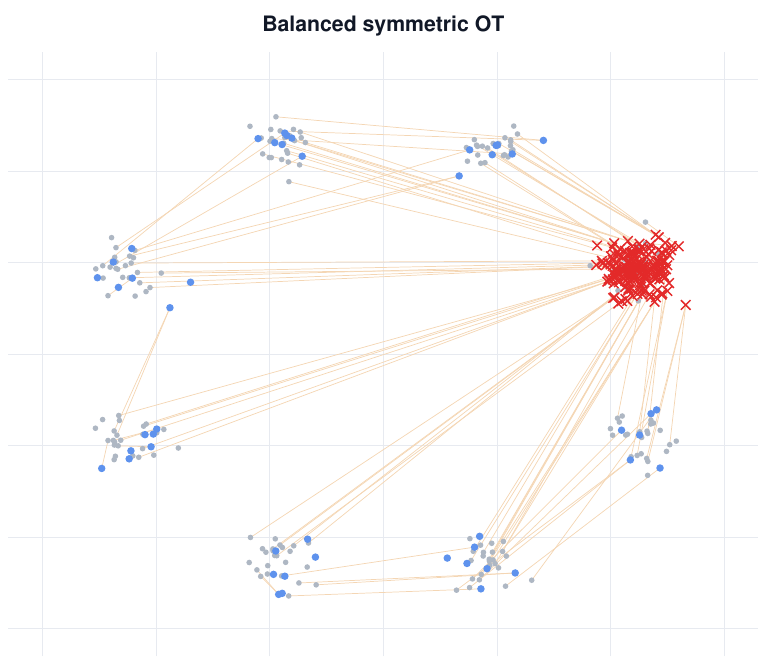}
        \caption{Balanced OT}
        \label{fig:uot_toy_a}
    \end{subfigure}\hfill
    \begin{subfigure}[b]{0.48\linewidth}
        \centering
        \includegraphics[width=\linewidth]{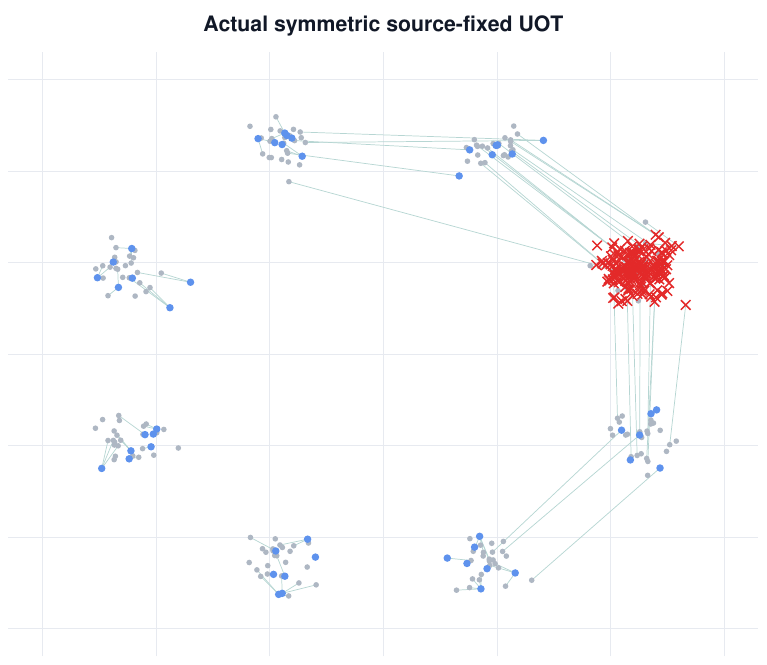}
        \caption{Symmetrized source-fixed UOT}
        \label{fig:uot_toy_b}
    \end{subfigure}
    \vspace{-6pt}
    \caption{\textbf{Oversampled mode in a minibatch.} One mode supplies 80\% of the targets (red crosses). Balanced OT (a) must deliver full mass to them, pulling sources from distant modes; symmetrized source-fixed UOT (b) reweights them and keeps most couplings local. Gray: sources; blue: remaining targets. More examples in Appendix~\ref{app:toy}.}
    \label{fig:uot_toy_examples}
    \vspace{-0.5cm}
\end{figure}

% \begin{figure*}[t]
%     \centering

%     % First row
%     \includegraphics[width=0.495\linewidth]
%     {figs/01_wrong_manifold_balanced_vs_uot.pdf}%
%     \hfill
%     \includegraphics[width=0.495\linewidth]
%     {figs/02_false_island_balanced_vs_uot.pdf}

%     % Second row
%     \includegraphics[width=0.495\linewidth]
%     {figs/03_cluster_proportion_mismatch.pdf}%
%     \hfill
%     \includegraphics[width=0.495\linewidth]
%     {figs/04_minibatch_duplication_burst.pdf}
%     \vspace{-10pt}
%     \caption{
% \textbf{Source-fixed UOT under target perturbations.}
% Balanced OT and symmetrized source-fixed UOT are compared on
% wrong-manifold, spurious-island, mode-imbalance, and minibatch-duplication examples. Gray dots denote sources; blue dots denote nominal targets; teal crosses mark imbalanced targets, and red crosses mark
% corrupted or duplicated targets. Each UOT solve fixes its source marginal, while the combined interaction permits reweighting of both effective marginals.
% }
% %     \caption{Balanced (left of each pair) versus source-fixed unbalanced transport (right)
% % on synthetic targets. Top left: target mass on a wrong manifold. Top right: a spurious
% % target island. Bottom left: equal-weight source against an 80/10/10 target. Bottom
% % right: duplicated mini-batch samples. In each case the hard marginal constraint
% % forces couplings that relaxing the target marginal avoids.}
%     \label{fig:uot_toy_examples}
%     \vspace{-1mm}
% \end{figure*}

\section{Preliminary}

Generative modeling can be formulated as learning a transport map
$f:\mathbb R^{d_0}\to\mathbb R^d$ that pushes a simple reference distribution
$p_{\mathrm{ref}}\in\mathcal P(\mathbb R^{d_0})$ toward the data distribution
$p:=p_{\mathrm{data}}\in\mathcal P(\mathbb R^d)$. Writing $f_\#$ for the pushforward
operator, the goal is to learn $f$ such that $f_\#p_{\mathrm{ref}}\approx p$.
Generative models parameterize $f$ by a neural network
$f_{\boldsymbol\theta}:\mathbb R^{d_0}\to\mathbb R^d$ and define the induced model
distribution $q_{\boldsymbol\theta}:=(f_{\boldsymbol\theta})_\#p_{\mathrm{ref}}$, so
that sampling $\mathbf z\sim p_{\mathrm{ref}}$ and computing
$\mathbf x=f_{\boldsymbol\theta}(\mathbf z)$ yields $\mathbf x\sim q_{\boldsymbol\theta}$.
Training seeks to minimize the discrepancy between $q_{\boldsymbol\theta}$ and $p$.

\begin{figure}[t]
    \centering
    \includegraphics[width=\linewidth]{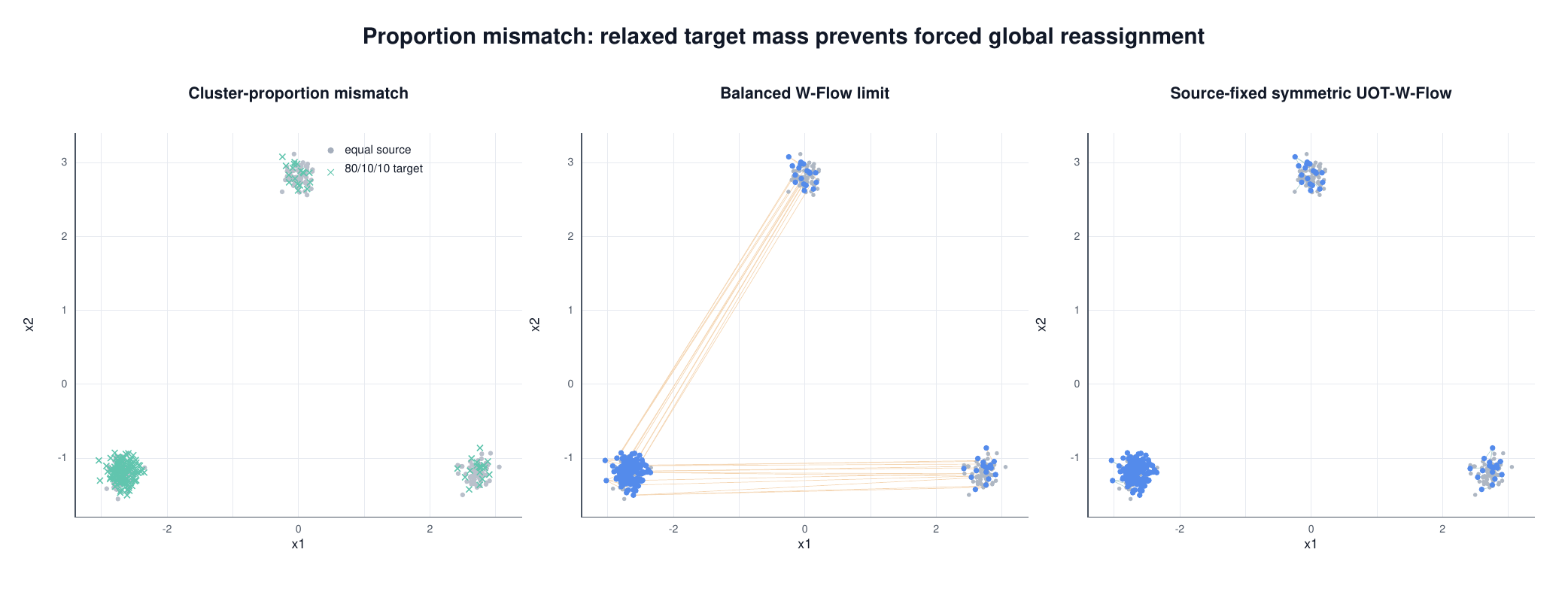}
    \caption{\textbf{Mode imbalance.} Equal-weight sources (gray dots) against an
    80/10/10 target (teal crosses), left. Balanced OT (middle) reassigns source mass
    across clusters to match the proportions; symmetrized source-fixed UOT (right)
    keeps sources on nearby clusters.}
    \label{fig:toy_imbalance}
\end{figure}

\subsection{Wasserstein Gradient Flow}\label{sec:wgf}

Beyond the static pushforward formulation above, a generative model can be viewed at the distributional level as a dynamical system that continuously evolves the model distribution toward $p$. Let $\mathcal{F}:\mathcal{P}(\mathbb{R}^d)\to\mathbb{R}$ be an energy functional and write $\frac{\delta\mathcal{F}}{\delta q}(q)$ for its first variation at $q$ (Definition~\ref{def:first_variation}), and let $\mathbf v_q := -\nabla\frac{\delta\mathcal F}{\delta q}(q)$ denote the induced velocity field. A curve of densities $\{q_t\}_{t\ge0}$ is a \emph{Wasserstein gradient flow}
(WGF) of $\mathcal F$ if it solves the continuity equation
\begin{equation}\label{WGF}
    \partial_t q_t = -\nabla\cdot(q_t\,\mathbf v_{q_t}),
\end{equation}
that is, mass is transported along the steepest-descent direction of $\mathcal{F}$ in the $2$-Wasserstein geometry. Along any sufficiently regular solution the energy dissipates as
\begin{equation}\label{dissipation}
    \frac{d}{dt}\mathcal{F}[q_t] = -\mathbb{E}_{q_t}\big[\|\mathbf v_{q_t}\|^2\big] \le 0.
\end{equation}
To drive the flow toward $p$, one takes $\mathcal{F}[q] = \mathsf{D}(q\,\|\,p)$ for a nonnegative, definite divergence $\mathsf{D}$, i.e., $\mathsf{D}(q\|p)\ge0$ with equality exactly when $q=p$, so that $p$ is the unique global minimizer and $\mathbf{v}_p = \mathbf{0}$. \Eqref{dissipation} guarantees that
$\mathcal{F}[q_t]$ is non-increasing, but not that it reaches zero: Wasserstein stationarity requires $\mathbf v_{q_t}=\mathbf 0$ only on $\mathrm{supp}(q_t)$, and non-target stationary states can exist. We return to this in Section~\ref{sec:longtime}.

The framework recovers familiar dynamics. For $\mathcal{F} = \mathrm{KL}(\cdot\,\|\,p)$, \eqref{WGF} is the Fokker--Planck equation $\partial_t q_t = -\nabla\cdot(q_t\nabla\log p) + \Delta q_t$, whose drift is the score of $p$; score-based models learn the scores of a time-indexed family of perturbed targets and integrate this equation numerically at inference time \citep{NEURIPS2019_3001ef25,ddpm,song2021scorebased,3618408.3619743}.

\subsection{Drifting Models}

Unlike score-based diffusion and ODE-based flow models, which numerically integrate
the sample dynamics over multiple steps at inference time, drifting models amortize
this transport into training: the generator $f_{\boldsymbol\theta}$ is updated so that
$q_{\boldsymbol\theta}=(f_{\boldsymbol\theta})_{\#}p_{\mathrm{ref}}$ approaches $p$, and
sampling afterwards requires a single evaluation of $f_{\hat{\boldsymbol\theta}}$.

% Let $f_i := f_{\boldsymbol\theta_i}$ and $q_i := (f_i)_{\#}p_{\mathrm{ref}}$ denote the
% generator and its induced distribution at iteration $i$. For a fixed
% $\boldsymbol z\sim p_{\mathrm{ref}}$, updating $\boldsymbol\theta_i$ to
% $\boldsymbol\theta_{i+1}$ displaces the generated sample by
% $\Delta\boldsymbol x_i := f_{i+1}(\boldsymbol z)-f_i(\boldsymbol z)$. Drifting models
% train this displacement to follow a prescribed velocity field $\mathbf v_{q_i}$ that
% depends on the current model distribution:
% \begin{equation}\label{eq:drift-update}
%     f_{i+1}(\boldsymbol z) \approx f_i(\boldsymbol z)
%     + h_i\,\mathbf v_{q_i}\big(f_i(\boldsymbol z)\big),
% \end{equation}
% with step size $h_i>0$. 

Let $f_n:=f_{\boldsymbol\theta_n}$ and $q_n:=(f_n)_\#p_{\mathrm{ref}}$ denote the
generator and its induced distribution at iteration $n$. For a fixed
$\mathbf z\sim p_{\mathrm{ref}}$, updating $\boldsymbol\theta_n$ to
$\boldsymbol\theta_{n+1}$ displaces the generated sample by
$\Delta\mathbf x_n:=f_{n+1}(\mathbf z)-f_n(\mathbf z)$. Drifting models train this
displacement to follow a prescribed velocity field $\mathbf v_{q_n}$ that depends on
the current model distribution:
\begin{equation}\label{eq:drift-update}
    f_{n+1}(\mathbf z)\approx f_n(\mathbf z)+h_n\,\mathbf v_{q_n}\big(f_n(\mathbf z)\big),
\end{equation}
with step size $h_n>0$. If the regression target is realized exactly,
\eqref{eq:drift-update} is an explicit Euler step of the mean-field dynamics
$\dot{\boldsymbol x}_t=\mathbf v_{q_t}(\boldsymbol x_t)$, and when $\mathbf v_q$ is the
WGF velocity \eqref{WGF} of an energy $\mathcal F$, the induced $q_t$ follows the
corresponding WGF. The idealized update is stationary
at an equilibrium $\mathbf v_{q_{\hat{\boldsymbol\theta}}}=\mathbf 0$ on
$\mathrm{supp}(q_{\hat{\boldsymbol\theta}})$, i.e., $f_{\hat{\boldsymbol\theta}}$ is a
fixed point of \eqref{eq:drift-update}; by Section~\ref{sec:longtime}, such an equilibrium
need not satisfy $q_{\hat{\boldsymbol\theta}}=p$. In practice $\mathbf v_{q_n}$ is estimated from
mini-batches of generated and real samples (Section~\ref{sec:impl}), so the field
actually followed is a finite-sample surrogate of \eqref{WGF}.

The choice of $\mathbf v_q$ is what distinguishes existing instances. The original
Drifting model \citep{deng2026generative} prescribes it heuristically through kernel-based
attraction toward $p$ and repulsion from $q$. W-Flow \citep{han2026one} instead takes $\mathcal F$ to be the debiased entropic Sinkhorn divergence $S_\varepsilon(q,p)$
\citep{Feydy2018InterpolatingBO}, whose WGF velocity is the difference of two barycentric
projections,
\begin{equation}\label{eq:wflow-velocity}
    \mathbf v_q(\mathbf x)
    = B_{q,p}(\mathbf x) - B_{q,q}(\mathbf x),
    \qquad
    B_{q,\alpha}(\mathbf x) := \mathbb E_{\pi^{q,\alpha}}[\mathbf y\mid\mathbf x],
\end{equation}
where $\pi^{q,\alpha}$ is the optimal plan of \emph{balanced} entropic transport between $q$ and $\alpha$, in which both marginals are matched exactly. This is the field we generalize: Section~\ref{sec:energy} replaces the balanced energy by an unbalanced one and derives the resulting force, and Section~\ref{sec:kinetic} recovers the mean-field dynamics underlying \eqref{eq:drift-update} as the overdamped, zero-temperature limit of a kinetic lifting.

\section{Theoretical analysis}

Several recent works have interpreted drifting models through Wasserstein gradient
flows \citep{he2026sinkhorn,caucheteux2026a,cao2026gradient,
gretton2026wasserstein,dumont2026learning,han2026one}. In particular, W-Flow
\citep{han2026one} formulates drifting as the gradient flow of the balanced
Sinkhorn divergence. Our method instead relaxes only the real-data marginal while
keeping the generated-sample marginal fixed. Because this asymmetric transport
changes both the energy and the induced force field, the existing balanced-flow
analysis does not directly apply.

To accommodate this setting, we first derive the transport force induced by a
generalized UOT energy and show that balanced W-Flow is recovered in the
hard-constraint limit (Section~\ref{sec:energy}). We then develop a nonlinear
Vlasov--Fokker--Planck formulation, establish free-energy dissipation, and recover
the first-order drifting dynamics through overdamped and zero-temperature limits
(Section~\ref{sec:kinetic}). Finally, we analyze long-time behavior, identifying
possible non-target stationary states and giving sufficient conditions for
convergence to the target distribution (Section~\ref{sec:longtime}).

\subsection{UOT--Sinkhorn Energy and Its First Variation}
\label{sec:energy}

Let $p,q\in\mathcal P(\Omega)$, $\Omega\subset\mathbb R^d$, and $c(\mathbf x,\mathbf y)=\tfrac12\|\mathbf x-\mathbf y\|^2$. A marginal penalty is generated either by the convex indicator $\iota$ ($\iota(1)=0$, $+\infty$ otherwise), which enforces a marginal exactly, or by a strictly convex $\varphi\in C^2((0,\infty))$ with $\varphi(1)=0$. 
For generators $(\varphi_1,\varphi_2)$ and $\rho,\varepsilon>0$,
\begin{equation}\label{eq:uot}
\mathrm{UOT}^{\varphi_1,\varphi_2}(\alpha_1,\alpha_2)
:=\inf_{\pi\in\mathcal M_+(\Omega^2)}
\int c\,d\pi+\rho D_{\varphi_1}(\pi_1\Vert\alpha_1)
+\rho D_{\varphi_2}(\pi_2\Vert\alpha_2)
+\varepsilon\,\mathrm{KL}(\pi\Vert\alpha_1\otimes\alpha_2),
\end{equation}
with $D_\varphi$ the Csisz\'ar divergence (Definition~\ref{def:phi_divergence}) and
$\tau:=\rho/(\rho+\varepsilon)$ the relaxation strength ($\tau\to1$ is the hard
constraint). 
Three members matter: balanced $(\iota,\iota)$, underlying W-Flow
\citep{han2026one}; symmetric $(\varphi,\varphi)$ \citep{Sjourn2019SinkhornDF}; and
source-fixed $(\iota,\mathrm{KL})$, used here (Figure~\ref{fig:uot_toy_examples}).
Let
$\mathcal A(\alpha_1,\alpha_2):=\tfrac12[\mathrm{UOT}^{\varphi_1,\varphi_2}(\alpha_1,\alpha_2)
+\mathrm{UOT}^{\varphi_1,\varphi_2}(\alpha_2,\alpha_1)]$ and define the
configurational energy
$\mathcal E_p(q):=\mathcal A(q,p)-\tfrac12\mathcal A(q,q)-\tfrac12\mathcal A(p,p)$
(Definition~\ref{def:configurational_energy}), which coincides on $\mathcal P(\Omega)$
with the unbalanced Sinkhorn divergence of \citet{Sjourn2019SinkhornDF} for
$(\varphi,\varphi)$ and with $S_\varepsilon(q,p)$
\citep{Feydy2018InterpolatingBO,han2026one} for $(\iota,\iota)$. Let
$\Phi_q:=\frac{\delta\mathcal E_p}{\delta q}(q)$ and $F_p[q]:=-\nabla\Phi_q$. By
construction $\mathcal E_p(p)=0$ and $F_p[p]=\mathbf 0$ for every choice of
generators.\footnote{The asymmetric divergence of
\citet[\S4.7.1]{Sjourn2019SinkhornDF} debiases with symmetric self-terms, which gives
positivity but $F_p[p]\neq\mathbf 0$. A drifting energy needs the target stationary,
so we debias with the same asymmetric cost.}

% Let $\mathcal A(\alpha_1,\alpha_2):=\tfrac12[\mathrm{UOT}^{\varphi_1,\varphi_2}(\alpha_1,\alpha_2) +\mathrm{UOT}^{\varphi_1,\varphi_2}(\alpha_2,\alpha_1)]$ and which coincides on $\mathcal P(\Omega)$ with the unbalanced Sinkhorn divergence of \citet{Sjourn2019SinkhornDF} for $(\varphi,\varphi)$ and with $S_\varepsilon(q,p)$ \citep{Feydy2018InterpolatingBO,han2026one} for $(\iota,\iota)$. Let $\mathcal E_p(q) :=\mathcal A(q,p)-\tfrac12\mathcal A(q,q)-\tfrac12\mathcal A(p,p)$, $\Phi_q:=\frac{\delta\mathcal E_p}{\delta q}(q)$ and $F_p[q]:=-\nabla\Phi_q$. By construction $\mathcal E_p(p)=0$ and $F_p[p]=\mathbf 0$ for every choice of generators.\footnote{The asymmetric divergence of \citet[\S4.7.1]{Sjourn2019SinkhornDF} debiases with symmetric self-terms, which gives
% positivity but $F_p[p]\neq\mathbf 0$. A drifting energy needs the target stationary,
% so we debias with the same asymmetric cost.}
\begin{proposition}[Force]\label{prop:uot_force}
Under Assumption~\ref{ass:uot_regularity}, assume that all optimal
plans appearing below are unique.
For a pair $(\alpha_1,\alpha_2)$ with optimal plan $\pi$ of
$\mathrm{UOT}^{\varphi_1,\varphi_2}(\alpha_1,\alpha_2)$ and a slot
$k\in\{1,2\}$, let
$r^{(k)}_{\alpha_1,\alpha_2}:=d\pi_k/d\alpha_k$ and
$B^{(k)}_{\alpha_1,\alpha_2}(\mathbf x)
:=\mathbb E_\pi[\mathbf x_{-k}\mid\mathbf x_k=\mathbf x]$,
with $r^{(k)}\equiv1$ if $\varphi_k=\iota$ and $r^{(k)}$
positive and $C^1$ otherwise, and define the force on the
$k$-th marginal
\begin{equation}\label{eq:slot_force}
F^{(k)}_{\alpha_1,\alpha_2}(\mathbf x)
:=
r^{(k)}_{\alpha_1,\alpha_2}(\mathbf x)
\bigl(B^{(k)}_{\alpha_1,\alpha_2}(\mathbf x)-\mathbf x\bigr).
\end{equation}
Then $F^{(k)}_{\alpha_1,\alpha_2}$ is the negative spatial gradient
of the first variation of
$\mathrm{UOT}^{\varphi_1,\varphi_2}(\alpha_1,\alpha_2)$
with respect to $\alpha_k$, and
\begin{equation}\label{eq:force_general}
F_p[q]
=
\tfrac12\bigl[F^{(1)}_{q,p}+F^{(2)}_{p,q}\bigr]
-\tfrac12\bigl[F^{(1)}_{q,q}+F^{(2)}_{q,q}\bigr].
\end{equation}
For $(\varphi,\varphi)$ this is
$r_{q,p}(B_{q,p}-\mathbf x)-r_{q,q}(B_{q,q}-\mathbf x)$;
for $(\iota,\iota)$ it is $B_{q,p}-B_{q,q}$, the W-Flow field
\eqref{eq:wflow-velocity}; for $(\iota,\mathrm{KL})$,
$r^{(1)}\equiv1$, $r^{(2)}>0$, and
$\int_\Omega r^{(2)}\,d\alpha_2=1$.
All force identities hold almost everywhere with respect to
the corresponding input measure.
\end{proposition}
% \begin{proposition}[Force]\label{prop:uot_force}
% For a pair $(\alpha_1,\alpha_2)$ with unique optimal plan $\pi$ of
% $\mathrm{UOT}^{\varphi_1,\varphi_2}(\alpha_1,\alpha_2)$ and a slot $k\in\{1,2\}$, let
% $r^{(k)}_{\alpha_1,\alpha_2}:=d\pi_k/d\alpha_k$ and
% $B^{(k)}_{\alpha_1,\alpha_2}(\mathbf x):=\mathbb E_\pi[\mathbf x_{-k}\mid\mathbf x_k=\mathbf x]$,
% with $r^{(k)}\equiv1$ if $\varphi_k=\iota$ and $r^{(k)}$ positive and $C^1$ otherwise,
% and define the force on the $k$-th marginal
% \begin{equation}\label{eq:slot_force}
% F^{(k)}_{\alpha_1,\alpha_2}(\mathbf x)
% :=r^{(k)}_{\alpha_1,\alpha_2}(\mathbf x)\big(B^{(k)}_{\alpha_1,\alpha_2}(\mathbf x)-\mathbf x\big).
% \end{equation}
% Then $F^{(k)}_{\alpha_1,\alpha_2}$ is the negative gradient of the first variation of
% $\mathrm{UOT}^{\varphi_1,\varphi_2}(\alpha_1,\alpha_2)$ with respect to $\alpha_k$, and
% \begin{equation}\label{eq:force_general}
% F_p[q]
% =\tfrac12\big[F^{(1)}_{q,p}+F^{(2)}_{p,q}\big]
% -\tfrac12\big[F^{(1)}_{q,q}+F^{(2)}_{q,q}\big].
% \end{equation}
% For $(\varphi,\varphi)$ this is $r_{q,p}(B_{q,p}-\mathbf x)-r_{q,q}(B_{q,q}-\mathbf x)$;
% for $(\iota,\iota)$ it is $B_{q,p}-B_{q,q}$, the W-Flow field
% \eqref{eq:wflow-velocity}; for $(\iota,\mathrm{KL})$, $r^{(1)}\equiv1$ and
% $r^{(2)}\le1$.
% \end{proposition}

Averaging the forward and transposed reverse source-fixed plans gives an effective coupling $\bar\pi$ with $\bar\pi_1\ge q/2$ and $\bar\pi_2\ge p/2$: marginals may be reweighted but keep at least half their mass, with no guarantee of nonzero force or target coverage (Appendix~\ref{app:velocity}).
On ImageNet, forward-plan reweighting is mild (0.92--1.03) yet stable
(Figure~\ref{fig:source_fixed_uot_mass}).%In the source-fixed case the slot where $q$ is the source carries unit weight, so every particle receives the full attraction $B^{(1)}_{q,p}-\mathbf x$, while the slot where $q$ is the relaxed target may be attenuated where the plan declines mass.

\subsection{Kinetic Lifting and the Drifting Flow}
\label{sec:kinetic}

Proposition~\ref{prop:uot_force} turns the static transport problem into a
self-consistent force field. We lift the positional dynamics to phase space, taking
$\Omega=\mathbb R^d$ here, and then eliminate momentum to recover the first-order flow
of Section~\ref{sec:wgf}. Let $\mu_t^{m,\beta}(\mathbf x,\mathbf u)$ be the joint
density of position and momentum with positional marginal
$q_t^{m,\beta}:=\int\mu_t^{m,\beta}\,d\mathbf u$. For inertia $m>0$, friction
$\gamma>0$, and inverse temperature $\beta>0$, consider the nonlinear
Vlasov--Fokker--Planck equation
\begin{equation}\label{eq:uot_vfp}
\partial_t\mu_t^{m,\beta}
+\nabla_{\mathbf x}\cdot\Big(\frac{\mathbf u}{m}\mu_t^{m,\beta}\Big)
+\nabla_{\mathbf u}\cdot\Big(F_p[q_t^{m,\beta}]\,\mu_t^{m,\beta}\Big)
=\frac{\gamma}{m}\nabla_{\mathbf u}\cdot\big(\mathbf u\,\mu_t^{m,\beta}\big)
+\frac{\gamma}{\beta}\Delta_{\mathbf u}\mu_t^{m,\beta},
\end{equation}
whose left-hand side is conservative transport driven by the UOT force and whose
right-hand side is friction and thermal diffusion in momentum. The lifting is
variational: with $q_\mu:=\int\mu\,d\mathbf u$, the free energy
$\mathcal H_{m,\beta}(\mu):=\mathcal E_p(q_\mu)+\int\frac{\|\mathbf u\|^2}{2m}\mu
+\frac1\beta\int\mu\log\mu$ is a Lyapunov functional.

\begin{proposition}[Free-energy dissipation]\label{prop:vfp_dissipation}
Under Assumption~\ref{ass:kinetic}, every sufficiently regular solution of
\eqref{eq:uot_vfp} satisfies
$\frac{d}{dt}\mathcal H_{m,\beta}(\mu_t^{m,\beta})
=-\gamma\int\mu_t^{m,\beta}\big\|\frac{\mathbf u}{m}
+\frac1\beta\nabla_{\mathbf u}\log\mu_t^{m,\beta}\big\|^2\le0$.
\end{proposition}

\Eqref{eq:uot_vfp} is thus a kinetic variational lifting of $\mathcal E_p$
rather than a gradient flow; the latter emerges as momentum relaxes.

\begin{theorem}[Quantitative Kramers--Smoluchowski limit]
\label{thm:overdamped}
Fix $\gamma,\beta,T>0$ and assume Assumption~\ref{ass:overdamped}. Suppose that the
initial laws share a positional marginal $q_0\in\mathcal P_2(\mathbb R^d)$ and satisfy
$\int\|\mathbf u\|^2\,d\mu_0^{m,\beta}\le C_0m$ uniformly for $0<m\le1$. Then
$\sup_{t\in[0,T]}\mathcal W_2(q_t^{m,\beta},q_t^\beta)\le C_{T,\beta}\sqrt m$, where
$C_{T,\beta}$ is independent of $m$ and
$q^\beta\in C([0,T];\mathcal P_2(\mathbb R^d))$ is the unique weak solution of
\begin{equation}\label{eq:finite_temp}
\partial_t q_t^\beta
=\frac1\gamma\nabla\cdot\bigl(q_t^\beta\nabla\Phi_{q_t^\beta}\bigr)
+\frac1{\gamma\beta}\Delta q_t^\beta,
\qquad q_0^\beta=q_0.
\end{equation}
\end{theorem}

\Eqref{eq:finite_temp} is the WGF \eqref{WGF} of
$\mathcal E_p+\beta^{-1}\int q\log q\,d\mathbf x$
with mobility $1/\gamma$.
As $\beta\to\infty$, Corollary~\ref{cor:zero_temp} gives convergence, uniformly on finite time intervals in $\mathcal W_2$, to
\begin{equation}\label{eq:zero_temp}
\partial_t q_t
=
-\frac1\gamma\nabla\cdot\bigl(q_tF_p[q_t]\bigr),
\qquad q_{t=0}=q_0.
\end{equation}
This is \eqref{WGF} with $\mathcal F=\mathcal E_p$ and
$\mathbf v_q=\gamma^{-1}F_p[q]$.
Its characteristic Euler update is
$\mathbf x^{n+1}=\mathbf x^n+(h/\gamma)F_p[q_n](\mathbf x^n)$, with the mobility
absorbed into the drifting step size $\eta=h/\gamma$ (cf.~\eqref{eq:drift-update}).
For $(\iota,\iota)$, the force is the W-Flow field.

\begin{remark}\label{rem:diffusion}
For finite $\beta$, diffusion yields an absolutely continuous
law at every positive time, whereas \eqref{eq:zero_temp}
transports measures without diffusion.
On $\mathbb R^d$, the zero-temperature stationary states
of Section~\ref{sec:longtime} do not remain stationary at
finite temperature.
Other finite-temperature equilibria may nevertheless exist
and need not coincide with $p$.
\end{remark}

\subsection{Long-Time Behavior: Obstruction and Conditional Convergence}
\label{sec:longtime}

Along any regular solution of \eqref{eq:zero_temp} the energy dissipates as
\begin{equation}\label{eq:long_time_dissipation}
\frac{d}{dt}\mathcal E_p(q_t)=-\frac1\gamma\mathcal D(q_t),
\qquad
\mathcal D(q):=\int_\Omega\|\nabla\Phi_q\|^2\,dq,
\end{equation}
but, as noted in Section~\ref{sec:wgf}, this does not force $q_t\to p$. Throughout
this subsection $\Omega\subset\mathbb R^d$ is compact and the kernel
$e^{-c/\varepsilon}$ is positive and universal. The obstruction is elementary: take a target symmetric about the origin and the single atom $q=\delta_{\mathbf 0}$.
Every optimal plan inherits the symmetry, so every barycenter at the origin
vanishes and $F_p[\delta_{\mathbf 0}](\mathbf 0)=\mathbf 0$; the continuity equation
cannot split an atom, so the model stays at $\delta_{\mathbf 0}$.

\begin{proposition}[Non-target stationary state]\label{prop:non_target_stationary}
% Let $\Omega=[-1,1]$, let $p\in\mathcal P(\Omega)$ satisfy $R_\#p=p$ for $R(x)=-x$
% with $p\neq\delta_0$, and assume $\Phi_{\delta_0}$ admits a $C^2$ representative on
% $\Omega$. Then, for every choice of generators, $q_t\equiv\delta_0$ is a stationary
% distributional solution of \eqref{eq:zero_temp} and the metric slope
% $|\partial\mathcal E_p|(\delta_0)$ vanishes, although $\delta_0\neq p$.
Let $\Omega=[-1,1]$, let $p\in\mathcal P(\Omega)$ satisfy $\mathsf R_\#p=p$ for $\mathsf R(x)=-x$ with $p\neq\delta_0$, and assume $\Phi_{\delta_0}$ admits a $C^2$ representative on $\Omega$. Then, for every choice of generators, $q_t\equiv\delta_0$ is a stationary distributional solution of \eqref{eq:zero_temp} and the metric slope
$|\partial\mathcal E_p|(\delta_0)$ vanishes, although $\delta_0\neq p$.
\end{proposition}

Wasserstein stationarity requires $\nabla\Phi_q=\mathbf 0$ only $q$-almost everywhere,
not on all of $\Omega$. In particular, for $(\iota,\iota)$ the flow is W-Flow, whose
velocity field vanishes identically only at $q=p$ \citep{han2026one}; the proposition
shows that vanishing on $\mathrm{supp}(q)$, which is all the flow sees, is strictly
weaker. The obstruction is a property of the zero-temperature limit
(Remark~\ref{rem:diffusion}), not of the transport geometry. Convergence therefore
needs a quantitative link between the energy gap and its dissipation.

% \begin{assumption}\label{ass:longtime}
% (i) $\mathcal E_p\ge0$ on $\mathcal P(\Omega)$ with equality iff $q=p$.
% (ii) For a.e.\ $t\ge0$: $p\ll q_t$ with $dp/dq_t\in L^2(q_t)$, $q_t$ satisfies a
% Poincar\'e inequality with constant $C_{\mathrm{PI}}(q_t)$, and
% $A:=\operatorname{ess\,sup}_{t}C_{\mathrm{PI}}(q_t)\,\chi^2(p\Vert q_t)<\infty$.
% \end{assumption}

% Condition (i) holds for $(\varphi,\varphi)$ by \citet[Thm.~5]{sejourne2019sinkhorn}
% and for $(\iota,\iota)$ by \citet{Feydy2018InterpolatingBO}; for $(\iota,\mathrm{KL})$
% it is not covered by their kernel-norm argument (Section~\ref{sec:energy}) and is
% verified numerically in Appendix~\ref{app:positivity}. Condition (ii) rules out
% collapse onto supports that miss $p$ and yields
% $\mathcal E_p(q_t)^2\le A\,\mathcal D(q_t)$ (Appendix~\ref{app:longtime}).

Convergence to $p$ itself requires that the flow keep covering the target.
Appendix~\ref{app:longtime} gives sufficient conditions (convexity of $\mathcal E_p$ along mixtures, bounded $\chi^2(p\Vert q_t)$, and a uniform Poincar\'e inequality along the trajectory) under which the energy--dissipation inequality $\mathcal E_p(q_t)^2\le A\,\mathcal D(q_t)$ holds and yields $\mathcal E_p(q_t)=\mathcal O(t^{-1})$ and $\mathcal W_2(q_t,p)\to0$ (Proposition~\ref{prop:decay}). The stationary state of Proposition~\ref{prop:non_target_stationary} violates the coverage condition, since $\chi^2(p\Vert\delta_{\mathbf 0})=\infty$; without it the flow still converges along a subsequence to a Wasserstein-stationary state (Proposition~\ref{prop:stationary_limit}).

The Dirac stationary example alone does not establish whether collapse can occur from a nondegenerate initialization. Whether the coverage condition in Assumption~\ref{ass:longtime} holds throughout training remains open (Appendix~\ref{sec:limitations}). We next instantiate \eqref{eq:zero_temp} as a training algorithm.

\section{Implementation}\label{sec:impl}

% \paragraph{Mini-batch approximation.}
% Given independent mini-batches $\{\mathbf x_i\}_{i=1}^N\sim q_t$ and
% $\{\mathbf y_j\}_{j=1}^M\sim p$ with uniform weights
% $\mathbf u_N:=\frac1N\mathbf 1_N$, $\mathbf u_M:=\frac1M\mathbf 1_M$, a coupling is a
% matrix $\mathbf P\in\mathbb R_+^{N\times M}$ and \eqref{eq:uot} becomes
% \begin{equation}\label{eq:empirical_uot}
% \inf_{\mathbf P\ge0}\;
% \langle\mathbf C,\mathbf P\rangle
% +\rho D_{\varphi_1}(\mathbf P\mathbf 1_M\,\|\,\mathbf u_N)
% +\rho D_{\varphi_2}(\mathbf P^{\!\top}\mathbf 1_N\,\|\,\mathbf u_M)
% +\varepsilon\,\mathrm{KL}(\mathbf P\,\|\,\mathbf u_N\mathbf u_M^{\!\top}),
% \end{equation}
% with $C_{ij}=c(\mathbf x_i,\mathbf y_j)$. For $\varphi_1=\iota$ the first penalty is
% the hard constraint $\mathbf P\mathbf 1_M=\mathbf u_N$.
\paragraph{Mini-batch approximation.}
Given independent mini-batches $\{\mathbf x_i\}_{i=1}^N\sim q_t$ and
$\{\mathbf y_j\}_{j=1}^M\sim p$ with uniform weights, a coupling is a matrix
$\mathbf P\in\mathbb R_+^{N\times M}$ and \eqref{eq:uot} becomes
\begin{equation}\label{eq:empirical_uot}
\inf_{\mathbf P\ge0}\;
\langle\mathbf C,\mathbf P\rangle
+\rho D_{\varphi_1}\big(\mathbf P\mathbf 1_M\,\big\|\,\tfrac1N\mathbf 1_N\big)
+\rho D_{\varphi_2}\big(\mathbf P^{\!\top}\mathbf 1_N\,\big\|\,\tfrac1M\mathbf 1_M\big)
+\varepsilon\,\mathrm{KL}\big(\mathbf P\,\big\|\,\tfrac1{NM}\mathbf 1_N\mathbf 1_M^{\!\top}\big),
\end{equation}
with $C_{ij}=c(\mathbf x_i,\mathbf y_j)$. For $\varphi_1=\iota$ the first penalty is
the hard constraint $\mathbf P\mathbf 1_M=\tfrac1N\mathbf 1_N$.

\paragraph{Training procedure.}
The remaining components follow W-Flow \citep{han2026one}; we summarize
them here and defer pseudocode to Appendix~\ref{app:impl}. Transport is carried out
in the space of frozen feature blocks $\psi_\ell$ rather than in latent space. At
iteration $n$, for each block $\ell$, we compute $\mathbf x_i=\psi_\ell(f_{\boldsymbol\theta_n}(\mathbf z_i))$,
build the costs $C_{ij}=\tfrac12\|\mathbf x_i-\psi_\ell(\mathbf y_j)\|^2$ against a real
batch and, with a second independent generated batch $\tilde{\mathcal X}_n$ in place of
the real one, the self-interaction costs, so that the self term avoids trivial diagonal
matches. The plans solving \eqref{eq:empirical_uot} yield the mini-batch estimate
$\widehat{\mathbf v}^\ell_i$ of $F_p[q_n]$ at $\mathbf x_i$ (Appendix~\ref{app:velocity}),
and the generator regresses onto an explicit Euler target,
\begin{equation}\label{eq:regression}
\bar{\mathbf x}^{\ell}_i=\operatorname{sg}\big(\mathbf x_i+\eta\,\widehat{\mathbf v}^\ell_i\big),
\qquad
\mathcal L(\boldsymbol\theta)
=\frac{1}{n_\psi N}\sum_{\ell,i}
\big\|\psi_\ell(f_{\boldsymbol\theta}(\mathbf z_i))-\bar{\mathbf x}^{\ell}_i\big\|^2,
\end{equation}
where $\eta$ absorbs the mobility $\gamma^{-1}$ and the stop-gradient avoids
differentiating through the transport solve. Guidance acts on the velocity,
$\widehat{\mathbf v}_{\cls}-\widehat{\mathbf v}_{\mathrm{self}}
+w\,(\widehat{\mathbf v}_{\cls}-\widehat{\mathbf v}_\varnothing)$, with $w$ sampled
during training and fixed at inference (Appendix~\ref{app:guidance}). The departure
from \citet{han2026one} is that the plans solve \eqref{eq:empirical_uot} rather than
balanced OT; at $\tau=1$ our symmetrized estimator reduces to theirs (the balanced
reverse plan is the transposed forward plan), and our independently trained
$\tau=1$ model reproduces W-Flow-B \citep{han2026one} theoretically.% where the mobility $\gamma^{-1}$ is absorbed into $\eta$ and the stop-gradient prevents

\newcolumntype{L}[1]{>{\raggedright\arraybackslash}p{#1}}

\sisetup{
  detect-family = true,
  detect-weight = true,
  table-number-alignment = center
}

\begin{table*}[t]
\centering
\captionsetup{
  font={footnotesize,normalfont},
  labelfont=normalfont,
  textfont=normalfont
}
\caption{Quantitative comparison of balanced OT, fully unbalanced OT,
and source-fixed UOT across marginal divergences and solver settings.
The six feature-specific FDr scores are followed by their mean,
$\mathrm{FDr}^{6}$. KID and KD-DINOv2 are scaled by $100$.
Lower is better except for IS. Red, blue, and green shading denotes
the best, second-best, and third-best results in each metric across
the entire table, respectively; ties share the same rank.}
\label{table:1}
\vspace{-2pt}

\begingroup
\small
\setlength{\tabcolsep}{3.2pt}
\renewcommand{\arraystretch}{1.12}
\sisetup{
  detect-family=true,
  detect-weight=true,
  table-text-alignment=center
}

% These subgroup headings do not change the scope of the rankings.
\def\uotsection#1{%
  \multicolumn{13}{@{}l}{\footnotesize\itshape #1}\\
}

\begin{adjustbox}{max width=\textwidth,center}
\begin{tabular}{
  @{}L{4.0cm}
  Z{1.2}
  Z{1.2}
  Z{2.2}
  Z{2.2}
  Z{2.2}
  Z{2.2}
  Z{2.2}
  @{\hspace{9pt}}
  Z{1.4}
  Z{2.2}
  Z{1.2}
  Z{3.2}
  @{\,${}\pm{}$\,}
  Z{1.2}
  @{}
}
\toprule
& \multicolumn{7}{c}{\textbf{FDr} $\downarrow$}
& \multicolumn{5}{c@{}}{\textbf{Other metrics}} \\
\cmidrule(lr){2-8}
\cmidrule(lr){9-13}
\textbf{Model / setting}
& \multicolumn{1}{c}{\textbf{Incep.}}
& \multicolumn{1}{c}{\textbf{ConvNeXt}}
& \multicolumn{1}{c}{\textbf{DINOv2}}
& \multicolumn{1}{c}{\textbf{MAE}}
& \multicolumn{1}{c}{\textbf{SigLIP}}
& \multicolumn{1}{c}{\textbf{CLIP}}
& \multicolumn{1}{c}{\makecell{\textbf{Mean}\\$\mathrm{FDr}^{6}$}}
& \multicolumn{1}{c}{\makecell{\textbf{KID}\\$\downarrow$}}
& \multicolumn{1}{c}{\makecell{\textbf{KD-DINOv2}\\$\downarrow$}}
& \multicolumn{1}{c}{\makecell{\textbf{FID}\\$\downarrow$}}
& \multicolumn{2}{c@{}}{\makecell{\textbf{IS}\\$\uparrow$}} \\
\midrule

\uotgroup{Balanced OT}{both marginals fixed}

\setting{W-FLOW}{\textnormal{baseline, our retraining}}
& 4.21 & 8.92 & 27.13 & 15.35 & 50.16 & 39.50 & 24.21
& 0.1321 & 40.70 & 7.07 & 140.37 & 1.64 \\

\midrule
\uotgroup{Fully unbalanced OT}{both marginals relaxed}

\setting{S\'ejourn\'e KL-UOT}{$\tau{=}0.95$}
& 4.27 & 9.17 & 28.17 & 15.61 & 51.46 & 40.74  & 24.90 
& {\uotranksecond{0.1273}} &  43.12 & 7.17 & \multicolumn{2}{c@{}}{\uotrankthird{145.28\,${}\pm{}$\,1.87}} \\

\setting{S\'ejourn\'e KL-UOT}{$\tau{=}0.975$}
& 4.16 & 8.77 & 26.99 & 15.05 & 49.81  & 39.02 & 23.97 
& 0.1303 & 40.40 & 7.00 & 143.28 & 3.12\\

\setting{S\'ejourn\'e KL-UOT}{$\tau{=}0.985$}
& 4.18 & 8.69 & 26.79 & 15.17  & 49.32 & 38.46  & 23.77 
& 0.1362  & 39.75 & 7.03  & 142.57 & 2.63 \\

\setting{S\'ejourn\'e KL-UOT}{$\tau{=}0.99$}
& 4.16 &8.67 & 26.74  & 15.06 & 49.23 &38.34 &23.70 
& 0.1414 & 39.77& 7.00 &142.65 & 2.83 \\

\midrule

\uotgroup{Forward-only source-fixed UOT}
{source fixed; target relaxed}

\setting{KL}{$\tau{=}0.985,i{=}10$}
& 4.23 & 8.82 & 26.73 & 15.20 & 49.20 & 38.71
& 23.81
& 0.1313 & 40.04 & 7.11
& 140.83 & 2.72 \\

\midrule

\uotgroup{Symmetrized source-fixed UOT}
{source fixed; target relaxed in each directional solve}

\setting{Pearson-$\chi^2$}{$\tau{=}0.95$}
& 4.20 & 8.50 & 26.57 & 15.39 & 49.36 & 38.78 & 23.80
& 0.1377 & 39.55 & 7.06 & 144.07 & 2.45 \\

\setting{Hellinger}{$\tau{=}0.95$}
& 4.20 & 8.36 & 26.29 & {\uotrankfirst{15.04}} & 48.83 & 38.26 & 23.50
& 0.1375 & 38.65 & 7.05 & 144.90 & 1.50 \\

\cmidrule(lr){1-13}
\uotsection{KL: Marginal relaxation $\tau$ variants}

\setting{KL}{$\tau{=}0.95$}
& {\uotrankthird{4.13}} & 8.53 & 26.64 & 15.12 & 48.93 & 38.41 & 23.63
& 0.1328 & 39.46 & 6.94 & 145.08 & 2.29 \\

\setting{KL}{$\tau{=}0.975$}
& 4.17 & 8.71 & 26.84 & 15.19 & 49.31 & 38.67 & 23.81
& 0.1327 & 39.92 & 7.00 & 143.10 & 2.42 \\

\setting{KL}{$\tau\!:\!0.95\!\rightarrow\!0.985$}
& {\uotrankthird{4.13}} & 8.61 & 26.79 & 15.12 & 49.28 & 38.93 & 23.81
& {\uotrankfirst{0.1268}} & 40.04 & {\uotrankthird{6.93}} & 142.82 & 2.66 \\

\setting{KL}{$\tau{=}0.99$}
& 4.25 & 8.79 & 26.87 & 15.13 & 49.53 & 38.80 & 23.89
& 0.1375 & 39.96 & 7.14 & 142.40 & 3.14 \\

\setting{KL}{$\tau{=}0.985$}
& {\uotranksecond{4.09}} & 8.57 & 26.53 & {\uotranksecond{15.05}} & 48.64 & 38.34 & 23.54
& 0.1343 & 39.32 & {\uotranksecond{6.86}} & 143.28 & 2.49 \\

\cmidrule(lr){1-13}
\uotsection{KL: Step size $\eta$ variant}

\setting{KL}{$\tau{=}0.985,\eta{=}0.5$}
& 4.23 & 9.00 & 27.07 & 15.44 & 50.20 & 39.31 & 24.21
& 0.1374 & 40.31 & 7.10 & 141.49 & 1.99 \\

\cmidrule(lr){1-13}
\uotsection{KL: Sinkhorn iteration $i$ variants}

\setting{KL}{$\tau{=}0.985,i{=}1$}
& {\uotrankfirst{4.00}} & {\uotranksecond{7.44}} & {\uotranksecond{24.86}} & 15.28 & {\uotranksecond{47.28}} & {\uotrankthird{37.35}} & {\uotranksecond{22.70}}
& 0.1495 & {\uotranksecond{35.89}} & {\uotrankfirst{6.72}} & \multicolumn{2}{c@{}}{\uotranksecond{156.59\,${}\pm{}$\,2.70}} \\

\setting{KL}{$\tau{=}0.985,i{=}5$}
& 4.19 & 8.60 & 26.49 & {\uotrankthird{15.09}} & 48.75 & 38.43 & 23.59
& 0.1364 & 39.51 & 7.04 & 144.05 & 2.85 \\

\setting{KL}{$\tau{=}0.985,i{=}15$}
& 4.30 & 8.90 & 26.89 & 15.13 & 49.34 & 38.70 & 23.88
& 0.1387 & 39.96 & 7.22 & 141.37 & 3.41 \\

\cmidrule(lr){1-13}
\uotsection{KL: Entropic regularization $\epsilon$ variants}

\setting{KL}{$\tau{=}0.985,\epsilon{=}0.01$}
& 4.15 & {\uotrankthird{8.26}} & {\uotrankthird{25.33}} & 15.14 & {\uotrankthird{47.38}} & {\uotranksecond{37.34}} & {\uotrankthird{22.93}}
& 0.1564 & {\uotrankthird{36.43}} & 6.98 & 144.18 & 2.25 \\

\setting{KL}{$\tau{=}0.985,\epsilon{=}0.01,i{=}1$}
& 4.16 & {\uotrankfirst{6.47}} & {\uotrankfirst{22.85}} & 16.04 & {\uotrankfirst{44.78}} & {\uotrankfirst{34.77}} & {\uotrankfirst{21.51}}
& 0.2236 & {\uotrankfirst{31.28}} & 6.99 & \multicolumn{2}{c@{}}{\uotrankfirst{164.49\,${}\pm{}$\,2.40}} \\

\setting{KL}{$\tau{=}0.985,\epsilon{=}0.1$}
& 4.49 & 9.08 & 28.62 & 16.04 & 52.49 & 41.03 & 25.29
& {\uotrankthird{0.1297}} & 43.52 & 7.54 & 143.13 & 3.52 \\

\bottomrule
\end{tabular}
\end{adjustbox}
\endgroup
\end{table*}

\begin{table*}[!t]
\centering
\caption{
\footnotesize{
Density, coverage, precision, and recall on ImageNet
$256{\times}256$. All results use 50K generated images and the
ImageNet validation set as reference. Density and coverage use $k=5$;
precision and recall use $k=3$. DINOv2 results use CLS-token features.
Higher is better for all metrics. Within each of the two comparison
blocks independently, red, blue, and green denote the best,
second-best, and third-best results in each metric column, respectively;
ties share the same rank. Subsections within a block share the same ranking.
}}
\label{tab:dc_pr_all}

\begingroup
\small
\setlength{\tabcolsep}{4pt}
\renewcommand{\arraystretch}{1.12}
\sisetup{
  detect-family=true,
  detect-weight=true,
  table-text-alignment=center
}

% Red: first; blue: second; green: third.
\newcommand{\dcfirst}[1]{%
  \cellcolor{red!18}\bfseries #1%
}
\newcommand{\dcsecond}[1]{%
  \cellcolor{blue!15}\bfseries #1%
}
\newcommand{\dcthird}[1]{%
  \cellcolor{green!18}\bfseries #1%
}

% Major gray row: a new ranking block.
\newcommand{\dcgroup}[2]{%
  \rowcolor{black!8}%
  \multicolumn{11}{@{}l}{%
    \textsc{#1}\;
    \textcolor{gray}{\textit{(#2)}}%
  }\\
}

% Subsection: does not start a new ranking block.
\newcommand{\dcsubgroup}[1]{%
  \multicolumn{11}{@{}l}{\textit{#1}}\\
}

\begin{adjustbox}{max width=\textwidth,center}
\begin{tabular}{
  @{}L{4.0cm}
  c
  c
  @{\hspace{10pt}}
  Z{1.6}
  Z{1.6}
  Z{1.6}
  Z{1.6}
  @{\hspace{10pt}}
  Z{1.6}
  Z{1.6}
  Z{1.6}
  Z{1.6}
  @{}
}
\toprule
\textbf{Model / setting}
& \multicolumn{1}{c}{\textbf{Step}}
& \multicolumn{1}{c}{\textbf{CFG}}
& \multicolumn{4}{c}{\textbf{Inception}}
& \multicolumn{4}{c}{\textbf{DINOv2-CLS}} \\
\cmidrule(lr){4-7}
\cmidrule(lr){8-11}
& &
& \multicolumn{1}{c}{\makecell{\textbf{Density}\\$\uparrow$}}
& \multicolumn{1}{c}{\makecell{\textbf{Coverage}\\$\uparrow$}}
& \multicolumn{1}{c}{\makecell{\textbf{Precision}\\$\uparrow$}}
& \multicolumn{1}{c}{\makecell{\textbf{Recall}\\$\uparrow$}}
& \multicolumn{1}{c}{\makecell{\textbf{Density}\\$\uparrow$}}
& \multicolumn{1}{c}{\makecell{\textbf{Coverage}\\$\uparrow$}}
& \multicolumn{1}{c}{\makecell{\textbf{Precision}\\$\uparrow$}}
& \multicolumn{1}{c}{\makecell{\textbf{Recall}\\$\uparrow$}} \\
\midrule

% ================================================================
% Ranking block 1: all listed 30K ablations, including forward-only
% ================================================================
\dcgroup{30K Ablations}{ranked within this block}

\dcsubgroup{Balanced OT: both marginals fixed}
\setting{W-FLOW}{baseline, our retraining}
& 30K & 1.50
& 1.275972 & 0.879240 & 0.766760 & 0.391320
& 0.360704 & 0.237440 & 0.441700 & 0.159980 \\

\cmidrule(lr){1-11}
\dcsubgroup{Fully unbalanced OT: both marginals relaxed}

\setting{S\'ejourn\'e KL-UOT}{$\tau{=}0.95$}
& 30K & 1.50
& 1.274408 & 0.876160 & 0.766900 & {\dcsecond{0.398180}} 
& 0.356112  & 0.227020 & 0.436980  & 0.162640 \\

\setting{S\'ejourn\'e KL-UOT}{$\tau{=}0.975$}
& 30K & 1.50
& 1.282540 & 0.880580 & 0.766400 & 0.391820 
& 0.357960 & 0.235980  & 0.439280  & 0.165760 \\

\setting{S\'ejourn\'e KL-UOT}{$\tau{=}0.985$}
& 30K & 1.50
& 1.280780 & 0.883520 & 0.763340 & {\dcthird{0.394160}}
& 0.353650 & 0.236880 & 0.438820 & {\dcthird{0.168260}}  \\

\setting{S\'ejourn\'e KL-UOT}{$\tau{=}0.99$}
& 30K & 1.50
& 1.289264 & 0.883520 & 0.768060  & 0.390340
& 0.355656 & 0.238200 & 0.443180 & 0.167400 \\

\cmidrule(lr){1-11}
\dcsubgroup{Forward-only source-fixed UOT: source fixed; target relaxed}

\setting{KL}{$\tau{=}0.985,i{=}10$}
& 30K & 1.50
& 1.285796 & 0.880000 & 0.763980 & 0.391540
& {\dcsecond{0.369696}} & 0.238960 & 0.443620 & 0.159780 \\

\cmidrule(lr){1-11}
\dcsubgroup{Symmetrized source-fixed UOT: source fixed; target relaxed in each directional solve}

\setting{Pearson-$\chi^2$}{$\tau{=}0.95$}
& 30K & 1.50
& 1.318908 & 0.885160 & 0.773060 & 0.382380
& {\dcthird{0.368712}} & 0.242340
& {\dcthird{0.445860}} & 0.149680 \\

\setting{Hellinger}{$\tau{=}0.95$}
& 30K & 1.50
& {\dcthird{1.325380}} & {\dcthird{0.886320}}
& {\dcthird{0.774840}} & 0.377060
& {\dcfirst{0.375876}} & {\dcthird{0.245320}}
& {\dcfirst{0.450580}} & 0.153000 \\

\cmidrule(lr){1-11}

\setting{KL}{$\tau{=}0.95$}
& 30K & 1.50
& 1.297344 & 0.881000 & 0.768580 & 0.384120
& 0.365412 & 0.239720 & 0.443620 & 0.164260 \\

\setting{KL}{$\tau{=}0.975$}
& 30K & 1.50
& 1.294996 & 0.883800 & 0.769840 & 0.387900
& 0.362096 & 0.237740 & 0.440280 & 0.156080 \\

\setting{KL}{$\tau\!:\!0.95\!\rightarrow\!0.985$}
& 30K & 1.50
& 1.295212 & 0.882240 & 0.766400 & 0.389620
& 0.358616 & 0.240260 & 0.441440 & 0.159420 \\

\setting{KL}{$\tau{=}0.99$}
& 30K & 1.50
& 1.286124 & 0.885620 & 0.767260 & 0.383880
& 0.360840 & 0.238480 & 0.443160 & 0.162020 \\

\setting{KL}{$\tau{=}0.985$}
& 30K & 1.50
& 1.295844 & {\dcsecond{0.886740}} & 0.769360 & 0.393320
& 0.356080 & 0.237660 & 0.442120 & 0.160920 \\

\cmidrule(lr){1-11}

\setting{KL}{$\tau{=}0.985,\eta{=}0.5$}
& 30K & 1.50
& 1.266760 & 0.878880 & 0.760660 & 0.393480
& 0.361552 & 0.236480 & 0.441680 & 0.152500 \\

\cmidrule(lr){1-11}

\setting{KL}{$\tau{=}0.985,i{=}1$}
& 30K & 1.50
& {\dcfirst{1.342936}} & {\dcfirst{0.888300}}
& {\dcfirst{0.783180}} & 0.372300
& 0.359704 & {\dcsecond{0.253880}}
& {\dcsecond{0.447560}} & 0.153800 \\

\setting{KL}{$\tau{=}0.985,i{=}5$}
& 30K & 1.50
& 1.288708 & 0.883800 & 0.772880 & 0.384600
& 0.358848 & 0.243140 & 0.441600 & 0.156920 \\

\setting{KL}{$\tau{=}0.985,i{=}15$}
& 30K & 1.50
& 1.299728 & 0.880700 & 0.766340 & 0.386240
& 0.366136 & 0.236440 & 0.445840 & 0.167740 \\

\cmidrule(lr){1-11}

\setting{KL}{$\tau{=}0.985,\epsilon{=}0.01$}
& 30K & 1.50
& 1.232700 & 0.875160 & 0.758220 & {\dcfirst{0.405240}}
& 0.344216 & 0.244160 & 0.436800 & {\dcsecond{0.175820}} \\

\setting{KL}{$\tau{=}0.985,\epsilon{=}0.01,i{=}1$}
& 30K & 1.50
& 1.258084 & 0.878020 & 0.772440 & 0.388720
& 0.348040 & {\dcfirst{0.257740}}
& 0.436420 & {\dcfirst{0.183040}} \\

\setting{KL}{$\tau{=}0.985,\epsilon{=}0.1$}
& 30K & 1.50
& {\dcsecond{1.327708}} & 0.877540
& {\dcsecond{0.780740}} & 0.359080
& 0.368236 & 0.224580 & 0.445700 & 0.140040 \\

\midrule[0.7pt]

% ================================================================
% Ranking block 2: scaled UOT + official/external models
% ================================================================
\dcgroup{Scaled and Official Models}{ranked within this block}

\dcsubgroup{Scaled symmetrized source-fixed UOT models}

\setting{KL B/2}{$\tau{=}0.95$}
& 200K & 1.19
& 1.067172 & 0.947400 & 0.721460 & 0.611280
& 0.497324 & 0.553880 & 0.544420 & 0.428020 \\

\setting{KL B/2}{$\tau{=}0.985$}
& 200K & 1.19
& 1.064864 & 0.948600 & {\dcthird{0.723980}} & 0.615300
& 0.502468 & 0.563280 & 0.547880 & 0.428720 \\

\setting{KL B/2}{$\tau{=}0.985,i{=}1$}
& 200K & 1.19
& {\dcsecond{1.079880}} & 0.948760
& {\dcsecond{0.731400}} & 0.602500
& 0.508476 & 0.563320 & 0.553580 & 0.429280 \\

\cmidrule(lr){1-11}

\setting{KL L/2}{$\tau{=}0.985$}
& 200K & 1.15
& 1.040424 & 0.948140 & 0.712040 & {\dcthird{0.627460}}
& 0.501716 & 0.576740 & 0.544980 & 0.467820 \\

\setting{KL L/2}{$\tau{=}0.985,i{=}1$}
& 200K & 1.15
& 1.048880 & 0.947580 & 0.721940 & 0.617800
& 0.510300 & 0.580980 & 0.556860 & 0.461920 \\

\cmidrule(lr){1-11}

\setting{KL XL/2}{$\tau{=}0.985$}
& 200K & 1.14
& 1.040012 & {\dcthird{0.949280}}
& 0.714240 & {\dcsecond{0.638660}}
& {\dcsecond{0.561476}} & {\dcsecond{0.644860}}
& {\dcfirst{0.585140}} & {\dcsecond{0.502080}} \\

\cmidrule(lr){1-11}
\dcsubgroup{Official checkpoints and external reference}

\setting{Pixel Drift-L}{\cite{deng2026generative}}
& 100K & 1.00
& {\dcfirst{1.091920}} & {\dcfirst{0.951360}}
& {\dcfirst{0.732360}} & 0.595420
& {\dcfirst{0.569008}} & {\dcfirst{0.658880}}
& {\dcsecond{0.580220}} & {\dcthird{0.485960}} \\

\setting{W-Flow B/2}{\cite{han2026one}}
& 200K & 1.19
& {\dcthird{1.072436}} & {\dcsecond{0.950140}}
& 0.721940 & 0.613040
& 0.505844 & 0.561520 & 0.548800 & 0.428560 \\

\setting{W-Flow L/2}{\cite{han2026one}}
& 200K & 1.14
& 1.041888 & 0.946180 & 0.711580 & 0.626180
& 0.511296 & 0.582060 & 0.556000 & 0.463300 \\

\setting{W-Flow XL/2}{\cite{han2026one}}
& 180K & 1.09
& 0.935748 & 0.937380 & 0.680620 & {\dcfirst{0.667500}}
& {\dcthird{0.538504}} & {\dcthird{0.639100}}
& {\dcthird{0.569640}} & {\dcfirst{0.516040}} \\

\bottomrule
\end{tabular}
\end{adjustbox}

\endgroup
\end{table*}

\section{Experiments}
\label{sec:exp}

\paragraph{Setup.}
We evaluate on class-conditional ImageNet $256\times256$. Following W-Flow \citep{han2026one}, a DiT-style generator operates in the SD-VAE latent space \citep{Rombach_2022_CVPR}, transport costs use the released pretrained MAE encoders \citep{deng2026generative}, and guidance acts on the velocity with its weight sampled during training (Appendix~\ref{app:guidance}); sampling takes one forward pass. By default we use ten Sinkhorn iterations and $\varepsilon=0.05$ at every scale, since for our method ten outperform one at B/2 and L/2 (Table~\ref{tab:fdr-system-level}); W-Flow instead uses one iteration at L/2 and XL/2 and $\varepsilon=0.01$ at XL/2, and its other training settings are kept (Appendix~\ref{app:impl}). Ablations use DiT-B/2 for 100 epochs. FID \citep{3295222.3295408} uses 50K samples under the ADM protocol \citep{3540261.3540933}.% We evaluate on class-conditional ImageNet $256\times256$. Following W-Flow \citep{han2026one}, a DiT-style generator operates in the SD-VAE latent space \citep{Rombach_2022_CVPR}, transport costs use the released pretrained MAE encoders \citep{deng2026generative}, and guidance acts on the velocity with its weight sampled during training (Appendix~\ref{app:guidance}); sampling takes one forward pass. By default we use ten Sinkhorn iterations and $\varepsilon=0.05$ at every scale, since ten outperform one at scale (Table~\ref{tab:fdr-system-level}); W-Flow instead uses one iteration at L/2 and XL/2 and $\varepsilon=0.01$ at XL/2, and its other training settings are kept (Appendix~\ref{app:impl}). Ablations use DiT-B/2 for 100 epochs. FID \citep{3295222.3295408} uses 50K samples under the ADM protocol \citep{3540261.3540933}.
% We evaluate on class-conditional ImageNet $256\times256$. Following
% W-Flow \citep{han2026one}, the generator is a DiT-style network in the
% SD-VAE latent space \citep{Rombach_2022_CVPR}, and transport costs are
% computed using the released pretrained MAE encoders
% \citep{deng2026generative}. Sampling uses a single forward pass, with
% guidance applied to the velocity field and its weight sampled during
% training (Appendix~\ref{app:guidance}). Ablations use DiT-B/2 for
% 100 epochs with $\varepsilon=0.05$. Unless otherwise specified, we use ten Sinkhorn iterations, which outperform one at scale (Table~\ref{tab:fdr-system-level}). Scaled models follow W-Flow's architecture and training settings but keep this solver (ten iterations, $\varepsilon=0.05$) at every scale, whereas W-Flow uses one iteration at L/2 and XL/2 and $\varepsilon=0.01$ at XL/2 (Appendix~\ref{app:impl}). FID \citep{3295222.3295408} is computed on 50K samples under the ADM protocol \citep{3540261.3540933}.

\paragraph{Run-to-run variation.}
Each row of Table~\ref{table:1} is a single training run, so we calibrate differences against independent retrainings of the same configuration. Evaluation is deterministic given a checkpoint and sampling seed; across three sampling seeds FID varies by $\pm$ 0.009%
. Our from-scratch retraining of the balanced baseline gives 7.07 FID and 24.21 FDr$^6$ at 30K steps, against 7.08 and 24.18 for the 30K checkpoint released by \citet{han2026one} under the same pipeline (they report 7.08 under their protocol). At 200K steps, our $\tau=1$ B/2 model and the released W-Flow-B checkpoint give 1.53 vs.\ 1.52 FID and 11.64 vs.\ 11.71 FDr$^6$ (Table~\ref{tab:fdr-system-level}). The two retrainings thus differ by 0.01 FID and 0.03 FDr$^6$ at 30K steps and by 0.01 FID and 0.07 FDr$^6$ at 200K steps, well below the differences discussed next. Over three independent training seeds, balanced and $\tau=0.985$ give $7.07\pm 0.05$ and $6.86\pm 0.06$ FID, and $24.21\pm 0.04$ and $23.54\pm 0.06$ FDr$^6$.

\paragraph{Marginal relaxation.}
Relaxing the target marginal helps: all seven symmetrized source-fixed runs in Table~\ref{table:1} improve on balanced transport in FDr$_6$ (23.50--23.89 vs.\ 24.21) and in each of the six feature spaces, and so does the two-sided KL-UOT of \citet{Sjourn2019SinkhornDF} at $\tau=0.975$ and $0.985$ (23.97 and 23.77). Source-fixed transport is nonetheless preferable: at the respective best settings it reaches 6.86 FID and 23.54 FDr$_6$ against 7.00 and 23.77 for two-sided; at matched $\tau=0.985$ it is better in all six spaces; and it is more robust to the relaxation strength, still improving on balanced at $\tau=0.95$ (23.63, 6.94) where two-sided falls below it. In FID the differences are smaller and more $\tau$-sensitive for both formulations. Forward-only transport matches the symmetrized family in FDr$_6$ (23.81) but not in FID (7.11 vs.\ 6.86). Regularization $\varepsilon$ also affects performance.

\paragraph{Scaling.}
The solver setting interacts with the marginal treatment: at 30K steps, balanced transport improves from one to ten Sinkhorn iterations whereas source-fixed UOT is best at one (6.72 vs.\ 6.86; Table~\ref{table:1}); at 200K steps UOT is best at ten (Table~\ref{tab:fdr-system-level}). We therefore compare each method at its best setting. UOT-GF ($\tau=0.985$) achieves FIDs of $\mathbf{1.46}/\mathbf{1.34}/\mathbf{1.22}$ at B/L/XL (Table~\ref{tab:fdr-system-level}); the XL result is a new one-step state of the art in FID among models trained from scratch. Two comparisons are controlled: at B/2, against balanced transport under identical settings (1.53; the released W-Flow-B gives 1.52), and at XL/2, where balanced transport under our recipe reaches 1.39 FID at 100K steps against 1.26 for UOT-GF with otherwise identical settings, so the gain persists at scale. The released W-Flow-L and W-Flow-XL checkpoints \citep{han2026one} (1.36/1.33) differ in solver settings (and W-Flow-XL trains 180K steps); at W-Flow's L/2 setting (one iteration) relaxation does not help (1.40 vs.\ 1.36), and balanced transport at ten iterations is untested at L/2. UOT-GF remains competitive with multi-step models at substantially lower sampling compute (Figure~\ref{fig:benchmark_samples}; Figure~\ref{fig:efficiency} and Figure~\ref{fig:efficiency_cd}). On other feature-space metrics the picture is mixed: at B/2, where settings and guidance match, precision, recall and coverage change by less than $0.003$ (Table~\ref{tab:dc_pr_all}); the lower XL recall ($0.639$ vs.\ $0.668$) coincides with stronger guidance ($1.14$ vs.\ $1.09$).

\begin{figure*}[t]
  \centering
  \includegraphics[width=\textwidth]{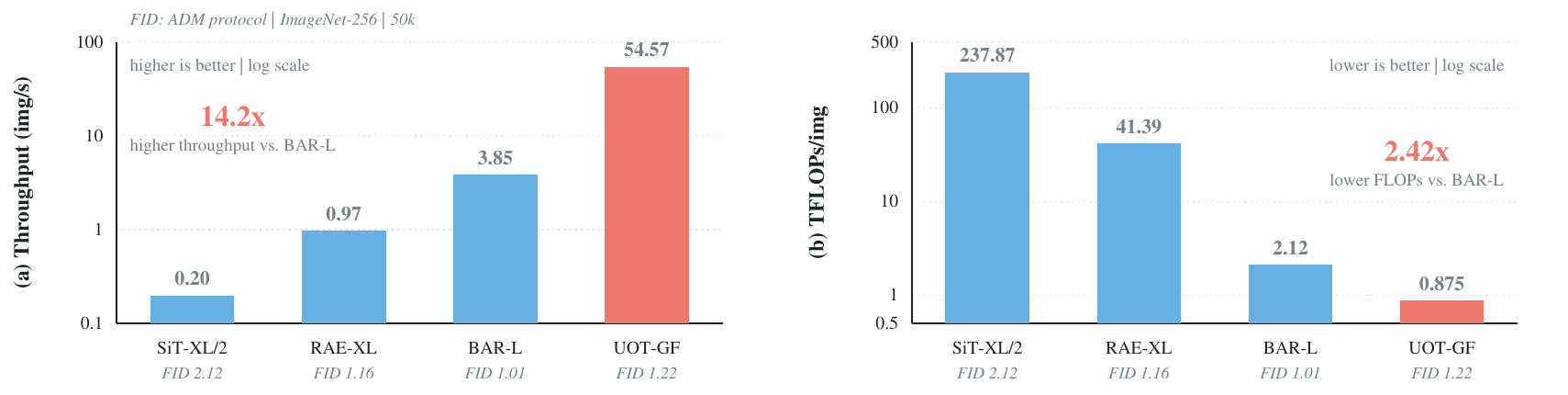}
  \captionsetup{skip=3pt}
  \caption{Sampling efficiency on ImageNet-256, measured on a single Nvidia H20 under identical settings (log scale). (a) Throughput in images per second; (b) compute per image in TFLOPs. UOT-GF uses one network evaluation. Measurement details are in Appendix~\ref{app:addr}.}
  \label{fig:efficiency}
\end{figure*}

\providecommand{\grayentry}[1]{\textcolor{black!55}{#1}}
\providecommand{\shadeword}[1]{\colorbox{gray!15}{\strut #1}}
\providecommand{\tblcite}[1]{\textcolor{blue!70!black}{[#1]}}
\definecolor{oursgray}{RGB}{238,235,231}
\definecolor{citationblue}{RGB}{0,102,153}

\DeclareRobustCommand{\shadeword}[1]{%
  {\setlength{\fboxsep}{1pt}\colorbox{oursgray}{#1}}%
}

\begin{table*}[t]
\centering

\captionsetup{
  font=normalfont,
  labelfont=normalfont,
  textfont=normalfont,
  justification=raggedright,
  singlelinecheck=false
}
\vspace{-2mm}
\caption{%
  System-level comparison on ImageNet $256\times256$.
  NFE and Params. denote sampling evaluations and model size. All metrics for all methods are computed by us under a unified evaluation pipeline; numbers may differ slightly from the original papers.
}
\label{tab:fdr-system-level}

\begingroup
\scriptsize
\setlength{\tabcolsep}{2.2pt}
\renewcommand{\arraystretch}{1.02}
\arrayrulecolor{black!65}

\begin{adjustbox}{max width=0.99\textwidth,center}
\begin{tabular}{@{}l|cc|ccccc|c|ccccc@{}}
Method
& NFE
& Params.
& Incep.
& ConvNeXt
& DINOv2
& MAE
& SigLIP
& FDr-CLIP
& FID$\downarrow$
& IS$\uparrow$
& FDr$^{6}\!\downarrow$
& Precision$\uparrow$
& Recall$\uparrow$
\\
\hline

\grayentry{50k validation images}
& \grayentry{N/A}
& \grayentry{N/A}
& \grayentry{1.00}
& \grayentry{1.00}
& \grayentry{1.00}
& \grayentry{1.00}
& \grayentry{1.00}
& \grayentry{1.00}
& \grayentry{1.68}
& \grayentry{232.2}
& \grayentry{1.00}
& \grayentry{0.75}
& \grayentry{0.66}
\\
\hline

\multicolumn{14}{@{}l}{\grayentry{\emph{discrete}}}\\
VAR-d30 \citep{var}
&  10$\times$2 & 2B
& 1.18 & 1.70 & 5.31 & 6.83 & 11.89
& 13.31 & 1.97 & 304.6 & 6.70& \textbf{0.82} & 0.59 \\

% BAR-B \tblcite{57}
% & $256 \times 2\times 4$ & 1.1B
% & 0.68 & 0.93 & 4.13 & 5.02 & 8.30
% & 6.78 & 1.15 & 273.9 & 4.31 & -- & -- \\

BAR-L \citep{yu2026autoregressive}
& 256$\times$2$\times$\tiny{4} & 1.1B
& \textbf{0.61} & \textbf{0.78} & 3.29 & 4.20 & 6.60
& \textbf{5.98} & \textbf{1.01} & 281.9 & 3.57 &0.77 &\textbf{0.68} \\
\hline

\multicolumn{14}{@{}l}{\grayentry{\emph{latent, multi-step}}}\\
\multicolumn{14}{@{}l}{%
  \hspace{0.4em}\grayentry{\emph{without semantic distillation}}}\\

SiT-XL/2 \citep{10.1007/978-3-031-72980-5_2}
& 250$\times$2 & 675M
& 1.26 & 2.02 & 7.89 & 5.62 & 16.14
& 17.69 & 2.12 & 256.7 & 8.44 & 0.81 & 0.60\\

MAR-L \citep{li2024autoregressive}
& 256$\times$2$\times$\tiny{100} & 478M
& 1.07 & 1.10 & 6.09 & 4.38 & 12.67
& 14.78 & 1.80 & 293.4 & 6.68 &0.80 &0.60\\

FlowAR-H \citep{ren2025flowar}
& 50$\times$2 & 1.9B
& 1.00 & 1.30 & 4.68 & 4.59 & 12.75
& 12.49 & 1.68 & 274.1 & 6.13 & 0.80 & 0.62 \\

MAR-H \citep{li2024autoregressive}
& 256$\times$2$\times$\tiny{100} & 942M
& 0.93 & 0.95 & 4.95 & 3.71 & 10.07
& 13.02 & 1.56 & 299.5 & 5.61 & 0.80 & 0.62 \\

MAR-L, DeTok \citep{yang2026latent}
& 256$\times$2$\times$\tiny{100} & 478M
& 0.83 & 1.36 & 4.66 & 4.40 & 9.57
& 12.12 & 1.39 & 306.2 & 5.49& 0.81 & 0.62  \\

\multicolumn{14}{@{}l}{%
  \hspace{0.4em}\grayentry{\emph{with semantic distillation}}}\\

REG \citep{wu2025representation}
& 250$\times$2 & 685M
& 0.92 & 1.14 & 3.45 & 3.02 & 8.42
& 10.86 & 1.54 & 302.9 & 4.64& 0.78 & 0.62   \\

SiT-XL/2-REPA \citep{yu2025representation}
& 250$\times$2  & 675M
& 0.85 & 1.22 & 4.27 & 3.85 & 9.87
& 12.65 & 1.42 & 306.1 & 5.45 & 0.80 & 0.65 \\

LightningDiT \citep{Yao2025ReconstructionVG}
& 250$\times$2  & 675M
& 0.85 & 1.09 & 3.76 & 3.02 & 8.47
& 10.21 & 1.42 & 294.3 & 4.57 & 0.80 & 0.64 \\

DDT-XL \citep{Wang_2026_CVPR}
& 250$\times$2  & 675M
& 0.75 & 1.02 & 4.26 & 4.11 & 10.16
& 13.86 & 1.26 & 309.3 & 5.70 & 0.79 & 0.66 \\

REPA-E \citep{11445131}
& 250$\times$2  & 676M
& 0.70 & 1.28 & 2.44 &\textbf{2.52}& 5.04
& 6.28 & 1.17 & 298.3 & \textbf{3.04} & 0.79 & 0.66 \\

RAE-XL \citep{zheng2026diffusion}
& 50$\times$2  & 839M
& 0.69 & 1.79 & \textbf{2.11} & 3.30 & \textbf{3.79}
& 7.87 & 1.16 & 261.0 & 3.26 & 0.77 & 0.67  \\
\hline

\multicolumn{14}{@{}l}{\grayentry{\emph{latent, one-step}}}\\

Drift-L \citep{deng2026generative}

& 1 & 463M
& 0.91 & 2.03 & 10.35 & 6.51 & 24.12
& 21.59 & 1.53 & 257.2 & 10.92 & 0.79 & 0.63 \\

iMF-XL \citep{Geng_2026_CVPR}
& 1 & 610M
& 1.09 & 1.72 & 7.30 & 6.09 & 17.02
& 17.14 & 1.82 & 278.9 & 8.39  & 0.78 & 0.63 \\

iMF-XL \citep{Geng_2026_CVPR}
& 2 & 610M
& 0.96 & 1.54 & 6.31 & 5.62 & 14.91
& 15.52 & 1.61 & 289.1 & 7.48  & 0.79 & 0.63 \\

W-Flow-B \citep{han2026one}
& 1 & 133M
& 0.91 & 2.10 & 11.48 & 7.54 & 25.67
& 22.54 & 1.52 & 271.9 & 11.71  & 0.72 & 0.61 \\

W-Flow-L \citep{han2026one}
& 1 & 463M
& 0.81 & 1.81 & 10.12 & 6.94 &22.71 
& 20.59 & 1.36 &   271.7  & 10.49  & 0.71 & 0.63\\

W-Flow-XL \citep{han2026one}
& 1 & 679M
& 0.79 & 1.68  & 8.89  &  6.69 &
20.48 & 18.71 & 1.33  & 265.6 & 9.54  & 0.68  & 0.67  \\
\hline

\multicolumn{12}{@{}l}{\grayentry{\emph{pixel, multi-step}}}\\

PixNerd-XL \citep{wang2026pixnerd}
& 250$\times$2 & 1.0B
& 1.25 & 1.21 & 3.57 & 3.56 & 9.12
& 11.36 & 2.10 & \textbf{318.8} & 5.01 & 0.81 & 0.59 \\

JiT-L \citep{Li_2026_CVPR}
& 250$\times$2 & 459M
& 1.54 & 3.49 & 6.10 & 8.07 & 19.37
& 25.82 & 2.59 & 288.5 & 10.73 & 0.79 & 0.59 \\

JiT-H \citep{Li_2026_CVPR}
& 250$\times$2 & 953M
& 1.18 & 2.52 & 4.28 & 5.65 & 11.91
& 20.40 & 1.97 & 296.0 & 7.66 & 0.78 & 0.63 \\
\hline

\multicolumn{12}{@{}l}{\grayentry{\emph{pixel, one-step}}}\\

Drift-L \citep{deng2026generative}

& 1 & 465M
& 0.87 & 1.26 & 3.88 & 3.44 & 8.70
& 10.48 & 1.46 & 298.7 & 4.40 & 0.81 & 0.60 \\
pMF-L \citep{lu2026onestep}

& 1 & 410M & 1.62 & 1.36 & 6.70 & 9.72 & 20.34
& 14.81 & 2.72 & 261.7 & 9.09  & 0.81 & 0.56 \\

pMF-H \citep{lu2026onestep}
&1&935M& 1.37 & 1.15 & 5.43 & 6.25 & 15.33
& 11.68 & 2.29 & 267.2 & 6.87  & 0.80 & 0.59 \\
\hline

\multicolumn{12}{@{}l}{\grayentry{\emph{Unbalanced optimal transport (\textbf{ours})}}}\\
\rowcolor{oursgray}

KL $\tau = 1$ - B
& 1 & 133M & 0.91 & 2.10 & 11.48 & 7.51
& 25.47 & 22.39 
&1.53 & 271.1 & 11.64 & 0.72 & 0.61
\\

\rowcolor{oursgray}
KL $\tau = 0.95$ - B & 1 & 133M & 0.90  & 2.01& 11.40 &7.55 &25.39 
& 22.38   & 1.51  &277.2 &11.61 & 0.72 & 0.61 \\

% \rowcolor{oursgray}
% S\'ejourn\'e KL $\tau = 0.985$ - B
% & 1 & 133M &  &  & 
% &  & 
% & &  & 
% & &  &  \\

\rowcolor{oursgray}
KL $\tau = 0.985$ - B
& 1 & 133M & 0.87 & 2.05 & 11.37
& 7.47 & 25.19
& 22.13 & 1.46 & 271.5
& 11.51 & 0.72 & 0.62 \\

\rowcolor{oursgray}
KL $\tau = 0.985$, $i=1$ - B
& 1 & 133M & 0.92 & 1.92  & 11.15  & 7.54  &25.47 
&22.41  & 1.54 & 280.5 & 11.56& 0.73 & 0.60 \\

\rowcolor{oursgray}
KL $\tau = 0.985$ - L
&1  & 463M & 0.80
& 1.91 & 10.57 
& 7.14  & 23.90  & 21.11
& 1.34 & 268.4 & 10.91 & 0.71 & 0.63 \\

\rowcolor{oursgray}
KL $\tau = 0.985$, $i=1$ - L
&1  & 463M & 0.84 & 1.76 & 10.21 & 7.03 & 23.41 
&20.92 & 1.40  & 280.6 & 10.70 & 0.72 & 0.62 \\

% \rowcolor{oursgray}
% W-Flow $\epsilon=0.05$, $i=10$ - XL - $100k$
% & 1  & 679M &  
% &  & 
% & &  & 
% & &  &  & &  \\

\rowcolor{oursgray}
KL $\tau = 0.985$ - XL - $100k$
& 1  & 679M & 0.75 
& 1.59 & 9.16
&6.34  & 20.74 &19.29 
&1.26 & 287.6 & 9.65 & 0.72 & 0.63   \\

\rowcolor{oursgray}
KL $\tau = 0.985$ - XL - $200k$
& 1 & 679M  & 0.74 & 1.60 &9.00 
& 6.20 & 20.26 & 19.05 & 1.22  & 284.2 & 9.47  & 0.71 & 0.64 \\

\end{tabular}
\end{adjustbox}
\endgroup
\end{table*}

\section{Conclusion}

% We introduced UOT-GF, a one-step generative framework that guides the training dynamics of image generators by unbalanced optimal transport. By fixing the generator marginal while relaxing the target marginal, the proposed geometry enables adaptive mass reweighting and yields improved transport directions. We derived the corresponding gradient-flow formulation and implemented it through mini-batch UOT plans, explicit Euler targets, and stop-gradient regression. On class-conditional ImageNet $256\times256$, UOT-GF consistently improves over balanced W-Flow across model scales and achieves an FID of \textbf{1.27}, establishing a new state of the art among one-step generators while remaining competitive with multi-step methods. These results demonstrate that unbalanced transport provides an effective and scalable geometry for efficient generative modeling. Future work may extend this framework to higher-resolution and more general conditional generation settings.

% UOT-GF drives drifting with symmetrized source-fixed transport, which improves on balanced transport under identical settings, whereas the tested two-sided relaxation does not. Our kinetic analysis recovers drifting as a limit, identifies non-target stationary states, and gives sufficient conditions for convergence. UOT-GF reaches 1.22 FID on ImageNet $256\times256$, a new one-step state of the art among models trained from scratch.

We introduced UOT-GF using symmetrized source-fixed transport, allowing marginal reweighting with lower bounds on both interaction marginals. In our experiments, it improves generation over balanced transport. Our kinetic analysis recovers drifting, including the balanced $\tau\to1$ limit, and identifies non-target stationary states and sufficient conditions for target convergence. On ImageNet $256$, UOT-GF achieves a one-step state-of-the-art FID of \textbf{1.22}. These results motivate extensions to higher resolutions and broader conditional settings.

% We introduced UOT-GF, a one-step generative framework that guides image-generator training through unbalanced optimal transport. Relaxing the target marginal enables adaptive mass reweighting and improved transport directions, which we realize using mini-batch UOT plans, explicit Euler targets, and stop-gradient regression. On class-conditional ImageNet $256\times256$, UOT-GF consistently outperforms balanced W-Flow across model scales and achieves an FID of \textbf{1.22}, setting a new state of the art among one-step generators while remaining competitive with multi-step methods. These results establish unbalanced transport as an effective and scalable geometry for efficient generation, with promising extensions to higher-resolution and broader conditional settings.

\bibliography{iclr2027_conference}
\bibliographystyle{iclr2027_conference}
\newpage
\appendix

\section*{AI Use Statement}

We used large language models as writing and engineering assistants: to polish the prose and draft text from the authors' outlines, to search for related work (every reference was checked against the original source by the authors), to help modify and debug experimental code, to typeset tables and algorithms, to check the presentation and notational consistency of the mathematical statements, and to obtain critical feedback on drafts. They were not used to generate research ideas, to design the method, to run or analyze experiments, or to produce any reported result or figure. All content produced with their assistance was reviewed and verified by the authors, who take full responsibility for the paper.

\section{Notation}
\label{app:notation}

This appendix collects the notation used throughout the paper. Vectors are bold
lowercase, matrices bold uppercase, and random vectors bold uppercase as well; the
two never appear in the same context. Generic measures are written
$\alpha,\alpha'$; $q$ is always the model distribution and $p$ the target. Symbols
that appear only inside a single proof are listed in the last block and are not
reused elsewhere.

{\renewcommand{\arraystretch}{1.4}
\begin{longtable}{p{1.25in}p{3.25in}}

\multicolumn{2}{l}{\bf Numbers and Arrays}\\*[2pt]
$\displaystyle a$ & A scalar\\
$\displaystyle \mathbf x$ & A vector (bold lowercase)\\
$\displaystyle \mathbf P$ & A matrix (bold uppercase)\\
$\displaystyle \mathbf 1_N$ & All-ones vector of length $N$\\
$\displaystyle \mathbf x^{\!\top},\,\mathbf P^{\!\top}$ & Transpose\\
$\displaystyle \|\mathbf x\|$ & Euclidean norm\\
$\displaystyle \langle\mathbf C,\mathbf P\rangle$ & Frobenius inner product, $\sum_{ij}C_{ij}P_{ij}$\\
$\displaystyle \delta_{\mathbf x}$ & Dirac mass at $\mathbf x$\\[6pt]

\multicolumn{2}{l}{\bf Sets and Spaces}\\*[2pt]
$\displaystyle \mathbb R,\,\mathbb R_+$ & Real numbers; nonnegative reals\\
$\displaystyle d,\,d_0$ & Dimension of the latent (data) space; of the noise space\\
$\displaystyle \Omega\subset\mathbb R^d$ & Domain of the latent space\\
$\displaystyle \{1,\dots,N\}$ & Integers from $1$ to $N$\\
$\displaystyle [a,b],\,(a,b]$ & Closed interval; half-open interval\\
$\displaystyle \mathcal P(\Omega),\,\mathcal P_2(\mathbb R^d)$ & Probability measures on $\Omega$; those with finite second moment\\
$\displaystyle \mathcal M_+(\Omega^2)$ & Nonnegative finite measures on $\Omega\times\Omega$\\
$\displaystyle \mathrm{supp}(q)$ & Support of a measure $q$\\
$\displaystyle C^k(\Omega),\,C_c^\infty(\Omega)$ & $k$-times differentiable functions; smooth compactly supported functions\\
$\displaystyle L^1(\alpha),\,L^2(\alpha)$ & Integrable and square-integrable functions w.r.t.\ $\alpha$\\
$\displaystyle \mathcal K$ & $\mathcal W_2$-compact set in Assumption~\ref{ass:longtime}\\[6pt]

\multicolumn{2}{l}{\bf Indexing}\\*[2pt]
$\displaystyle x_i,\,P_{ij}$ & Element $i$ of vector $\mathbf x$; element $(i,j)$ of matrix $\mathbf P$\\
$\displaystyle i,\,j$ & Indices of generated and real samples in a mini-batch\\
$\displaystyle n$ & Training iteration\\
$\displaystyle t$ & Continuous time\\
$\displaystyle k\in\{1,2\}$ & Marginal slot of a transport problem\\
$\displaystyle \ell$ & Feature-block index\\
$\displaystyle \mathbf x_k,\,\mathbf x_{-k}$ & The $k$-th coordinate of a pair; the other coordinate\\
$\displaystyle (\cdot)^{(k)}$ & Quantity attached to the $k$-th marginal\\
$\displaystyle (\cdot)_{\alpha_1,\alpha_2}$ & Quantity attached to the pair $(\alpha_1,\alpha_2)$\\
$\displaystyle \pi_{1,i},\,\pi_{2,j}$ & Row and column sums of a mini-batch plan, $\sum_jP_{ij}$ and $\sum_iP_{ij}$\\[6pt]

\multicolumn{2}{l}{\bf Calculus and Operators}\\*[2pt]
$\displaystyle \partial_t$ & Partial derivative in time\\
$\displaystyle \nabla,\,\nabla_{\mathbf x},\,\nabla_{\mathbf u}$ & Gradient; gradient in position; gradient in momentum\\
$\displaystyle \nabla\cdot$ & Divergence\\
$\displaystyle \Delta,\,\Delta_{\mathbf u}$ & Laplacian; Laplacian in momentum\\
$\displaystyle \frac{d}{dt}$ & Total time derivative along a curve\\
$\displaystyle |\partial\mathcal F|$ & Metric slope of $\mathcal F$ in $(\mathcal P(\Omega),\mathcal W_2)$\\
$\displaystyle \frac{\delta\mathcal F}{\delta q}$ & First variation of a functional $\mathcal F$ at $q$ (Definition~\ref{def:first_variation})\\
$\displaystyle \int f\,d\alpha$ & Integral of $f$ against the measure $\alpha$\\
$\displaystyle \frac{d\alpha}{d\alpha'}$ & Radon--Nikodym derivative\\
$\displaystyle \operatorname{LSE}_j(\cdot)$ & Log-sum-exp over index $j$\\
$\displaystyle \mathcal O(\cdot)$ & Big-O notation\\
$\displaystyle \mathrm{sg}(\cdot)$ & Stop-gradient\\[6pt]

\multicolumn{2}{l}{\bf Probability and Information Theory}\\*[2pt]
$\displaystyle \mathbf x\sim q$ & $\mathbf x$ is distributed according to $q$\\
$\displaystyle \mathbb E_{q}[f]$ & Expectation of $f$ under $q$\\
$\displaystyle \mathbb E_\pi[\mathbf x_{-k}\mid\mathbf x_k]$ & Conditional expectation under a plan $\pi$\\
$\displaystyle \mathrm{Var}_q(h)$ & Variance of $h$ under $q$\\
$\displaystyle \mathcal N(0,\mathbf{I})$ & Standard Gaussian\\
$\displaystyle \operatorname{Law}(\mathbf X)$ & Probability law of a random vector $\mathbf X$\\
$\displaystyle \mathbf X,\,\mathbf Y,\,\mathbf U,\,\mathbf B$ & Random vectors in $\mathbb R^d$\\
$\displaystyle \mathrm{KL}(\alpha\Vert\alpha')$ & Generalized Kullback--Leibler divergence (Definition~\ref{def:generalized_kl})\\
$\displaystyle D_\varphi(\alpha\Vert\alpha')$ & Csisz\'ar $\varphi$-divergence (Definition~\ref{def:phi_divergence})\\
$\displaystyle \chi^2(p\Vert q)$ & Chi-square divergence\\
$\displaystyle \mathcal W_2(\alpha,\alpha')$ & 2-Wasserstein distance\\
$\displaystyle \alpha\ll\alpha'$ & $\alpha$ is absolutely continuous w.r.t.\ $\alpha'$\\
$\displaystyle \alpha\otimes\alpha'$ & Product measure\\
$\displaystyle f_\#\alpha$ & Pushforward of $\alpha$ by a map $f$\\
$\displaystyle \alpha(\Omega)$ & Total mass of $\alpha$\\
$\displaystyle C_{\mathrm{PI}}(q)$ & Poincar\'e constant of $q$\\[6pt]

\multicolumn{2}{l}{\bf Functions}\\*[2pt]
$\displaystyle f:\mathcal A\to\mathcal B$ & Function with domain $\mathcal A$ and codomain $\mathcal B$\\
$\displaystyle \log x$ & Natural logarithm\\
$\displaystyle x^+$ & Positive part, $\max(0,x)$\\
$\displaystyle \varphi'$ & Derivative of a scalar function $\varphi$\\
$\displaystyle \varphi'_\infty$ & Recession constant, $\lim_{r\to\infty}\varphi(r)/r$\\
$\displaystyle \iota$ & Convex indicator of $\{1\}$: $0$ at $1$, $+\infty$ elsewhere\\
$\displaystyle \mathsf R$ & Reflection $\mathsf R(x)=-x$ (Proposition~\ref{prop:non_target_stationary})\\[6pt]

\multicolumn{2}{l}{\bf Distributions and Measures}\\*[2pt]
$\displaystyle p$ & Target (data) distribution, $p=p_{\mathrm{data}}$\\
$\displaystyle p_{\mathrm{ref}}$ & Reference (noise) distribution\\
$\displaystyle q_{\boldsymbol\theta}$ & Model distribution, $(f_{\boldsymbol\theta})_\#p_{\mathrm{ref}}$\\
$\displaystyle q_t,\,q_n$ & Model distribution at time $t$; at iteration $n$\\
$\displaystyle \widehat q,\,\widehat p$ & Empirical mini-batch measures\\
$\displaystyle \alpha,\,\alpha',\,\alpha_1,\,\alpha_2$ & Generic measures (arguments of a transport cost)\\
$\displaystyle \pi,\,\pi^{\alpha_1,\alpha_2}$ & Transport plan; optimal plan for the pair $(\alpha_1,\alpha_2)$\\
$\displaystyle \pi_1,\,\pi_2$ & First and second marginals of $\pi$\\
$\displaystyle \kappa^{(k)}$ & Conditional kernel of $\pi$ given the $k$-th coordinate\\
$\displaystyle \mu^{m,\beta}_t(\mathbf x,\mathbf u)$ & Phase-space density\\
$\displaystyle q^{m,\beta}_t,\,q^\beta_t$ & Positional marginal of $\mu^{m,\beta}_t$; its overdamped limit\\
$\displaystyle q_\mu$ & Positional marginal of a phase-space density $\mu$\\[6pt]

\multicolumn{2}{l}{\bf Generator}\\*[2pt]
$\displaystyle f_{\boldsymbol\theta}$ & Generator network with parameters $\boldsymbol\theta$\\
$\displaystyle f_n$ & Generator at training iteration $n$, $f_{\boldsymbol\theta_n}$\\
$\displaystyle \mathbf z$ & Latent noise, $\mathbf z\sim p_{\mathrm{ref}}$\\
$\displaystyle \mathbf x$ & Generated sample, $f_{\boldsymbol\theta}(\mathbf z)$\\
$\displaystyle \Delta\mathbf x_n$ & Displacement of a sample between iterations\\
$\displaystyle \psi_\ell$ & Frozen feature encoder, block $\ell$\\
$\displaystyle n_\psi$ & Number of feature blocks\\
$\displaystyle \cls$ & Class label\\
$\displaystyle \mathcal X,\,\mathcal Y$ & Mini-batches of generated and real samples\\[6pt]

\multicolumn{2}{l}{\bf Transport}\\*[2pt]
$\displaystyle c(\mathbf x,\mathbf y)$ & Ground cost, $\tfrac12\|\mathbf x-\mathbf y\|^2$\\
$\displaystyle \mathbf C$ & Cost matrix, $C_{ij}=c(\mathbf x_i,\mathbf y_j)$\\
$\displaystyle \mathbf P$ & Coupling matrix on a mini-batch\\
$\displaystyle \mathbf K$ & Log-kernel, $-\mathbf C/\varepsilon$\\
$\displaystyle \varepsilon$ & Entropic regularization strength\\
$\displaystyle \rho$ & Marginal-penalty weight\\
$\displaystyle \tau$ & Relaxation strength, $\rho/(\rho+\varepsilon)$; $\tau\to1$ is balanced\\
$\displaystyle \varphi_1,\,\varphi_2$ & Marginal-penalty generators for source and target\\
$\displaystyle \mathrm{UOT}^{\varphi_1,\varphi_2}_{\varepsilon,\rho}$ & Entropic unbalanced transport cost (Definition~\ref{def:phi_uot})\\
$\displaystyle \mathcal A(\alpha_1,\alpha_2)$ & Symmetrized transport cost\\
$\displaystyle S_\varepsilon,\,S^{U,\varphi}_{\varepsilon,\rho}$ & Balanced Sinkhorn divergence; symmetric unbalanced Sinkhorn divergence\\
$\displaystyle r^{(k)}_{\alpha_1,\alpha_2}$ & Transported-mass ratio of the $k$-th marginal, $d\pi_k/d\alpha_k$\\
$\displaystyle B^{(k)}_{\alpha_1,\alpha_2}$ & Conditional barycenter of the other coordinate given the $k$-th\\
$\displaystyle a_i,\,b_j$ & Sinkhorn scaling factors in Algorithm~\ref{alg:sinkhorn_sf}, entries of $\mathbf a\in\mathbb R^N$, $\mathbf b\in\mathbb R^M$\\[6pt]

\multicolumn{2}{l}{\bf Energies and Forces}\\*[2pt]
$\displaystyle \mathcal F$ & Generic energy functional on $\mathcal P(\Omega)$\\
$\displaystyle \mathsf D(q\Vert p)$ & Generic nonnegative, definite divergence (Section~\ref{sec:wgf})\\
$\displaystyle \mathcal E_p$ & Configurational energy (Definition~\ref{def:configurational_energy})\\
$\displaystyle \mathcal E_p^\beta$ & Finite-temperature energy, $\mathcal E_p+\tfrac1\beta\int q\log q$\\
$\displaystyle \mathcal H_{m,\beta}$ & Kinetic free energy on phase space\\
$\displaystyle \Phi_q$ & First variation of $\mathcal E_p$ at $q$\\
$\displaystyle F^{(k)}_{\alpha_1,\alpha_2}$ & Force on the $k$-th marginal, $r^{(k)}(B^{(k)}-\mathbf x)$\\
$\displaystyle F_p[q]$ & Force induced by $\mathcal E_p$, $-\nabla\Phi_q$\\
$\displaystyle \mathbf v_q$ & Wasserstein gradient-flow velocity\\
$\displaystyle \mathcal D(q)$ & Dissipation, $\int\|\nabla\Phi_q\|^2dq$\\
$\displaystyle A$ & Uniform nondegeneracy constant in Assumption~\ref{ass:longtime}\\[6pt]

\multicolumn{2}{l}{\bf Dynamics}\\*[2pt]
$\displaystyle \mathbf u$ & Momentum (phase-space coordinate)\\
$\displaystyle m,\,\gamma,\,\beta$ & Inertia, friction, inverse temperature\\
$\displaystyle h_n$ & Drifting step size at iteration $n$\\
$\displaystyle \eta$ & Euler step size, $h/\gamma$\\
$\displaystyle T$ & Time horizon\\
$\displaystyle L_{\mathbf x},\,L_{\mathcal W},\,\Lambda$ & Lipschitz constants of the force in position and in law; their sum\\
$\displaystyle K_0$ & Force at the origin, $\|F_p[\delta_{\mathbf 0}](\mathbf 0)\|$\\
$\displaystyle C_0,\,C_{T,\beta},\,C_F$ & Constants in Theorem~\ref{thm:overdamped}\\[6pt]

\multicolumn{2}{l}{\bf Implementation}\\*[2pt]
$\displaystyle N,\,M$ & Number of generated and real samples in a mini-batch\\
$\displaystyle N_{\cls}$ & Classes per training step\\
$\displaystyle N_{\mathrm{pos}},\,N_{\mathrm{neg}}$ & Real and generated samples per class\\
$\displaystyle N_{\mathrm{unc}}$ & Unconditional real samples per step\\
$\displaystyle L$ & Number of Sinkhorn iterations\\
$\displaystyle \mathbf P^{\to},\,\mathbf P^{\leftarrow}$ & Forward and reverse source-fixed plans\\
$\displaystyle \widehat{\mathbf v}^{\to},\,\widehat{\mathbf v}^{\leftarrow}$ & Forward and reverse velocity estimates\\
$\displaystyle \widehat{\mathbf v}_{\cls},\,\widehat{\mathbf v}_{\mathrm{self}},\,\widehat{\mathbf v}_\varnothing$ & Estimated velocities toward class, self, and unconditional batches\\
$\displaystyle \widehat{\mathbf v}^\ell_i$ & Estimated velocity of sample $i$ in feature block $\ell$\\
$\displaystyle w,\,s$ & Guidance weight; reported scale $s=w+1$\\
$\displaystyle \mathcal L(\boldsymbol\theta)$ & Regression loss\\[6pt]

\multicolumn{2}{l}{\bf Proof-Local Symbols}\\*[2pt]
$\displaystyle \mathcal U,\,\Psi^{(k)}$ & Shorthand for $\mathrm{UOT}^{\varphi_1,\varphi_2}$ and its first variation in slot $k$\\
$\displaystyle \mathcal T_t,\,\mathsf S$ & Perturbation diffeomorphism; coordinate swap (Proposition~\ref{prop:uot_force})\\
$\displaystyle \Xi,\,\vartheta_R,\,\nu,\,\sigma$ & Free-energy density, cutoff, growth and decay exponents (Proposition~\ref{prop:vfp_dissipation})\\
$\displaystyle \varkappa_m,\,\mathsf e,\,\mathsf e_\beta$ & Momentum kernel; coupling errors (Theorem~\ref{thm:overdamped}, Corollary~\ref{cor:zero_temp})\\
$\displaystyle \mathcal G,\,\mathcal S,\,\varsigma,\,\lambda,\,C_\Phi$ & Cross and self terms, plan mass, mixing parameter, Hessian bound (Proposition~\ref{prop:non_target_stationary})\\

\end{longtable}
}

\section{Related Work}
\label{sec:related}

\paragraph{Optimal transport and its relaxations.}
Optimal transport lifts a ground cost on samples to a geometry on distributions
\citep{villani2009optimal,santambrogio2015optimal,peyré2019computational}. Entropic
regularization makes it tractable via Sinkhorn iterations \citep{2999792.2999868}
and improves sample complexity from the curse-of-dimensionality rate of plain OT to
a parametric rate \citep{Genevay2018SampleCO,3454287.3454695}; the entropic bias is removed by the Sinkhorn divergence, which is positive, definite, and metrizes weak
convergence \citep{pmlr-v84-genevay18a,Feydy2018InterpolatingBO}. Unbalanced OT
replaces the hard marginal constraints by Csisz\'ar penalties, allowing mass to be
created or destroyed \citep{Liero2015OptimalEP,Chizat2016ScalingAF}, and
\citet{Sjourn2019SinkhornDF} extend Sinkhorn divergences to this setting, proving
positivity for symmetric penalties. In the mini-batch regime, balanced OT is known to
produce couplings that the population plan would never form, and unbalanced
mini-batch OT was proposed as the remedy \citep{pmlr-v108-fatras20a,pmlr-v139-fatras21a}.
Our energy is a source-fixed member of the unbalanced family, debiased so that the
target is stationary (Section~\ref{sec:energy}).

\paragraph{Wasserstein gradient flows.}
Gradient flows of energy functionals in the Wasserstein geometry originate in the
JKO scheme and Otto calculus \citep{doi:10.1137/S0036141096303359,article111}, with the
metric-space theory of \citet{article222}; see \citet{article333}
for an overview. In machine learning they have been instantiated with kernel
discrepancies \citep{Arbel2019MaximumMD,pmlr-v130-mroueh21a,glaser2021kale,hertrich2024generative},
variational and discriminator-based divergences
\citep{pmlr-v97-gao19b,ansari2021refining,pmlr-v162-fan22d}, JKO steps parameterized by
input-convex networks \citep{mokrov2021largescale,alvarez-melis2022optimizing,pmlr-v151-bunne22a},
and unbalanced transport with an inner adversarial solve \citep{3692070.3692414}.
These methods integrate the flow at inference time or through a costly inner
optimization; drifting models instead amortize it into the generator's training
trajectory. Kinetic lifts of Wasserstein flows and their overdamped limits are
classical in PDE theory \citep{article56312} and have been studied as
accelerated sampling schemes \citep{Ma2019IsTA}; we use them to place drifting
inside a Langevin-type hierarchy (Section~\ref{sec:kinetic}).

\paragraph{Optimal transport in generative modeling.}
The Wasserstein distance entered generative modeling as a GAN critic
\citep{pmlr-v70-arjovsky17a,3295222.3295327,salimans2018improving} and as a
Sinkhorn-divergence loss \citep{pmlr-v84-genevay18a}. A second line learns the
transport map itself, through input-convex potentials
\citep{pmlr-v119-makkuva20a,korotin2021wasserstein} or neural dual solvers
\citep{korotin2023neural,rout2022generative,fan2022scalable}. Unbalanced transport
was introduced for robustness to outliers via robust OT \citep{NEURIPS2020_9719a00e}, machine unlearning generative models \citep{choi2026unlearning} and semi-dual UOT maps \citep{3666122.3667962}; these rely on adversarial or
semi-dual formulations and have been demonstrated up to CelebA-HQ. OT couplings
also straighten flow-matching paths \citep{tong2024improving,3618408.3619574}
and underlie Schr\"odinger-bridge samplers \citep{bortoli2021diffusion,shi2023diffusion}. We
use unbalanced transport differently: not as a loss to be minimized by an inner
solver, but as the geometry that defines the velocity of an explicit drifting update.

\paragraph{Drifting models and their interpretation.}
Drifting models \citep{deng2026generative} train a one-step generator by regressing
each sample onto a stop-gradient target displaced along a kernel-based
attraction--repulsion field, reaching 1.54 FID on ImageNet-256. Several works have
since interpreted the field. With Gaussian kernels it is a difference of
kernel-density score estimates, connecting drifting to score matching and to
Wasserstein flows of smoothed divergences
\citep{lai2026unified,cao2026gradient,gretton2026wasserstein};
\citet{franz2026drifting} show the original field is not conservative, and
\citet{caucheteux2026a} give a unifying variational account of generative
Wasserstein flows. Finite-particle rates for conservative and non-conservative
fields are established by \citet{balasubramanian2026finite}. On the algorithmic side,
\citet{he2026sinkhorn} replace one-sided kernel normalization by Sinkhorn coupling,
restoring identifiability; W-Flow \citep{han2026one} realizes drifting as the
Wasserstein gradient flow of the balanced Sinkhorn divergence with a convergence
guarantee for the particle dynamics and state-of-the-art one-step results;
\citet{esteban2026kernel} use kernel gradients, \citet{zhang2026lookahead} regress
onto multi-step drifts, \citet{dumont2026learning} constrain the dynamics to
transport maps, and \citet{feng2026representation,zhang2026distilling} match distributions in
representation space. Drifting has also been applied to speech enhancement, medical
imaging, and policy optimization
\citep{xu2026speech,li2026generative,zhang2026positive}. All of these inherit balanced
transport, or a kernel surrogate of it, as the geometry of the field; to our knowledge,
our work is the first to relax its marginal constraints, showing that a source-fixed
unbalanced geometry improves on it under identical settings on ImageNet.

\paragraph{One-step generation.}
Beyond drifting, one-step generators are obtained by distilling multi-step models
\citep{3618408.3619743,lu2025simplifying,frans2025one}, by learning average
velocities \citep{geng2025mean}, or adversarially
\citep{styleGAN,lin2026adversarial}. Drifting-type methods differ in
requiring neither a teacher nor a discriminator.

\section{Definition}\label{app:def}

\begin{definition}[First variation]
\label{def:first_variation}
Let $\mathcal F:\mathcal P(\Omega)\to\mathbb R$ be a functional. A measurable function
$\Phi_q:\Omega\to\mathbb R$ is called a first variation of $\mathcal F$ at $q$ if, for
every finite signed Borel measure $\omega$ with $\Phi_q\in L^1(|\omega|)$ and
$q+t\omega\in\mathcal P(\Omega)$ for all $t\in[0,t_0)$ and some $t_0>0$,
\begin{equation}
\mathcal F[q+t\omega]=\mathcal F[q]+t\int_\Omega\Phi_q\,d\omega+o(t)
\qquad\text{as } t\to0^+.
\label{eq:first_variation}
\end{equation}
We write $\Phi_q=\frac{\delta\mathcal F}{\delta q}(q)$. Since admissible
perturbations have zero total mass, $\Phi_q$ is determined only up to an additive
constant.
\end{definition}
Applying Definition~\ref{def:first_variation} to $\mathcal E_p$ from
Definition~\ref{def:configurational_energy}, we write
$\Phi_q:=\frac{\delta\mathcal E_p}{\delta q}(q)$ and $F_p[q]:=-\nabla\Phi_q$.

\begin{definition}[The 2-Wasserstein distance]
\label{def:wasserstein_distance}
Let $\mathcal P_2(\mathbb R^d)$ denote the space of probability measures with finite
second moments. For $\alpha,\alpha'\in\mathcal P_2(\mathbb R^d)$, their 2-Wasserstein
distance is
\begin{equation}
\mathcal W_2^2(\alpha,\alpha')
:=\inf_{\pi\in\Pi(\alpha,\alpha')}
\int_{\mathbb R^d\times\mathbb R^d}\|\mathbf x-\mathbf y\|^2\,d\pi(\mathbf x,\mathbf y),
\end{equation}
where $\Pi(\alpha,\alpha')$ denotes the set of couplings whose marginals are $\alpha$
and $\alpha'$.
\end{definition}

\begin{definition}[$\varphi$-divergence]
\label{def:phi_divergence}
Let $\varphi:\mathbb R_+\to\mathbb R\cup\{+\infty\}$ be a convex lower-semicontinuous
function satisfying $\varphi(1)=0$. For two nonnegative finite measures
$\alpha,\alpha'\in\mathcal M_+(\Omega)$, let
\begin{equation}
\alpha=\frac{d\alpha}{d\alpha'}\alpha'+\alpha^\perp
\end{equation}
be the Lebesgue decomposition of $\alpha$ with respect to $\alpha'$. The Csisz\'ar
$\varphi$-divergence from $\alpha$ to $\alpha'$ is defined by
\begin{equation}
D_\varphi(\alpha\Vert\alpha')
:=\int_\Omega\varphi\Big(\frac{d\alpha}{d\alpha'}\Big)d\alpha'
+\varphi'_\infty\,\alpha^\perp(\Omega),
\label{eq:phi_divergence}
\end{equation}
where $\varphi'_\infty:=\lim_{r\to+\infty}\varphi(r)/r$ is the recession constant of
$\varphi$, with the convention $0\cdot(+\infty)=0$.
\end{definition}
In particular, for the convex indicator $\iota$ one has $\varphi'_\infty=+\infty$ and
$D_\iota(\alpha\Vert\alpha')=0$ if $\alpha=\alpha'$ and $+\infty$ otherwise, so the corresponding marginal penalty acts as a hard constraint.

\begin{definition}[Generalized Kullback--Leibler divergence]
\label{def:generalized_kl}
For $\alpha,\alpha'\in\mathcal M_+(\Omega)$ with $\alpha\ll\alpha'$, the generalized
Kullback--Leibler divergence is
\begin{equation}
\mathrm{KL}(\alpha\Vert\alpha')
:=\int_\Omega\log\Big(\frac{d\alpha}{d\alpha'}\Big)d\alpha
-\alpha(\Omega)+\alpha'(\Omega),
\end{equation}
and $\mathrm{KL}(\alpha\Vert\alpha')=+\infty$ whenever $\alpha\not\ll\alpha'$.
\end{definition}

\begin{definition}[Entropic unbalanced optimal transport]
\label{def:phi_uot}
Let $c:\Omega\times\Omega\to\mathbb R_+$ be a nonnegative ground cost, let
$\varepsilon,\rho>0$, and let $\varphi_1,\varphi_2$ be marginal-penalty generators,
each either the convex indicator $\iota$ or a strictly convex
$\varphi\in C^2((0,\infty))$ with $\varphi(1)=0$. For
$\alpha_1,\alpha_2\in\mathcal M_+(\Omega)$,
\begin{equation}
\begin{aligned}
\mathrm{UOT}^{\varphi_1,\varphi_2}_{\varepsilon,\rho}(\alpha_1,\alpha_2)
:=\inf_{\pi\in\mathcal M_+(\Omega^2)}\Big\{
&\int_{\Omega^2}c(\mathbf x,\mathbf y)\,d\pi(\mathbf x,\mathbf y)
+\rho D_{\varphi_1}(\pi_1\Vert\alpha_1)\\
&+\rho D_{\varphi_2}(\pi_2\Vert\alpha_2)
+\varepsilon\,\mathrm{KL}(\pi\Vert\alpha_1\otimes\alpha_2)
\Big\},
\end{aligned}
\label{eq:phi_uot}
\end{equation}
where $\pi_1,\pi_2$ are the marginals of $\pi$. The case $\varphi_k=\iota$ imposes
the hard constraint $\pi_k=\alpha_k$; a $C^2$ generator penalizes deviation from it.
We write $\tau:=\rho/(\rho+\varepsilon)$.
\end{definition}

\begin{definition}[Symmetric unbalanced Sinkhorn divergence \citep{Sjourn2019SinkhornDF}]
\label{def:phi_sinkhorn}
For a single generator $\varphi$ applied to both marginals, the debiased divergence
\begin{equation}
S^{U,\varphi}_{\varepsilon,\rho}(\alpha_1,\alpha_2)
:=
\mathrm{UOT}^{\varphi,\varphi}_{\varepsilon,\rho}(\alpha_1,\alpha_2)
-\tfrac12\mathrm{UOT}^{\varphi,\varphi}_{\varepsilon,\rho}(\alpha_1,\alpha_1)
-\tfrac12\mathrm{UOT}^{\varphi,\varphi}_{\varepsilon,\rho}(\alpha_2,\alpha_2)
+\tfrac{\varepsilon}{2}\bigl(\alpha_1(\Omega)-\alpha_2(\Omega)\bigr)^2
\label{eq:phi_sinkhorn}
\end{equation}
is nonnegative, definite, and metrizes weak convergence when the kernel
$e^{-c/\varepsilon}$ is positive and universal. On $\mathcal P(\Omega)$ the mass
correction vanishes, and for $\varphi=\iota$ it reduces to the Sinkhorn divergence
$S_\varepsilon$ of \citet{Feydy2018InterpolatingBO}.
\end{definition}

\begin{definition}[Configurational energy]
\label{def:configurational_energy}
Let $p\in\mathcal P(\Omega)$ be the fixed target. With the symmetrized cost
$\mathcal A(\alpha_1,\alpha_2):=\tfrac12\big[\mathrm{UOT}^{\varphi_1,\varphi_2}_{\varepsilon,\rho}(\alpha_1,\alpha_2)
+\mathrm{UOT}^{\varphi_1,\varphi_2}_{\varepsilon,\rho}(\alpha_2,\alpha_1)\big]$,
the configurational energy of $q\in\mathcal P(\Omega)$ is
\begin{equation}
\mathcal E_p(q)
:=
\mathcal A(q,p)-\tfrac12\mathcal A(q,q)-\tfrac12\mathcal A(p,p).
\label{eq:configurational_energy}
\end{equation}
For $(\varphi,\varphi)$ this coincides on $\mathcal P(\Omega)$ with
Definition~\ref{def:phi_sinkhorn}. For $(\iota,\mathrm{KL})$ the self-term
$\mathcal A(q,q)$ uses the same asymmetric cost as the cross-term, so that
$\mathcal E_p(p)=0$ and $F_p[p]=\mathbf 0$ (Section~\ref{sec:energy}); this
construction is not covered by the positivity proof of
\citet{Sjourn2019SinkhornDF}, see Assumption~\ref{ass:longtime}.
\end{definition}

\begin{definition}[Transported-mass ratio and conditional barycenter]
\label{def:survival_barycenter}
Let $\pi$ be an optimal plan of
$\mathrm{UOT}^{\varphi_1,\varphi_2}_{\varepsilon,\rho}(\alpha_1,\alpha_2)$ with
$\pi_k\ll\alpha_k$ for $k\in\{1,2\}$, and let
$\pi=\pi_k\otimes\kappa^{(k)}$ be its disintegration with respect to the $k$-th
marginal, where $\kappa^{(k)}(\cdot\mid\mathbf x)$ is a probability kernel. Define
\begin{equation}
r^{(k)}_{\alpha_1,\alpha_2}(\mathbf x):=\frac{d\pi_k}{d\alpha_k}(\mathbf x),
\qquad
B^{(k)}_{\alpha_1,\alpha_2}(\mathbf x):=\int_\Omega\mathbf x_{-k}\,\kappa^{(k)}(d\mathbf x_{-k}\mid\mathbf x),
\label{eq:survival_barycenter}
\end{equation}
the transported-mass ratio of the $k$-th marginal and the conditional barycenter of
the other coordinate given $\mathbf x_k=\mathbf x$. If $\varphi_k=\iota$ then
$r^{(k)}\equiv1$.
\end{definition}

\section{Assumptions}

% For $p\ll q$, define the chi-square divergence
% \begin{equation}
%     \chi^2(p\Vert q)
%     :=
%     \int_\Omega
%     \left(
%         \frac{dp}{dq}-1
%     \right)^2dq.
%     \label{eq:chi_square}
% \end{equation}
% We say that $q$ satisfies a Poincar\'e inequality with constant
% $C_{\mathrm{PI}}(q)$ if
% \begin{equation}
%     \operatorname{Var}_q(h)
%     \leq
%     C_{\mathrm{PI}}(q)
%     \int_\Omega\|\nabla h(x)\|^2\,dq(x)
%     \label{eq:poincare_q}
% \end{equation}
% for every admissible $h$, where
% \begin{equation}
%     \operatorname{Var}_q(h)
%     :=
%     \int_\Omega
%     \left|
%         h-\int_\Omega h\,dq
%     \right|^2dq.
% \end{equation}

\begin{assumption}[Regularity of the transport cost]\label{ass:uot_regularity}
Write $\mathcal U:=\mathrm{UOT}^{\varphi_1,\varphi_2}_{\varepsilon,\rho}$. For every
pair $(\alpha_1,\alpha_2)$ under consideration, $\mathcal U(\alpha_1,\alpha_2)$ is
finite, its optimal plan has finite second moment, and for each slot $k\in\{1,2\}$
the first variation
$\Psi^{(k)}_{\alpha_1,\alpha_2}:=\frac{\delta\,\mathcal U}{\delta\alpha_k}(\alpha_1,\alpha_2)$
admits a $C^1$ representative satisfying the transport chain rule
\begin{equation}\label{eq:uot_transport_chain_rule}
\mathcal U\big|_{\alpha_k\to(\mathcal T_t)_\#\alpha_k}
=\mathcal U(\alpha_1,\alpha_2)
+t\int_\Omega\nabla\Psi^{(k)}_{\alpha_1,\alpha_2}\cdot\boldsymbol\xi\,d\alpha_k
+o(t)\qquad\text{as }t\to0,
\end{equation}
for every $\boldsymbol\xi\in C_c^\infty(\operatorname{int}\Omega;\mathbb R^d)$, where
$\mathcal T_t:=\mathrm{id}+t\boldsymbol\xi$ and the left-hand side denotes
$\mathcal U$ with only its $k$-th argument replaced by $(\mathcal T_t)_\#\alpha_k$.
\end{assumption}

\begin{assumption}[Variational regularity]\label{ass:kinetic}
$\mathcal E_p$ admits a first variation $\Phi_q\in C^1(\mathbb R^d)$ with
$F_p[q]=-\nabla_{\mathbf x}\Phi_q$, and along every smooth curve of densities
$t\mapsto q_t$ under consideration the chain rule
$\frac{d}{dt}\mathcal E_p(q_t)=\int_{\mathbb R^d}\Phi_{q_t}\,\partial_tq_t\,d\mathbf x$
holds with an absolutely convergent integral. Moreover, with
$\mathbf z=(\mathbf x,\mathbf u)$, there exist $\nu\ge0$ and $\sigma>2d+2\nu$ such
that, uniformly on compact time intervals,
\begin{equation}
|\Phi_q|+\|\nabla\Phi_q\|+\|F_p[q]\|\le C(1+\|\mathbf x\|)^{\nu},\qquad
|\log\mu_t^{m,\beta}|+\|\nabla_{\mathbf z}\log\mu_t^{m,\beta}\|\le C(1+\|\mathbf z\|)^{\nu},
\end{equation}
\begin{equation}
\mu_t^{m,\beta}+\|\nabla_{\mathbf z}\mu_t^{m,\beta}\|+|\partial_t\mu_t^{m,\beta}|
\le C(1+\|\mathbf z\|)^{-\sigma}.
\end{equation}
\end{assumption}

% \paragraph{Uniform force stability.}
% For fixed $p$, write $F[\mu]:=F_p[\mu]$.
% Assume that $F$ is Borel measurable and that there exist
% constants $L_x,L_\mu\ge0$, independent of $m$ and $\beta$,
% such that, for all $\mu,\nu\in\mathcal P_2(\mathbb R^n)$
% and $\mathbf x,\mathbf y\in\mathbb R^n$,
% \begin{equation}\label{eq:force_position_stability}
% \|F[\mu](\mathbf x)-F[\mu](\mathbf y)\|
% \le L_x\|\mathbf x-\mathbf y\|,
% \end{equation}
% and
% \begin{equation}\label{eq:force_measure_stability}
% \left(
% \int_{\mathbb R^n}
% \|F[\mu](\mathbf x)-F[\nu](\mathbf x)\|^2
% \,d\nu(\mathbf x)
% \right)^{1/2}
% \le L_\mu\mathcal W_2(\mu,\nu).
% \end{equation}
% Set
% \begin{equation}
% \Lambda:=L_x+L_\mu,
% \qquad
% K:=\|F[\delta_{\mathbf0}](\mathbf0)\|<\infty.
% \end{equation}

\begin{assumption}[Force stability and initial data]
\label{ass:overdamped}
Fix $\gamma>0$ and the target $p$, and write $F[\alpha]:=F_p[\alpha]$, assumed
independent of $m$ and $\beta$ and jointly Borel measurable in $(\alpha,\mathbf x)$.
There exist $L_{\mathbf x},L_{\mathcal W}\ge0$ such that, for all
$\alpha,\alpha'\in\mathcal P_2(\mathbb R^d)$ and $\mathbf x,\mathbf y\in\mathbb R^d$,
\begin{equation}\label{eq:ks_stability}
\|F[\alpha](\mathbf x)-F[\alpha](\mathbf y)\|\le L_{\mathbf x}\|\mathbf x-\mathbf y\|,
\qquad
\Big(\int\|F[\alpha]-F[\alpha']\|^2\,d\alpha'\Big)^{1/2}
\le L_{\mathcal W}\,\mathcal W_2(\alpha,\alpha');
\end{equation}
set $\Lambda:=L_{\mathbf x}+L_{\mathcal W}$ and
$K_0:=\|F[\delta_{\mathbf0}](\mathbf0)\|<\infty$. The initial laws
$\mu_0^{m,\beta}\in\mathcal P_2(\mathbb R^{2d})$ share the positional marginal
$q_0\in\mathcal P_2(\mathbb R^d)$ and satisfy
$\int\|\mathbf u\|^2\,d\mu_0^{m,\beta}\le C_0m$ for $0<m\le1$, with $C_0$
independent of $m$.
\end{assumption}

\begin{assumption}[Long-time convergence]\label{ass:longtime}
There exists a $\mathcal W_2$-compact set
$\mathcal K\subset\mathcal P_2(\Omega)$ containing $p$
and the trajectory $(q_t)_{t\ge0}$ such that
$\mathcal E_p$ is lower semicontinuous on $\mathcal K$
and has $p$ as its unique minimizer there.
Moreover, for some $A>0$,
\begin{equation}
\mathcal E_p(q_t)^2
\le A\,\mathcal D(q_t)
\qquad\text{for a.e. }t\ge0.
\end{equation}
\end{assumption}

\begin{lemma}[Sufficient condition for Assumption~\ref{ass:longtime}(ii)]\label{lem:ed}
Suppose $\mathcal E_p$ is convex along linear mixtures, $p\ll q$ with
$\frac{dp}{dq}\in L^2(q)$, and $q$ satisfies the Poincar\'e inequality
$\mathrm{Var}_q(h)\le C_{\mathrm{PI}}(q)\int\|\nabla h\|^2dq$ for admissible $h$. Then
\begin{equation}
\mathcal E_p(q)^2\le C_{\mathrm{PI}}(q)\,\chi^2(p\Vert q)\,\mathcal D(q).
\end{equation}
In particular, Assumption~\ref{ass:longtime}(ii) holds with
$A=\operatorname{ess\,sup}_{t\ge0}C_{\mathrm{PI}}(q_t)\,\chi^2(p\Vert q_t)$ whenever the
right-hand side is finite.
\end{lemma}
\begin{proof}
Convexity and $\mathcal E_p(p)=0$ give
$\mathcal E_p(q)\le\int\Phi_q\,d(q-p)=\int(\Phi_q-\bar\Phi_q)\big(1-\tfrac{dp}{dq}\big)\,dq$
with $\bar\Phi_q:=\int\Phi_q\,dq$, the constant dropping out because both measures
have unit mass. Cauchy--Schwarz bounds this by
$\sqrt{\mathrm{Var}_q(\Phi_q)}\,\sqrt{\chi^2(p\Vert q)}$, and the Poincar\'e
inequality gives $\mathrm{Var}_q(\Phi_q)\le C_{\mathrm{PI}}(q)\int\|\nabla\Phi_q\|^2dq
=C_{\mathrm{PI}}(q)\,\mathcal D(q)$.
\end{proof}

\section{Proofs}

\textbf{Proposition 1} (Force).
Under Assumption~\ref{ass:uot_regularity}, assume that all optimal
plans appearing below are unique.
For a pair $(\alpha_1,\alpha_2)$ with optimal plan $\pi$ of
$\mathrm{UOT}^{\varphi_1,\varphi_2}(\alpha_1,\alpha_2)$ and a slot
$k\in\{1,2\}$, let
$r^{(k)}_{\alpha_1,\alpha_2}:=d\pi_k/d\alpha_k$ and
$B^{(k)}_{\alpha_1,\alpha_2}(\mathbf x)
:=\mathbb E_\pi[\mathbf x_{-k}\mid\mathbf x_k=\mathbf x]$,
with $r^{(k)}\equiv1$ if $\varphi_k=\iota$ and $r^{(k)}$
positive and $C^1$ otherwise, and define the force on the
$k$-th marginal
\begin{equation}\label{eq:slot_force2}
F^{(k)}_{\alpha_1,\alpha_2}(\mathbf x):=r^{(k)}_{\alpha_1,\alpha_2}(\mathbf x)
\bigl(B^{(k)}_{\alpha_1,\alpha_2}(\mathbf x)-\mathbf x\bigr).
\end{equation}
Then $F^{(k)}_{\alpha_1,\alpha_2}$ is the negative spatial gradient
of the first variation of
$\mathrm{UOT}^{\varphi_1,\varphi_2}(\alpha_1,\alpha_2)$
with respect to $\alpha_k$, and
\begin{equation}\label{eq:force_general2}
F_p[q]
=
\tfrac12\bigl[F^{(1)}_{q,p}+F^{(2)}_{p,q}\bigr]
-\tfrac12\bigl[F^{(1)}_{q,q}+F^{(2)}_{q,q}\bigr].
\end{equation}
For $(\varphi,\varphi)$ this is
$r_{q,p}(B_{q,p}-\mathbf x)-r_{q,q}(B_{q,q}-\mathbf x)$;
for $(\iota,\iota)$ it is $B_{q,p}-B_{q,q}$, the W-Flow field
\eqref{eq:wflow-velocity}; for $(\iota,\mathrm{KL})$,
$r^{(1)}\equiv1$, $r^{(2)}>0$, and
$\int_\Omega r^{(2)}\,d\alpha_2=1$.
All force identities hold almost everywhere with respect to
the corresponding input measure.

\begin{proof}
Write
\begin{equation}
\mathcal U(\alpha_1,\alpha_2)
:=\mathrm{UOT}^{\varphi_1,\varphi_2}(\alpha_1,\alpha_2),
\end{equation}
and denote its first variation in the $k$-th argument by
\begin{equation}
\Psi^{(k)}_{\alpha_1,\alpha_2}
:=\frac{\delta\,\mathcal U}{\delta\alpha_k}(\alpha_1,\alpha_2).
\end{equation}

We first establish the force identity for the first argument. Let
$\boldsymbol\xi\in C_c^\infty(\Omega;\mathbb R^d)$ and, for sufficiently small $|t|$,
define the diffeomorphism
\begin{equation}
\mathcal T_t(\mathbf x):=\mathbf x+t\,\boldsymbol\xi(\mathbf x).
\end{equation}
Set
\begin{equation}
\alpha_{1,t}:=(\mathcal T_t)_\#\alpha_1,
\qquad
\pi_t:=(\mathcal T_t,\mathrm{id})_\#\pi.
\end{equation}
The marginals of $\pi_t$ are
\begin{equation}
(\pi_t)_1=(\mathcal T_t)_\#\pi_1,
\qquad
(\pi_t)_2=\pi_2.
\end{equation}

Csisz\'ar divergences are invariant under simultaneous pushforward of both arguments
by a measurable bijection. Consequently,
\begin{equation}
D_{\varphi_1}\bigl((\pi_t)_1\Vert\alpha_{1,t}\bigr)=D_{\varphi_1}(\pi_1\Vert\alpha_1).
\end{equation}
This identity also holds for the hard marginal constraint. Moreover, since
\begin{equation}
\alpha_{1,t}\otimes\alpha_2=(\mathcal T_t,\mathrm{id})_\#(\alpha_1\otimes\alpha_2),
\end{equation}
the entropic regularization satisfies
\begin{equation}
\mathrm{KL}\bigl(\pi_t\Vert\alpha_{1,t}\otimes\alpha_2\bigr)
=
\mathrm{KL}\bigl(\pi\Vert\alpha_1\otimes\alpha_2\bigr).
\end{equation}
The second marginal penalty is unchanged as well.

Using $\pi_t$ as a competitor for $\mathcal U(\alpha_{1,t},\alpha_2)$ and expanding
the quadratic cost therefore yields
\begin{align}
\mathcal U(\alpha_{1,t},\alpha_2)-\mathcal U(\alpha_1,\alpha_2)
&\le
\int_{\Omega^2}
\bigl[c(\mathcal T_t(\mathbf x),\mathbf y)-c(\mathbf x,\mathbf y)\bigr]\,d\pi(\mathbf x,\mathbf y)\\
&=
t\int_{\Omega^2}(\mathbf x-\mathbf y)\cdot\boldsymbol\xi(\mathbf x)\,d\pi(\mathbf x,\mathbf y)
+\frac{t^2}{2}\int_\Omega\|\boldsymbol\xi(\mathbf x)\|^2\,d\pi_1(\mathbf x).
\end{align}

Equality holds at $t=0$. Since the left-hand side is differentiable at $t=0$,
dividing by $t$ and taking the limits from both sides gives
\begin{equation}
\left.\frac{d}{dt}\mathcal U(\alpha_{1,t},\alpha_2)\right|_{t=0}
=
\int_{\Omega^2}(\mathbf x-\mathbf y)\cdot\boldsymbol\xi(\mathbf x)\,d\pi(\mathbf x,\mathbf y).
\end{equation}
Disintegrating $\pi$ with respect to $\pi_1$ and using
$d\pi_1=r^{(1)}_{\alpha_1,\alpha_2}\,d\alpha_1$, we obtain
\begin{equation}
\left.\frac{d}{dt}\mathcal U(\alpha_{1,t},\alpha_2)\right|_{t=0}
=
\int_\Omega
r^{(1)}_{\alpha_1,\alpha_2}(\mathbf x)
\bigl(\mathbf x-B^{(1)}_{\alpha_1,\alpha_2}(\mathbf x)\bigr)
\cdot\boldsymbol\xi(\mathbf x)\,d\alpha_1(\mathbf x).
\end{equation}
On the other hand, the transport chain rule gives
\begin{equation}
\left.\frac{d}{dt}\mathcal U(\alpha_{1,t},\alpha_2)\right|_{t=0}
=
\int_\Omega
\nabla\Psi^{(1)}_{\alpha_1,\alpha_2}(\mathbf x)
\cdot\boldsymbol\xi(\mathbf x)\,d\alpha_1(\mathbf x).
\end{equation}
Since $\boldsymbol\xi$ is arbitrary, it follows that
\begin{equation}
-\nabla\Psi^{(1)}_{\alpha_1,\alpha_2}(\mathbf x)
=
r^{(1)}_{\alpha_1,\alpha_2}(\mathbf x)
\bigl(B^{(1)}_{\alpha_1,\alpha_2}(\mathbf x)-\mathbf x\bigr)
=
F^{(1)}_{\alpha_1,\alpha_2}(\mathbf x)
\end{equation}
for $\alpha_1$-almost every $\mathbf x$. Perturbing the second coordinate instead
proves the same identity for $k=2$. Thus, for either slot,
\begin{equation}
-\nabla\frac{\delta\,\mathcal U}{\delta\alpha_k}(\alpha_1,\alpha_2)
=
F^{(k)}_{\alpha_1,\alpha_2}.
\end{equation}

Next, the definition of $\mathcal A$ implies
\begin{equation}
\mathcal E_p(q)
=
\frac12\mathcal U(q,p)+\frac12\mathcal U(p,q)
-\frac12\mathcal U(q,q)-\frac12\mathcal U(p,p).
\end{equation}
Applying the chain rule to both occurrences of $q$ in $\mathcal U(q,q)$, we obtain,
up to an additive spatial constant,
\begin{equation}
\Phi_q
=
\frac12\left[
\Psi^{(1)}_{q,p}+\Psi^{(2)}_{p,q}
\right]
-\frac12\left[
\Psi^{(1)}_{q,q}+\Psi^{(2)}_{q,q}
\right].
\end{equation}
Taking the negative spatial gradient yields
\begin{equation}
F_p[q]
=
\frac12\left[
F^{(1)}_{q,p}+F^{(2)}_{p,q}
\right]
-\frac12\left[
F^{(1)}_{q,q}+F^{(2)}_{q,q}
\right],
\end{equation}
which proves \eqref{eq:force_general}.

Suppose now that $\varphi_1=\varphi_2=\varphi$. Let
$\mathsf S(\mathbf x,\mathbf y):=(\mathbf y,\mathbf x)$. By symmetry of the cost and
of the marginal penalties, together with uniqueness of the optimal plans,
\begin{equation}
\pi^{p,q}=\mathsf S_\#\pi^{q,p},
\qquad
\pi^{q,q}=\mathsf S_\#\pi^{q,q}.
\end{equation}
Hence
\begin{equation}
F^{(2)}_{p,q}=F^{(1)}_{q,p},
\qquad
F^{(2)}_{q,q}=F^{(1)}_{q,q}.
\end{equation}
Writing $r_{\alpha,\alpha'}:=r^{(1)}_{\alpha,\alpha'}$ and
$B_{\alpha,\alpha'}:=B^{(1)}_{\alpha,\alpha'}$, we conclude that
\begin{equation}
F_p[q](\mathbf x)
=
r_{q,p}(\mathbf x)\bigl(B_{q,p}(\mathbf x)-\mathbf x\bigr)
-
r_{q,q}(\mathbf x)\bigl(B_{q,q}(\mathbf x)-\mathbf x\bigr).
\end{equation}
In the balanced case $(\iota,\iota)$, both marginal density ratios are identically
one, and this reduces to
\[
F_p[q](\mathbf x)=B_{q,p}(\mathbf x)-B_{q,q}(\mathbf x),
\]
as claimed.

Finally, in the source-fixed case $(\iota,\mathrm{KL})$, the hard constraint gives
$\pi_1=\alpha_1$, hence $r^{(1)}\equiv1$. Since both input measures are
probabilities,
\[
\int_\Omega r^{(2)}\,d\alpha_2
=
\pi_2(\Omega)
=
\pi_1(\Omega)
=
\alpha_1(\Omega)
=
1.
\]
Together with the assumed positivity of $r^{(2)}$, this establishes the final
assertion.
\end{proof}

\textbf{Proposition 2} (Free-energy dissipation).
Under Assumption~\ref{ass:kinetic}, every sufficiently regular solution of
\eqref{eq:uot_vfp} satisfies
$\frac{d}{dt}\mathcal H_{m,\beta}(\mu_t^{m,\beta})
=-\gamma\int\mu_t^{m,\beta}\big\|\frac{\mathbf u}{m}
+\frac1\beta\nabla_{\mathbf u}\log\mu_t^{m,\beta}\big\|^2\le0$.

\begin{proof}
Write $\mu=\mu_t^{m,\beta}$, $q=q_t^{m,\beta}$,
$\Phi=\delta\mathcal E_p/\delta q(q)$, and
$F=F_p[q]=-\nabla_{\mathbf x}\Phi$.
All integrals below are over $\mathbb R^{2n}$ unless otherwise
specified, with $d\mathbf z=d\mathbf x\,d\mathbf u$.

Define
\begin{equation}
\Xi(\mathbf x,\mathbf u)
:=
\Phi(\mathbf x)
+\frac{\|\mathbf u\|^2}{2m}
+\frac1\beta(1+\log\mu(\mathbf x,\mathbf u)).
\end{equation}
Then
\begin{equation}
\nabla_{\mathbf u}\Xi
=
\frac{\mathbf u}{m}
+\frac1\beta\nabla_{\mathbf u}\log\mu,
\end{equation}
and \eqref{eq:uot_vfp} can be written as
\begin{equation}
\partial_t\mu
=
-\nabla_{\mathbf x}\cdot\left(\frac{\mathbf u}{m}\mu\right)
-\nabla_{\mathbf u}\cdot(F\mu)
+\gamma\nabla_{\mathbf u}\cdot
\bigl(\mu\nabla_{\mathbf u}\Xi\bigr).
\end{equation}

By Assumption~\ref{ass:kinetic}, with constants allowed to change
from line to line,
\begin{equation}
|\Xi|+\|\nabla_{\mathbf z}\Xi\|
+\|F\|+\frac{\|\mathbf u\|}{m}
\le C(1+\|\mathbf z\|)^{\nu}.
\end{equation}
In particular,
\begin{equation}\label{eq:kinetic_integrability_bound}
|\Xi\,\partial_t\mu|
+\mu\left(\frac{\|\mathbf u\|}{m}+\|F\|\right)
 \|\nabla_{\mathbf z}\Xi\|
+\mu\|\nabla_{\mathbf u}\Xi\|^2\le C(1+\|\mathbf z\|)^{2\nu-\sigma}
\in L^1(\mathbb R^{2d}),
\end{equation}
since $\sigma>2d+2\nu$.

To justify integration by parts on the unbounded phase space,
choose $\vartheta\in C_c^\infty(\mathbb R^{2d})$ such that
$0\le\vartheta\le1$, $\vartheta=1$ on $B_1$, and
$\operatorname{supp}\vartheta\subset B_2$.
For $R\ge1$, set
\begin{equation}
\vartheta_R(\mathbf z):=\vartheta(\mathbf z/R),
\qquad
A_R:=B_{2R}\setminus B_R,
\end{equation}
where $B_R$ is the ball of radius $R$ in $\mathbb R^{2d}$.
Then
\begin{equation}
\operatorname{supp}\nabla_{\mathbf z}\vartheta_R
\subset A_R,
\qquad
\|\nabla_{\mathbf z}\vartheta_R\|_{L^\infty}
\le \frac{C}{R}.
\end{equation}

Multiplying the equation by $\vartheta_R\Xi$ and integrating by parts
with compact support gives
\begin{equation}\label{eq:kinetic_cutoff_identity}
\int \vartheta_R\Xi\,\partial_t\mu\,d\mathbf z
=
\int \vartheta_R\mu
\left(
\frac{\mathbf u}{m}\cdot\nabla_{\mathbf x}\Xi
+
F\cdot\nabla_{\mathbf u}\Xi
\right)\,d\mathbf z
-\gamma\int \vartheta_R\mu
\|\nabla_{\mathbf u}\Xi\|^2\,d\mathbf z
+\mathcal R_R,
\end{equation}
where
\begin{equation}
\mathcal R_R
:=
\int \mu\Xi\left(
\frac{\mathbf u}{m}\cdot\nabla_{\mathbf x}\vartheta_R
+
F\cdot\nabla_{\mathbf u}\vartheta_R
-
\gamma\nabla_{\mathbf u}\Xi
\cdot\nabla_{\mathbf u}\vartheta_R
\right)\,d\mathbf z.
\end{equation}
The decay and growth bounds imply
\begin{equation}\label{eq:kinetic_cutoff_error}
|\mathcal R_R|
\le
\frac{C}{R}
\int_{A_R}(1+\|\mathbf z\|)^{2\nu-\sigma}\,d\mathbf z
\le
C R^{2d+2\nu-\sigma-1}
\longrightarrow0
\qquad\text{as} \quad R\to\infty.
\end{equation}
Here we used $|A_R|\le C R^{2d}$.

It remains to treat the conservative contribution.
Since $F=-\nabla_{\mathbf x}\Phi$, its integrand satisfies
the pointwise identity
\begin{equation}
\mu\left(
\frac{\mathbf u}{m}\cdot\nabla_{\mathbf x}\Xi
+
F\cdot\nabla_{\mathbf u}\Xi
\right)
=
\frac1\beta\left(
\frac{\mathbf u}{m}\cdot\nabla_{\mathbf x}\mu
+
F\cdot\nabla_{\mathbf u}\mu
\right).
\end{equation}
Moreover,
$\nabla_{\mathbf x}\cdot(\mathbf u/m)=0$ and
$\nabla_{\mathbf u}\cdot F=0$.
Another compactly supported integration by parts therefore yields
\begin{equation}
\int \vartheta_R\mu
\left(
\frac{\mathbf u}{m}\cdot\nabla_{\mathbf x}\Xi
+
F\cdot\nabla_{\mathbf u}\Xi
\right)\,d\mathbf z=
-\frac1\beta\int \mu
\left(
\frac{\mathbf u}{m}\cdot\nabla_{\mathbf x}\vartheta_R
+
F\cdot\nabla_{\mathbf u}\vartheta_R
\right)\,d\mathbf z.
\end{equation}
Its absolute value is bounded by
\begin{equation}
\frac{C}{R}\int_{A_R}
(1+\|\mathbf z\|)^{\nu-\sigma}\,d\mathbf z
\le C R^{2d+\nu-\sigma-1}
\longrightarrow0.
\end{equation}

Passing to the limit $R\to\infty$ in
\eqref{eq:kinetic_cutoff_identity}, using dominated convergence
and \eqref{eq:kinetic_integrability_bound}, gives
\begin{equation}
\int \Xi\,\partial_t\mu\,d\mathbf z
=
-\gamma\int \mu\|\nabla_{\mathbf u}\Xi\|^2\,d\mathbf z.
\end{equation}

Finally, the assumed chain rule for $\mathcal E_p$, together
with differentiation under the integral sign for the kinetic
and entropy terms, gives
\begin{equation}
\frac{d}{dt}\mathcal H_{m,\beta}(\mu)
=
\int \Xi\,\partial_t\mu\,d\mathbf z.
\end{equation}
The latter differentiations are justified by the uniform
integrable bounds
\begin{equation}
\left(
1+\|\mathbf u\|^2+|\log\mu|
\right)|\partial_t\mu|
\le C(1+\|\mathbf z\|)^{\nu-\sigma}
\in L^1(\mathbb R^{2d}).
\end{equation}
Consequently,
\begin{equation}
\frac{d}{dt}\mathcal H_{m,\beta}(\mu_t^{m,\beta})
=
-\gamma\int_{\mathbb R^{2d}}
\mu_t^{m,\beta}
\left\|
\frac{\mathbf u}{m}
+\frac1\beta\nabla_{\mathbf u}\log\mu_t^{m,\beta}
\right\|^2
\,d\mathbf x\,d\mathbf u
\le0.
\end{equation}
\end{proof}

\textbf{Theorem 1} (Quantitative Kramers--Smoluchowski limit). Fix $\gamma,\beta,T>0$
and assume Assumption~\ref{ass:overdamped}. Suppose that the initial laws share a
positional marginal $q_0\in\mathcal P_2(\mathbb R^d)$ and satisfy
$\int\|\mathbf u\|^2\,d\mu_0^{m,\beta}\le C_0m$ uniformly for $0<m\le1$. Then
$\sup_{t\in[0,T]}\mathcal W_2(q_t^{m,\beta},q_t^\beta)\le C_{T,\beta}\sqrt m$, where
$C_{T,\beta}$ is independent of $m$ and $q^\beta\in C([0,T];\mathcal P_2(\mathbb R^d))$
is the unique weak solution of
\begin{equation}
\partial_t q_t^\beta
=
\frac1\gamma\nabla\cdot
\bigl(q_t^\beta\nabla\Phi_{q_t^\beta}\bigr)
+\frac1{\gamma\beta}\Delta q_t^\beta,
\qquad q_0^\beta=q_0.
\end{equation}

\begin{proof}
For an $\mathbb R^d$-valued random variable $\mathbf X$, we denote its probability law
by $\operatorname{Law}(\mathbf X)$, namely,
\begin{equation}
\operatorname{Law}(\mathbf X)(\Gamma)
:=
\mathbb P(\mathbf X\in\Gamma),
\qquad
\Gamma\subset\mathbb R^d \text{ Borel}.
\end{equation}
For square-integrable random vectors, we write
\begin{equation}
\|\mathbf X\|_{L^2}
:=
\left(\mathbb E\|\mathbf X\|^2\right)^{1/2}.
\end{equation}
For any square-integrable random vectors $\mathbf X,\mathbf Y$ defined on the same
probability space, the force stability assumptions give
\begin{equation}\label{eq:ks_lifted_stability}
\begin{aligned}
\|F[\operatorname{Law}(\mathbf X)](\mathbf X)-F[\operatorname{Law}(\mathbf Y)](\mathbf Y)\|_{L^2}
&\le
L_{\mathbf x}\|\mathbf X-\mathbf Y\|_{L^2}
+L_{\mathcal W}\,\mathcal W_2\big(\operatorname{Law}(\mathbf X),\operatorname{Law}(\mathbf Y)\big)\\
&\le\Lambda\|\mathbf X-\mathbf Y\|_{L^2}.
\end{aligned}
\end{equation}
Taking $\mathbf Y\equiv\mathbf0$ also yields
\begin{equation}\label{eq:ks_growth}
\|F[\operatorname{Law}(\mathbf X)](\mathbf X)\|_{L^2}
\le K_0+\Lambda\|\mathbf X\|_{L^2}.
\end{equation}
These Lipschitz and growth estimates ensure well-posedness of the following
McKean--Vlasov equations by Picard iteration.

Choose $(\mathbf X_0,\mathbf U_0^m)$ with law $\mu_0^{m,\beta}$, and let $\mathbf B_t$
be a standard $d$-dimensional Brownian motion independent of the initial data. Define
\begin{equation}\label{eq:ks_kinetic_sde}
\begin{aligned}
d\mathbf X_t^m
&=\frac{\mathbf U_t^m}{m}\,dt,\\
d\mathbf U_t^m
&=F[q_t^{m,\beta}](\mathbf X_t^m)\,dt
-\frac{\gamma}{m}\mathbf U_t^m\,dt
+\sqrt{\frac{2\gamma}{\beta}}\,d\mathbf B_t,
\end{aligned}
\end{equation}
where $q_t^{m,\beta}=\operatorname{Law}(\mathbf X_t^m)$. The joint law of
$(\mathbf X_t^m,\mathbf U_t^m)$ solves \eqref{eq:uot_vfp}.

Using the same Brownian motion and the same initial position, define
\begin{equation}\label{eq:ks_overdamped_sde}
d\mathbf Y_t =\frac1\gamma F[q_t^\beta](\mathbf Y_t)\,dt
+\sqrt{\frac2{\gamma\beta}}\,d\mathbf B_t,
\qquad
\mathbf Y_0=\mathbf X_0,
\qquad
q_t^\beta=\operatorname{Law}(\mathbf Y_t).
\end{equation}
Its law is the unique weak solution in $C([0,T];\mathcal P_2(\mathbb R^d))$ of
\begin{equation}
\partial_tq_t^\beta
=
-\frac1\gamma\nabla\cdot
\bigl(q_t^\beta F[q_t^\beta]\bigr)
+\frac1{\gamma\beta}\Delta q_t^\beta.
\end{equation}
Since $F[q]=-\nabla\Phi_q$, this is \eqref{eq:finite_temp}, interpreted in the weak
sense.

\paragraph{An overdamped moment bound.}
Let $M(t):=\mathbb E\|\mathbf Y_t\|^2$. By It\^o's formula and \eqref{eq:ks_growth},
\begin{equation}
M'(t)
=
\frac2\gamma
\mathbb E\bigl[\mathbf Y_t\cdot F[q_t^\beta](\mathbf Y_t)\bigr]
+\frac{2n}{\gamma\beta}\le
\frac{1+2\Lambda}{\gamma}M(t)
+\frac{K_0^2}{\gamma}
+\frac{2n}{\gamma\beta}.
\end{equation}
Gronwall's inequality gives $\sup_{0\le t\le T}M(t)<\infty$. Consequently,
\begin{equation}\label{eq:ks_force_moment}
C_F:= \sup_{0\le t\le T} \mathbb E\|F[q_t^\beta](\mathbf Y_t)\|^2 <\infty.
\end{equation}
The constant $C_F$ is independent of $m$, because the overdamped process does not
depend on $m$.

\paragraph{The small-mass estimate.}
Set
\begin{equation}
\varkappa_m(t):=e^{-\gamma t/m}.
\end{equation}
Solving the linear momentum equation and integrating in time gives
\begin{equation}\label{eq:ks_mild_position}
\begin{aligned}
\mathbf X_t^m
&=
\mathbf X_0+\frac{1-\varkappa_m(t)}{\gamma}\mathbf U_0^m\\
&\quad+
\frac1\gamma\int_0^t
\bigl(1-\varkappa_m(t-s)\bigr)
F[q_s^{m,\beta}](\mathbf X_s^m)\,ds + \sqrt{\frac2{\gamma\beta}}
\int_0^t\bigl(1-\varkappa_m(t-s)\bigr)\,d\mathbf B_s.
\end{aligned}
\end{equation}
Define
\begin{equation}
\mathbf D_t:=\mathbf X_t^m-\mathbf Y_t,
\qquad
\Delta F_s
:=
F[q_s^{m,\beta}](\mathbf X_s^m)-F[q_s^\beta](\mathbf Y_s).
\end{equation}
Subtracting the integral equation for $\mathbf Y_t$ from \eqref{eq:ks_mild_position}
yields
\begin{equation}\label{eq:ks_difference}
\begin{aligned}
\mathbf D_t
&=
\frac{1-\varkappa_m(t)}{\gamma}\mathbf U_0^m\\
&\quad+
\frac1\gamma\int_0^t
\bigl(1-\varkappa_m(t-s)\bigr)\Delta F_s\,ds-
\frac1\gamma\int_0^t
\varkappa_m(t-s)F[q_s^\beta](\mathbf Y_s)\,ds-
\sqrt{\frac2{\gamma\beta}}
\int_0^t \varkappa_m(t-s)\,d\mathbf B_s.
\end{aligned}
\end{equation}

Write $\mathsf e(t):=\mathbb E\|\mathbf D_t\|^2$. By \eqref{eq:ks_lifted_stability},
\begin{equation}
\mathbb E\|\Delta F_s\|^2\le\Lambda^2\mathsf e(s).
\end{equation}
Moreover,
\begin{equation}
0\le \varkappa_m(t)\le1,
\qquad
\int_0^t \varkappa_m(t-s)\,ds\le\frac m\gamma,
\qquad
\int_0^t \varkappa_m(t-s)^2\,ds\le\frac m{2\gamma}.
\end{equation}
Using $\|\boldsymbol\xi_1+\boldsymbol\xi_2+\boldsymbol\xi_3+\boldsymbol\xi_4\|^2
\le4\sum_{j=1}^4\|\boldsymbol\xi_j\|^2$, Cauchy--Schwarz, and It\^o's isometry in
\eqref{eq:ks_difference}, we obtain
\begin{equation}
\mathsf e(t)
\le
\frac{4C_0m}{\gamma^2}
+\frac{4\Lambda^2T}{\gamma^2}\int_0^t\mathsf e(s)\,ds
+\frac{4C_Fm^2}{\gamma^4}
+\frac{4dm}{\gamma^2\beta}.
\end{equation}
Since $0<m\le1$, Gronwall's inequality gives
\begin{equation}
\sup_{0\le t\le T}\mathsf e(t)
\le
\frac4{\gamma^2}
\left(C_0+\frac n\beta+\frac{C_F}{\gamma^2}\right)
e^{4\Lambda^2T^2/\gamma^2}\,m.
\end{equation}
The pair $(\mathbf X_t^m,\mathbf Y_t)$ is a coupling of $q_t^{m,\beta}$ and
$q_t^\beta$. Therefore,
\begin{equation}
\mathcal W_2(q_t^{m,\beta},q_t^\beta)^2
\le \mathsf e(t),
\end{equation}
and hence
\begin{equation}
\sup_{0\le t\le T}
\mathcal W_2(q_t^{m,\beta},q_t^\beta)
\le C_{T,\beta}\sqrt m,
\end{equation}
with
\begin{equation}
C_{T,\beta}
=
\frac2\gamma
\left(C_0+\frac n\beta+\frac{C_F}{\gamma^2}\right)^{1/2}
e^{2\Lambda^2T^2/\gamma^2},
\end{equation}
which is independent of $m$.
\end{proof}

\begin{corollary}[Zero-temperature limit]\label{cor:zero_temp}
Under the force stability conditions of Assumption~\ref{ass:overdamped}, let $q_0\in\mathcal P_2(\mathbb R^d)$. Let $q_t^\beta$ and $q_t$ be the unique weak solutions of \eqref{eq:finite_temp} and \eqref{eq:zero_temp}, respectively, with the same initial law $q_0$. Then, for every $T>0$,
\begin{equation}
\sup_{t\in[0,T]}
\mathcal W_2(q_t^\beta,q_t)
\le
\sqrt{\frac{2dT}{\gamma}}\,
e^{\Lambda T/\gamma}\,\beta^{-1/2},
\qquad
\Lambda:=L_{\mathbf x}+L_{\mathcal W}.
\end{equation}
In particular, $q^\beta\to q$ uniformly on every finite time interval in
$\mathcal W_2$ as $\beta\to\infty$.
\end{corollary}

\begin{proof}
Let $\mathbf B_t$ be a standard $d$-dimensional Brownian motion. Couple the processes
$\mathbf Y_t^\beta$ and $\mathbf Y_t$ through
\begin{equation}
\begin{aligned}
d\mathbf Y_t^\beta
&=
\frac1\gamma F_p[q_t^\beta](\mathbf Y_t^\beta)\,dt
+\sqrt{\frac2{\gamma\beta}}\,d\mathbf B_t,\\
d\mathbf Y_t
&=
\frac1\gamma F_p[q_t](\mathbf Y_t)\,dt,
\end{aligned}
\end{equation}
with a common initial random variable of law $q_0$, independent of $\mathbf B$. Their
time-marginal laws are $q_t^\beta$ and $q_t$, respectively.

Set
\begin{equation}
\mathsf e_\beta(t):=\mathbb E\|\mathbf Y_t^\beta-\mathbf Y_t\|^2.
\end{equation}
Since $(\mathbf Y_t^\beta,\mathbf Y_t)$ is a coupling of $q_t^\beta$ and $q_t$,
\begin{equation}\label{eq:zero_temp_coupling}
\mathcal W_2(q_t^\beta,q_t)^2\le \mathsf e_\beta(t).
\end{equation}
The force stability assumptions therefore imply
\begin{equation}
\begin{aligned}
\left(
\mathbb E
\|F_p[q_t^\beta](\mathbf Y_t^\beta)-F_p[q_t](\mathbf Y_t)\|^2
\right)^{1/2}&\le
L_{\mathbf x}\,\mathsf e_\beta(t)^{1/2}
+L_{\mathcal W}\,\mathcal W_2(q_t^\beta,q_t)\\
&\le\Lambda\,\mathsf e_\beta(t)^{1/2}.
\end{aligned}
\end{equation}

Applying It\^o's formula and Cauchy--Schwarz gives, for almost every $t\in[0,T]$,
\begin{equation}
\begin{aligned}
\mathsf e_\beta'(t)
&=
\frac2\gamma
\mathbb E\left[
(\mathbf Y_t^\beta-\mathbf Y_t)\cdot
\bigl(
F_p[q_t^\beta](\mathbf Y_t^\beta)-F_p[q_t](\mathbf Y_t)
\bigr)
\right]
+\frac{2d}{\gamma\beta}\\
&\le
\frac{2\Lambda}{\gamma}\mathsf e_\beta(t)
+\frac{2d}{\gamma\beta}.
\end{aligned}
\end{equation}
Since $\mathsf e_\beta(0)=0$, Gronwall's inequality yields
\begin{equation}
\mathsf e_\beta(t)
\le
\frac{2d}{\gamma\beta}
\int_0^t e^{2\Lambda(t-s)/\gamma}\,ds
\le
\frac{2dt}{\gamma\beta}e^{2\Lambda t/\gamma}.
\end{equation}
Combining this estimate with \eqref{eq:zero_temp_coupling} and taking the supremum
over $t\in[0,T]$ proves the claim.
\end{proof}

\textbf{Proposition 3} (Non-target stationary state). Let $\Omega=[-1,1]$, let
$p\in\mathcal P(\Omega)$ satisfy $\mathsf R_\#p=p$ for $\mathsf R(x)=-x$ with
$p\neq\delta_0$, and assume $\Phi_{\delta_0}$ admits a $C^2$ representative on
$\Omega$. Then, for every choice of generators, $q_t\equiv\delta_0$ is a stationary
distributional solution of \eqref{eq:zero_temp} and the metric slope
$|\partial\mathcal E_p|(\delta_0)$ vanishes, although $\delta_0\neq p$.

\begin{proof}
Write
\begin{equation}
\mathcal U(\alpha,\alpha')
:=
\mathrm{UOT}^{\varphi_1,\varphi_2}(\alpha,\alpha'),
\qquad
\mathcal G(q):=\frac12\bigl[\mathcal U(q,p)+\mathcal U(p,q)\bigr],
\qquad
\mathcal S(q):=\mathcal U(q,q).
\end{equation}
Then
\begin{equation}\label{eq:stationary_energy_decomposition}
\mathcal E_p(q)
=
\mathcal G(q)-\frac12\mathcal S(q)-\frac12\mathcal U(p,p).
\end{equation}
Set $\Phi:=\Phi_{\delta_0}$. The first variation is understood in the sense that, for
every $q\in\mathcal P(\Omega)$,
\begin{equation}\label{eq:stationary_first_variation}
\lim_{\lambda\downarrow0}
\frac{
\mathcal E_p((1-\lambda)\delta_0+\lambda q)-\mathcal E_p(\delta_0)
}{\lambda}
=
\int_\Omega\bigl(\Phi(x)-\Phi(0)\bigr)\,dq(x).
\end{equation}

\paragraph{Reflection symmetry and stationarity.}
Simultaneous reflection of both coordinates preserves the quadratic cost and all
divergence terms. Hence
\begin{equation}
\mathcal U(\mathsf R_\#\alpha,\mathsf R_\#\alpha')=\mathcal U(\alpha,\alpha').
\end{equation}
Since $\mathsf R_\#p=p$, it follows that
\begin{equation}\label{eq:stationary_reflection}
\mathcal E_p(\mathsf R_\#q)=\mathcal E_p(q).
\end{equation}
This argument does not interchange the two marginal slots and therefore does not
require $\varphi_1=\varphi_2$.

Apply \eqref{eq:stationary_reflection} to $q_\lambda=(1-\lambda)\delta_0+\lambda\delta_x$
and differentiate at $\lambda=0+$. By \eqref{eq:stationary_first_variation},
\begin{equation}
\Phi(x)-\Phi(0)=\Phi(-x)-\Phi(0).
\end{equation}
Thus $\Phi$ is even, and its differentiability implies
\begin{equation}
\Phi'(0)=0,
\qquad
F_p[\delta_0](0)=-\Phi'(0)=0.
\end{equation}
Consequently, for every $\zeta\in C_c^\infty(\mathbb R)$, the constant curve
$q_t=\delta_0$ satisfies
\begin{equation}
\frac{d}{dt}\int_\Omega\zeta\,dq_t
=
0
=
\frac1\gamma
\int_\Omega\zeta'(x)F_p[\delta_0](x)\,d\delta_0(x).
\end{equation}
It is therefore a stationary distributional solution of \eqref{eq:zero_temp}.

\paragraph{A bound for the self-interaction term.}
Define
\begin{equation}
H(\varsigma)
:=
\rho\varphi_1(\varsigma)+\rho\varphi_2(\varsigma)
+\varepsilon(\varsigma\log\varsigma-\varsigma+1),
\qquad \varsigma\ge0,
\end{equation}
with $0\log0:=0$, and let $\varsigma_*$ be a minimizer of $H$. Such a minimizer
exists by lower semicontinuity and coercivity. If either generator is $\iota$, the
constraint enforces $\varsigma_*=1$.

For any finite-cost plan $\pi$ with total mass $\varsigma=\pi(\Omega^2)$, Jensen's
inequality gives
\begin{equation}
D_{\varphi_k}(\pi_k\Vert q)\ge\varphi_k(\varsigma),\quad k\in\{1,2\},
\qquad
\mathrm{KL}(\pi\Vert q\otimes q)
\ge \varsigma\log\varsigma-\varsigma+1.
\end{equation}
Since the transport cost is nonnegative,
\begin{equation}
\mathcal S(q)\ge H(\varsigma_*)=\mathcal S(\delta_0).
\end{equation}
Conversely, the competitor $\pi=\varsigma_*\,q\otimes q$ gives
\begin{equation}
\begin{aligned}
\mathcal S(q)
&\le
H(\varsigma_*)+\frac{\varsigma_*}{2}
\int_{\Omega^2}(x-y)^2\,dq(x)\,dq(y)\\
&=
\mathcal S(\delta_0)
+\varsigma_*\left(
\int_\Omega x^2\,dq(x)
-\left(\int_\Omega x\,dq(x)\right)^2
\right).
\end{aligned}
\end{equation}
Therefore,
\begin{equation}\label{eq:stationary_self_bound}
0\le \mathcal S(q)-\mathcal S(\delta_0)
\le \varsigma_*\int_\Omega x^2\,dq(x).
\end{equation}

\paragraph{A quadratic lower bound for the energy.}
The maps $q\mapsto\mathcal U(q,p)$ and $q\mapsto\mathcal U(p,q)$ are convex. Indeed,
the primal objective is jointly convex in the plan and the varying input measure, by
joint convexity of Csisz\'ar divergences and the affine dependence of the product
reference measure on that input. Taking the infimum over the plan preserves
convexity. Thus $\mathcal G$ is convex.

Fix $q\in\mathcal P(\Omega)$ and set $q_\lambda=(1-\lambda)\delta_0+\lambda q$. By
\eqref{eq:stationary_energy_decomposition}, \eqref{eq:stationary_self_bound}, and
convexity of $\mathcal G$,
\begin{equation}
\begin{aligned}
\mathcal E_p(q_\lambda)-\mathcal E_p(\delta_0)
&=
\mathcal G(q_\lambda)-\mathcal G(\delta_0)
-\frac12\bigl(\mathcal S(q_\lambda)-\mathcal S(\delta_0)\bigr)\\
&\le \mathcal G(q_\lambda)-\mathcal G(\delta_0)\\
&\le \lambda\bigl(\mathcal G(q)-\mathcal G(\delta_0)\bigr).
\end{aligned}
\end{equation}
Dividing by $\lambda$ and letting $\lambda\downarrow0$ yields
\begin{equation}
\mathcal G(q)-\mathcal G(\delta_0)
\ge
\int_\Omega\bigl(\Phi(x)-\Phi(0)\bigr)\,dq(x).
\end{equation}
Using \eqref{eq:stationary_self_bound} once more, we obtain
\begin{equation}
\mathcal E_p(q)-\mathcal E_p(\delta_0)
\ge
\int_\Omega\bigl(\Phi(x)-\Phi(0)\bigr)\,dq(x)
-\frac{\varsigma_*}{2}\int_\Omega x^2\,dq(x).
\end{equation}

Let $C_\Phi:=\|\Phi''\|_{L^\infty(\Omega)}<\infty$. Since $\Phi'(0)=0$, Taylor's
theorem gives
\begin{equation}
\Phi(x)-\Phi(0)\ge-\frac{C_\Phi}{2}x^2.
\end{equation}
Consequently,
\begin{equation}\label{eq:stationary_quadratic_bound}
\mathcal E_p(q)-\mathcal E_p(\delta_0)
\ge
-\frac{C_\Phi+\varsigma_*}{2}\int_\Omega x^2\,dq(x)
=
-\frac{C_\Phi+\varsigma_*}{2}\mathcal W_2(q,\delta_0)^2.
\end{equation}

\paragraph{Vanishing metric slope.}
By the definition of the descending metric slope,
\begin{equation}
\begin{aligned}
|\partial\mathcal E_p|(\delta_0)
&:=
\limsup_{\substack{
q\to\delta_0\ \mathrm{in}\ \mathcal W_2\\
q\neq\delta_0}}
\frac{
[\mathcal E_p(\delta_0)-\mathcal E_p(q)]_+
}{
\mathcal W_2(q,\delta_0)
}\\
&\le
\limsup_{q\to\delta_0}
\frac{C_\Phi+\varsigma_*}{2}\mathcal W_2(q,\delta_0)
=0.
\end{aligned}
\end{equation}
The slope is nonnegative, so it equals zero. Finally, $\delta_0\neq p$ by assumption.
\end{proof}

\section{Long-Time Behavior}
\label{app:longtime}

This appendix contains the two long-time results referenced in
Section~\ref{sec:longtime}: convergence to a stationary state without further
assumptions, and algebraic convergence to the target under
Assumption~\ref{ass:longtime}.

\begin{proposition}[Convergence to a stationary state]\label{prop:stationary_limit}
Let $\Omega\subset\mathbb R^d$ be compact and let $(q_t)_{t\ge0}\subset\mathcal P(\Omega)$
be a global regular solution of \eqref{eq:zero_temp} such that
$E(t):=\mathcal E_p(q_t)$ is absolutely continuous on every bounded interval with
$E'(t)=-\frac1\gamma\mathcal D(q_t)$ for a.e.\ $t\ge0$. Suppose
$\mathcal D:\mathcal P(\Omega)\to[0,\infty]$ is lower semicontinuous with respect to
$\mathcal W_2$. Then there exist $t_j\uparrow\infty$ and $q_\infty\in\mathcal P(\Omega)$
such that $\mathcal W_2(q_{t_j},q_\infty)\to0$ and $\mathcal D(q_\infty)=0$; in
particular, $\nabla\Phi_{q_\infty}=\mathbf 0$ $q_\infty$-almost everywhere, so
$q_\infty$ is a Wasserstein-stationary point of $\mathcal E_p$.
\end{proposition}

\begin{proof}
\emph{Step 1: boundedness of the energy.}
Since $\Omega$ is compact, $c_{\max}:=\sup_{\Omega^2}c<\infty$. For
$\alpha_1,\alpha_2\in\mathcal P(\Omega)$, all terms in \eqref{eq:phi_uot} are
nonnegative, so $\mathrm{UOT}^{\varphi_1,\varphi_2}_{\varepsilon,\rho}(\alpha_1,\alpha_2)\ge0$;
the product coupling $\pi=\alpha_1\otimes\alpha_2$ satisfies both marginal
constraints exactly and has zero entropic penalty, so
$\mathrm{UOT}^{\varphi_1,\varphi_2}_{\varepsilon,\rho}(\alpha_1,\alpha_2)\le\int c\,d(\alpha_1\otimes\alpha_2)\le c_{\max}$.
Hence $0\le\mathcal A\le c_{\max}$ on $\mathcal P(\Omega)^2$ and
$|\mathcal E_p(q)|\le c_{\max}$ for every $q\in\mathcal P(\Omega)$.

\emph{Step 2: integrability of the dissipation.}
By hypothesis $E'\le0$ a.e., so $E$ is nonincreasing and, by Step~1, bounded below;
thus $E_\infty:=\lim_{t\to\infty}E(t)$ exists in $[-c_{\max},c_{\max}]$. For every
$T>0$, absolute continuity gives
$\frac1\gamma\int_0^T\mathcal D(q_t)\,dt=E(0)-E(T)$, and letting $T\to\infty$ by
monotone convergence,
\begin{equation}\label{eq:diss_integrable}
\int_0^\infty\mathcal D(q_t)\,dt=\gamma\big(E(0)-E_\infty\big)\le2\gamma c_{\max}<\infty.
\end{equation}

\emph{Step 3: a sequence along which the dissipation vanishes.}
Fix $j\in\mathbb N$. The set $\{t\ge0:\mathcal D(q_t)>1/j\}$ has finite Lebesgue
measure by \eqref{eq:diss_integrable}, so its complement is unbounded. Choose
inductively $t_j>t_{j-1}+1$ with $\mathcal D(q_{t_j})\le1/j$; then $t_j\uparrow\infty$
and $\mathcal D(q_{t_j})\to0$.

\emph{Step 4: compactness.}
Since $\Omega$ is compact, $(\mathcal P(\Omega),\mathcal W_2)$ is compact
\citep{villani2009optimal}. Hence $(q_{t_j})_j$ admits a subsequence, not
relabeled, and a limit $q_\infty\in\mathcal P(\Omega)$ with
$\mathcal W_2(q_{t_j},q_\infty)\to0$.

\emph{Step 5: stationarity of the limit.}
Lower semicontinuity of $\mathcal D$ and Step~3 give
$0\le\mathcal D(q_\infty)\le\liminf_{j\to\infty}\mathcal D(q_{t_j})=0$. Since
$\mathcal D(q_\infty)=\int_\Omega\|\nabla\Phi_{q_\infty}\|^2\,dq_\infty$, this forces
$\nabla\Phi_{q_\infty}=\mathbf 0$ $q_\infty$-a.e., which is the Wasserstein
stationarity condition for \eqref{eq:zero_temp}.
\end{proof}

\begin{proposition}[Conditional algebraic decay]\label{prop:decay}
Let $\Omega\subset\mathbb R^d$ be compact and $(q_t)_{t\ge0}$ a global regular solution
of \eqref{eq:zero_temp} satisfying \eqref{eq:long_time_dissipation} and
Assumption~\ref{ass:longtime}. With $E_0:=\mathcal E_p(q_0)$,
\begin{equation}\label{eq:algebraic_energy_rate}
\mathcal E_p(q_t)\le\frac{\gamma AE_0}{\gamma A+tE_0}\quad\forall t\ge0,
\end{equation}
and $\mathcal W_2(q_t,p)\to0$ as $t\to\infty$. Assumption~\ref{ass:longtime}(ii)
holds with $A=\operatorname{ess\,sup}_tC_{\mathrm{PI}}(q_t)\,\chi^2(p\Vert q_t)$
whenever $\mathcal E_p$ is convex along mixtures, $p\ll q_t$ with
$\frac{dp}{dq_t}\in L^2(q_t)$, and $q_t$ satisfies a Poincar\'e inequality with
constant $C_{\mathrm{PI}}(q_t)$ (Lemma~\ref{lem:ed}).
\end{proposition}

\begin{proof}
Set $E(t):=\mathcal E_p(q_t)$. Combining \eqref{eq:long_time_dissipation} with
Assumption~\ref{ass:longtime}(ii) gives $E'(t)\le-E(t)^2/(\gamma A)$ for a.e.\ $t$,
and $E\ge0$ since $p$ minimizes $\mathcal E_p$ with $\mathcal E_p(p)=0$. If $E_0=0$
the bound is trivial. Otherwise, on any interval where $E>0$, the chain rule for
absolutely continuous functions yields $\frac{d}{dt}\big(1/E\big)\ge1/(\gamma A)$,
hence $1/E(t)\ge1/E_0+t/(\gamma A)$, which is \eqref{eq:algebraic_energy_rate}; if
$E$ reaches zero at some $t_*$, monotonicity and nonnegativity keep it at zero
thereafter, so the bound holds for all $t$. In particular $E(t)\to0$. Since
$\mathcal P(\Omega)$ is $\mathcal W_2$-compact, every sequence $t_j\to\infty$ has a
subsequence with $q_{t_j}\to q_*$, and lower semicontinuity gives
$\mathcal E_p(q_*)\le\liminf_jE(t_j)=0$, so $q_*=p$ by uniqueness of the
minimizer. As every subsequence has a further subsequence converging to $p$,
$\mathcal W_2(q_t,p)\to0$.
\end{proof}

\section{Implementation Details}
\label{app:impl}

\paragraph{ImageNet experiments.}
Detailed hyperparameters and configurations are given in Table~\ref{tab:hparams}. We largely follow the setup of W-Flow \citep{han2026one} without extensive hyperparameter tuning; the only new hyperparameters are the marginal penalties $(\varphi_1,\varphi_2)$, the relaxation strength $\tau$, and the symmetrization of the velocity estimator. Experiments are conducted on Nvidia H20 GPUs with 141\,GB of memory each. The B/2, L/2, and XL/2 models are trained on 8 nodes of 8 GPUs; the ablations use a single node of 8 GPUs. Training L/2 takes about 5 days.

\subsection{Training Algorithm}
\label{app:algorithm}

Algorithm~\ref{alg:uotgf} gives one training step. All transport quantities are
computed under \texttt{no\_grad}; gradients reach the generator only through the
regression loss, so the generator regresses onto a stop-gradient target displaced
along the UOT force.

\begin{algorithm}[t]
\caption{One UOT-GF training step}
\label{alg:uotgf}
\KwIn{generator $f_{\boldsymbol\theta}(\cdot;\cls,w)$; frozen feature blocks
$\{\psi_\ell\}_{\ell=1}^{n_\psi}$; $\varepsilon,\tau,\eta$; guidance distribution
$p(w)$; Sinkhorn iterations $L$}
Sample $w\sim p(w)$ and $N_{\cls}$ classes; for each class $\cls$: real batch
$\mathcal Y_{\cls}=\{\mathbf y_{\cls,j}\}_{j=1}^{N_{\mathrm{pos}}}$, noise
$\{\mathbf z_{\cls,i}\}_{i=1}^{N_{\mathrm{neg}}}$ and $\{\tilde{\mathbf z}_{\cls,i}\}_{i=1}^{N_{\mathrm{neg}}}$\;
Sample an unconditional real batch $\mathcal Y_\varnothing$ ($N_{\mathrm{unc}}$)\;
$\mathcal X_{\cls}\leftarrow\{f_{\boldsymbol\theta}(\mathbf z_{\cls,i};\cls,w)\}_i$,\quad
$\tilde{\mathcal X}_{\cls}\leftarrow\{f_{\boldsymbol\theta}(\tilde{\mathbf z}_{\cls,i};\cls,w)\}_i$
\tcp*{no gradient}
\For{each block $\ell$ and class $\cls$}{
  $\mathbf x_i\leftarrow\psi_\ell(f_{\boldsymbol\theta}(\mathbf z_{\cls,i};\cls,w))$ for all $i$\;
  $\widehat{\mathbf v}_{\cls}\leftarrow\textsc{Velocity}(\{\mathbf x_i\},\psi_\ell(\mathcal Y_{\cls}))$,\quad
  $\widehat{\mathbf v}_{\mathrm{self}}\leftarrow\textsc{Velocity}(\{\mathbf x_i\},\psi_\ell(\tilde{\mathcal X}_{\cls}))$,\quad
  $\widehat{\mathbf v}_\varnothing\leftarrow\textsc{Velocity}(\{\mathbf x_i\},\psi_\ell(\mathcal Y_\varnothing))$\;
  $\widehat{\mathbf v}^\ell_{\cls,i}\leftarrow(\widehat{\mathbf v}_{\cls,i}-\widehat{\mathbf v}_{\mathrm{self},i})+w\,(\widehat{\mathbf v}_{\cls,i}-\widehat{\mathbf v}_{\varnothing,i})$\;
  $\bar{\mathbf x}^\ell_{\cls,i}\leftarrow\mathrm{sg}\big(\mathbf x_i+\eta\,\widehat{\mathbf v}^\ell_{\cls,i}\big)$\;
}
$\mathcal L(\boldsymbol\theta)\leftarrow\dfrac{1}{n_\psi N_{\cls}N_{\mathrm{neg}}}
\displaystyle\sum_{\ell,\cls,i}\big\|\psi_\ell(f_{\boldsymbol\theta}(\mathbf z_{\cls,i};\cls,w))-\bar{\mathbf x}^\ell_{\cls,i}\big\|^2$\;
Update $\boldsymbol\theta$ by a gradient step on $\mathcal L$\;
\end{algorithm}

\begin{algorithm}[t]
\caption{Source-fixed Sinkhorn}
\label{alg:sinkhorn_sf}
\KwIn{cost $\mathbf C\in\mathbb R^{N\times M}$; $\varepsilon$; $\tau=\rho/(\rho+\varepsilon)$; iterations $L$}
\KwOut{plan $\mathbf P$ with $\mathbf P\mathbf 1_M=\tfrac1N\mathbf 1_N$}
$\mathbf K\leftarrow-\mathbf C/\varepsilon$;\quad $\log\mathbf b\leftarrow\mathbf 0$\;
\For{$t=1,\dots,L$}{
  $\log a_i\leftarrow\log\tfrac1N-\operatorname{LSE}_j\big(K_{ij}+\log b_j\big)$
  \tcp*{exact source projection}
  $\log b_j\leftarrow\log\tfrac1M-\tau\operatorname{LSE}_i\big(K_{ij}+\log a_i\big)$
  \tcp*{relaxed target projection}
}
$\log a_i\leftarrow\log\tfrac1N-\operatorname{LSE}_j\big(K_{ij}+\log b_j\big)$
\tcp*{final projection}
$P_{ij}\leftarrow\exp\big(\log a_i+K_{ij}+\log b_j\big)$\;
\end{algorithm}

\subsection{Source-Fixed Sinkhorn Solver}
\label{app:sinkhorn}

For a cost matrix $\mathbf C\in\mathbb R^{N\times M}$ with uniform weights
$\tfrac1N\mathbf 1_N$, $\tfrac1M\mathbf 1_M$, the source-fixed problem
\eqref{eq:empirical_uot} with $(\varphi_1,\varphi_2)=(\iota,\mathrm{KL})$ is solved
in the log domain. With $\mathbf K:=-\mathbf C/\varepsilon$ and
$\tau=\rho/(\rho+\varepsilon)$, initialize $\log\mathbf b=\mathbf 0$ and iterate $L$
times
\begin{equation}\label{eq:sinkhorn_sf}
\log a_i\leftarrow\log\tfrac1N-\operatorname{LSE}_j\big(K_{ij}+\log b_j\big),\qquad
\log b_j\leftarrow\log\tfrac1M-\tau\operatorname{LSE}_i\big(K_{ij}+\log a_i\big),
\end{equation}
where LSE is log-sum-exp. The first update is the exact source projection
($\varphi_1=\iota$); the second is the KL-relaxed target projection, whose exponent
$\tau<1$ is the standard damping for a $\rho\,\mathrm{KL}$ penalty
\citep{Sjourn2019SinkhornDF}. After the loop one further source update is applied so
that the returned plan $P_{ij}=\exp(\log a_i+K_{ij}+\log b_j)$ satisfies
$\mathbf P\mathbf 1_M=\tfrac1N\mathbf 1_N$ exactly for any finite $L$; this also keeps
the particle velocities correctly scaled. Setting $\tau=1$ recovers balanced Sinkhorn
and hence W-Flow.

\subsection{Velocity Estimator}
\label{app:velocity}

\textsc{Velocity}$(\mathbf x,\mathbf y)$ estimates the symmetrized
cross-interaction contribution in Proposition~\ref{prop:uot_force}
from the empirical plans. Let
$\mathbf P^{\to}\in\mathbb R_+^{N\times M}$ be the plan of
$\mathrm{UOT}(\widehat q,\widehat p)$
($\mathbf x$ source-fixed), and
$\mathbf P^{\leftarrow}\in\mathbb R_+^{M\times N}$ that of
$\mathrm{UOT}(\widehat p,\widehat q)$
($\mathbf y$ source-fixed, $\mathbf x$ relaxed).
The barycenters and mass ratios below are computed from these
empirical plans:
\begin{align}
\text{forward:}\quad&
\widehat{\mathbf v}^{\to}_i
=
\frac{\sum_j P^{\to}_{ij}\mathbf y_j}
     {\sum_j P^{\to}_{ij}}
-\mathbf x_i
=
B^{(1)}_{q,p}(\mathbf x_i)-\mathbf x_i,
\label{eq:fwd_vel}\\
\text{reverse:}\quad&
\widehat{\mathbf v}^{\leftarrow}_i
=
N\Big(
\sum_j P^{\leftarrow}_{ji}\mathbf y_j
-\pi^{\leftarrow}_{2,i}\mathbf x_i
\Big)
=
r^{(2)}_{p,q}(\mathbf x_i)
\big(
B^{(2)}_{p,q}(\mathbf x_i)-\mathbf x_i
\big),
\qquad
\pi^{\leftarrow}_{2,i}:=\sum_j P^{\leftarrow}_{ji}.
\label{eq:rev_vel}
\end{align}
We use
\textsc{Velocity}$(\mathbf x,\mathbf y)
=\tfrac12(
\widehat{\mathbf v}^{\to}
+\widehat{\mathbf v}^{\leftarrow})$.
In \eqref{eq:fwd_vel}, the row mass equals $1/N$ exactly,
so $r^{(1)}\equiv1$; in \eqref{eq:rev_vel}, the incoming mass
$\pi^{\leftarrow}_{2,i}$ may deviate from $1/N$, giving
$r^{(2)}_i=N\pi^{\leftarrow}_{2,i}$, which averages to one
over the batch but may be below or above one locally.
The same estimator with $\mathbf y$ drawn from a second
generated batch gives the self-interaction term
$\widehat{\mathbf v}_{\mathrm{self}}$; with $\mathbf y$
drawn from unconditional data, it gives the unconditional
term $\widehat{\mathbf v}_{\varnothing}$ for guidance
(Appendix~\ref{app:guidance}).

\paragraph{Effective interaction and marginal bounds.}
Define the effective cross-interaction matrix
\[
\overline{\mathbf P}
:=
\tfrac12\big(
\mathbf P^{\to}
+(\mathbf P^{\leftarrow})^\top
\big).
\]
The estimator can equivalently be written as
\[
\text{\textsc{Velocity}}(\mathbf x,\mathbf y)_i
=
N\sum_j \overline P_{ij}(\mathbf y_j-\mathbf x_i).
\]
Since each directed solve fixes its source marginal,
\[
\mathbf P^{\to}\mathbf 1_M=\tfrac1N\mathbf 1_N,
\qquad
\mathbf P^{\leftarrow}\mathbf 1_N=\tfrac1M\mathbf 1_M.
\]
Consequently,
\begin{align}
\overline{\mathbf P}\mathbf 1_M
&=
\tfrac12\Big(
\tfrac1N\mathbf 1_N
+(\mathbf P^{\leftarrow})^\top\mathbf 1_M
\Big)
\ge \tfrac1{2N}\mathbf 1_N,
\\
\overline{\mathbf P}^{\top}\mathbf 1_N
&=
\tfrac12\Big(
(\mathbf P^{\to})^\top\mathbf 1_N
+\tfrac1M\mathbf 1_M
\Big)
\ge \tfrac1{2M}\mathbf 1_M,
\end{align}
where the inequalities are componentwise.
Equivalently, the associated measure
$\bar\pi=\sum_{i,j}\overline P_{ij}
\delta_{(\mathbf x_i,\mathbf y_j)}$ satisfies
\[
\bar\pi_1\ge\tfrac12\widehat q,
\qquad
\bar\pi_2\ge\tfrac12\widehat p.
\]
Thus, the effective cross interaction allows marginal
reweighting while retaining at least half of each reference
marginal. These bounds also hold for finitely many Sinkhorn
iterations, up to floating-point error, because both directed
solves end with a source projection.
They concern cross-interaction mass and guarantee neither
a nonzero net force nor target coverage.

\paragraph{Symmetry of the construction.}
For the symmetric cost and common regularization parameters
used here, transposing a transport plan gives
\[
\mathrm{UOT}^{\phi_1,\phi_2}(q,p)
=
\mathrm{UOT}^{\phi_2,\phi_1}(p,q).
\]
Exchanging the two marginal penalties therefore leaves the
symmetrized cost $A$, the energy $E_p$, and its exact gradient
field unchanged. Here ``source-fixed'' describes each
directed solve; the full construction uses symmetrized
one-sided relaxation. This symmetry concerns the exact
transport objectives, whereas the marginal bounds above
also apply to the finite-iteration plans.

% \subsection{Velocity Estimator}
% \label{app:velocity}

% \textsc{Velocity}$(\mathbf x,\mathbf y)$ estimates the force of
% Proposition~\ref{prop:uot_force} from the empirical plans. With
% $\mathbf P^{\to}$ the plan of $\mathrm{UOT}(\widehat q,\widehat p)$ ($\mathbf x$
% source-fixed) and $\mathbf P^{\leftarrow}$ that of $\mathrm{UOT}(\widehat p,\widehat q)$
% ($\mathbf y$ source-fixed, $\mathbf x$ relaxed),
% \begin{align}
% \text{forward:}\quad&
% \widehat{\mathbf v}^{\to}_i
% =\frac{\sum_jP^{\to}_{ij}\mathbf y_j}{\sum_jP^{\to}_{ij}}-\mathbf x_i
% =B^{(1)}_{q,p}(\mathbf x_i)-\mathbf x_i,\label{eq:fwd_vel}\\
% \text{reverse:}\quad&
% \widehat{\mathbf v}^{\leftarrow}_i
% =N\Big(\sum_jP^{\leftarrow}_{ji}\mathbf y_j-\pi^{\leftarrow}_{2,i}\,\mathbf x_i\Big)
% =r^{(2)}_{p,q}(\mathbf x_i)\big(B^{(2)}_{p,q}(\mathbf x_i)-\mathbf x_i\big),
% \qquad \pi^{\leftarrow}_{2,i}:=\sum_jP^{\leftarrow}_{ji},\label{eq:rev_vel}
% \end{align}
% and $\textsc{Velocity}=\tfrac12(\widehat{\mathbf v}^{\to}+\widehat{\mathbf v}^{\leftarrow})$.
% In \eqref{eq:fwd_vel} the row mass equals $1/N$ exactly, so $r^{(1)}\equiv1$; in
% \eqref{eq:rev_vel} the incoming mass $\pi^{\leftarrow}_{2,i}$ may deviate from $1/N$,
% giving $r^{(2)}_i=N\pi^{\leftarrow}_{2,i}$, which averages to one over the batch but
% is not bounded by it. The same estimator with $\mathbf y$ drawn from a second
% generated batch gives the self-interaction term $\widehat{\mathbf v}_{\mathrm{self}}$;
% with $\mathbf y$ drawn from unconditional data it gives the unconditional term
% $\widehat{\mathbf v}_\varnothing$ for guidance (Appendix~\ref{app:guidance}).

\subsection{Guidance}
\label{app:guidance}

Following \citet{han2026one}, the guidance weight $w$ is sampled at each training step from the distribution in Table~\ref{tab:hparams} and the generator is
conditioned on it, so a single model covers $w\in[0,3]$; at inference a fixed $w$ is
chosen by search and sampling remains one forward pass. We report guidance as
$s:=w+1$ in the main text. For a class $\cls$ and weight $w$, three velocities are
estimated from the same generated batch: $\widehat{\mathbf v}_{\cls}$ toward real samples
of class $\cls$, $\widehat{\mathbf v}_{\mathrm{self}}$ toward a second generated batch of
class $\cls$, and $\widehat{\mathbf v}_\varnothing$ toward $N_{\mathrm{unc}}$ real samples
drawn from all classes. The regression target uses
\begin{equation}\label{eq:cfg}
\widehat{\mathbf v}_w
=(\widehat{\mathbf v}_c-\widehat{\mathbf v}_{\mathrm{self}})
+w\,(\widehat{\mathbf v}_c-\widehat{\mathbf v}_\varnothing),
\end{equation}
in which the self-interaction term cancels from the guidance difference. No label
dropout is used; the unconditional direction comes from the unconditional data batch,
not from an unconditional generator. 

\subsection{Feature-Space Transport}
\label{app:features}

Costs are squared Euclidean distances between features of the released pretrained
MAE encoder \citep{deng2026generative}, $C^\ell_{ij}=\frac12\|\psi_\ell(\mathbf x_i)-\psi_\ell(\mathbf y_j)\|^2$,
computed independently for each selected block $\ell$. Plans, velocities, and the
regression target are formed per block, and the per-block losses are averaged.
Block selection and feature normalization follow \citet{han2026one}.

\subsection{Hyperparameters}
\label{app:hparams}

Table~\ref{tab:hparams} lists the configuration. Unless stated, values are inherited
from W-Flow \citep{han2026one}.

\begin{table}[t]
\centering\footnotesize
\setlength{\tabcolsep}{3pt}
\renewcommand{\arraystretch}{0.92}
\resizebox{\textwidth}{!}{%
\begin{tabular}{l|c|c|c|c}
\toprule
& \textbf{Ablation} (Table~\ref{table:1}) & \textbf{UOT-GF B/2} & \textbf{UOT-GF L/2} & \textbf{UOT-GF XL/2}\\
\midrule
\rowcolor{gray!15}\multicolumn{5}{l}{\textit{Generator Architecture}}\\
arch & DiT-B/2 & DiT-B/2 & DiT-L/2 & DiT-XL/2\\
input size & $32{\times}32{\times}4$ & $32{\times}32{\times}4$ & $32{\times}32{\times}4$ & $32{\times}32{\times}4$\\
patch size & $2{\times}2$ & $2{\times}2$ & $2{\times}2$ & $2{\times}2$\\
hidden dim & 768 & 768 & 1024 & 1152\\
depth & 12 & 12 & 24 & 28\\
register tokens & 16 & 16 & 16 & 16\\
style embedding tokens & 32 & 32 & 32 & 32\\
\midrule
\rowcolor{gray!15}\multicolumn{5}{l}{\textit{Feature Encoder}}\\
architecture & ResNet & ResNet & ResNet & ResNet\\
SSL pre-train method & latent-MAE & latent-MAE & latent-MAE & latent-MAE\\
ResNet: input size & $32{\times}32{\times}4$ & $32{\times}32{\times}4$ & $32{\times}32{\times}4$ & $32{\times}32{\times}4$\\
ResNet: conv$_1$ stride & 1 & 1 & 1 & 1\\
ResNet: base width & 256 & 640 & 640 & 640\\
ResNet: block type & bottleneck & bottleneck & bottleneck & bottleneck\\
ResNet: blocks / stage & [3, 4, 6, 3] & [3, 4, 6, 3] & [3, 4, 6, 3] & [3, 4, 6, 3]\\
ResNet: size / stage & $[32^2,16^2,8^2,4^2]$ & $[32^2,16^2,8^2,4^2]$ & $[32^2,16^2,8^2,4^2]$ & $[32^2,16^2,8^2,4^2]$\\
MAE: masking ratio & 50\% & 50\% & 50\% & 50\%\\
MAE: pre-train epochs & 192 & 1280 & 1280 & 1280\\
classification finetune & No & 3k steps & 3k steps & 3k steps\\
\midrule
\rowcolor{gray!15}\multicolumn{5}{l}{\textit{Generator Optimizer}}\\
optimizer & AdamW & AdamW & AdamW & AdamW\\
$(\beta_1,\beta_2)$ & (0.9, 0.95) & (0.9, 0.95) & (0.9, 0.95) & (0.9, 0.95)\\
learning rate & 2e-4 & 4e-4 & 4e-4 & 3e-4\\
weight decay & 0.01 & 0.0 & 0.01 & 0.01\\
warmup steps & 5k & 10k & 10k & 10k\\
gradient clip & 2.0 & 2.0 & 2.0 & 2.0\\
training steps & 30k & 200k & 200k & 200k\\
training epochs & 100 & 1280 & 1280 & 1280\\
EMA decay & 0.999 & 0.999 & 0.999 & 0.999\\
\midrule
\rowcolor{gray!15}\multicolumn{5}{l}{\textit{Training Loss Computation}}\\
class labels $N_{\cls}$ & 64 & 128 & 128 & 128\\
positive samples $N_{\mathrm{pos}}$ & 64 & 64 & 64 & 128\\
generated samples $N_{\mathrm{neg}}$ & 64 & 64 & 64 & 64\\
effective batch $B$ ($N_{\cls}{\times}N_{\mathrm{neg}}$) & 4096 & 8192 & 8192 & 8192\\
marginal penalties $(\varphi_1,\varphi_2)$ & $(\iota,\mathrm{KL})$ & $(\iota,\mathrm{KL})$ & $(\iota,\mathrm{KL})$ & $(\iota,\mathrm{KL})$\\
relaxation $\tau$ & 0.985 & 0.985 & 0.985 & 0.985\\
regularization $\varepsilon$ & 0.05 & 0.05 & 0.05 & 0.05\\
Sinkhorn iterations $L$ & 10 & 10 & 10 & 10\\
symmetrized velocity & yes & yes & yes & yes\\
self-transport diagonal mask & no & no & no & no\\
step size $\eta$ & 1.0 & 1.0 & 1.0 & 1.0\\
\midrule
\rowcolor{gray!15}\multicolumn{5}{l}{\textit{CFG Configuration}}\\
train: CFG $w$ range & $[0,3]$ & $[0,3]$ & $[0,3]$ & $[0,3]$\\
train: CFG $w$ sampling & $(w{+}1)^{-3}$ & $(w{+}1)^{-5}$ & 50\%: $w{=}0$; 50\%: $(w{+}1)^{-3}$ & 20\%: $w{=}0$; 80\%: $(w{+}1)^{-4}$\\
train: uncond samples $N_{\mathrm{unc}}$ & 16 & 32 & 32 & 32\\
inference: CFG $w$ search & $[0.0,2.5]$ & $[0.0,2.5]$ & $[0.0,2.5]$ & $[0.0,2.5]$\\
inference: selected $s{=}w{+}1$ & 1.50 & 1.19 & 1.15 & 1.14\\
\bottomrule
\end{tabular}}
\caption{Configurations and hyperparameters. All rows except the transport rows
($(\varphi_1,\varphi_2)$, $\tau$, $L$, $\epsilon$, symmetrization, diagonal mask) follow
\citet[Table~5]{han2026one}. CFG sampling entries abbreviate $p(w)\propto(\cdot)$.}
\label{tab:hparams}
\end{table}

Ablations (Section~\ref{sec:exp}) use DiT-B/2 with $N_c=64$ (effective batch 4096),
30K steps, $\varepsilon=0.05$, and guidance $s=1.5$.

\subsection{Diagnostics}
\label{app:diagnostics}

The solver logs, per plan: the source residual $\max_i|\pi_{1,i}-1/N|$ with $\pi_{1,i}=\sum_jP_{ij}$ (zero up to floating point after the final projection); the target $\mathrm{KL}(\mathbf P^{\!\top}\mathbf 1_N\,\|\,\tfrac1M\mathbf 1_M)$; the target effective sample size fraction $\big(\sum_j\pi_{2,j}\big)^2/\big(M\sum_j\pi_{2,j}^2\big)$ with $\pi_{2,j}=\sum_iP_{ij}$, which measures how much target mass is reallocated; the range $[\min_j,\max_j]$ of $M\pi_{2,j}$, the empirical $r^{(2)}$; and the reverse fraction $\|\widehat{\mathbf v}^{\leftarrow}\|_{\mathrm{rms}}/\|\widehat{\mathbf v}^{\to}\|_{\mathrm{rms}}$.
Figure~\ref{fig:source_fixed_uot_mass} reports $M\pi_{2,j}$ on real ImageNet targets for our checkpoint.

\subsection{Evaluation}
\label{app:eval}

\begin{table}[!htbp]
\centering
\caption{FID under the JiT and ADM evaluation protocols.}
\label{tab:scaling_fid}
\small
\setlength{\tabcolsep}{3.5pt}
\begin{adjustbox}{max width=\columnwidth,center}
\begin{tabular}{@{}lccc@{}}
\toprule
Model / setting & CFG & FID(JiT) $\downarrow$ & FID(ADM) $\downarrow$ \\
\midrule
Drifting Pixel-L   \citep{deng2026generative}              & 1.00 & 1.597925 & 1.463253 \\
W-Flow-B \citep{han2026one}                 & 1.19 & 1.518592 & 1.521997 \\
W-Flow-L  \citep{han2026one}                        & 1.14 & 1.352408 & 1.357978 \\
W-Flow-XL  \citep{han2026one}                     & 1.09 & 1.295583 & 1.333832 \\
\midrule
UOT-GF-B ($\tau=0.985$)          & 1.19 & 1.469181 & 1.461928 \\
UOT-GF-B ($\tau=0.985,\ i=1$)    & 1.19 & 1.548405 & 1.539864 \\
UOT-GF-L ($\tau=0.985$)          & 1.15 & 1.348107 & 1.341342 \\
UOT-GF-L ($\tau=0.985,\ i=1$)    & 1.15 & 1.405842 & 1.400581 \\
UOT-GF-XL 100k     & 1.15 & 1.261174 & 1.261597 \\
UOT-GF-XL 200k     & 1.14 & 1.225218 & 1.219600 \\
\bottomrule
\end{tabular}
\end{adjustbox}
\end{table}
All metrics use 50K generated images against the ImageNet validation reference. FID in the main text follows the ADM evaluation suite \citep{3540261.3540933}. W-Flow's published numbers use the JiT protocol \citep{Li_2026_CVPR}; Table~\ref{tab:scaling_fid} reports both protocols for W-Flow and for UOT-GF. Precision/recall use $k=3$ and density/coverage $k=5$ \citep{3454287.3454640,3524938.3525603}, in both Inception and DINOv2-CLS feature spaces. Guidance scales for reported FIDs are those of Table~\ref{tab:hparams}; Figure~\ref{fig:cfg_sweep_with_table} gives the FID--IS trade-off.

\subsection{Toy example}\label{app:toy}

\begin{figure*}[t]
    \centering
    \includegraphics[width=\textwidth]
    {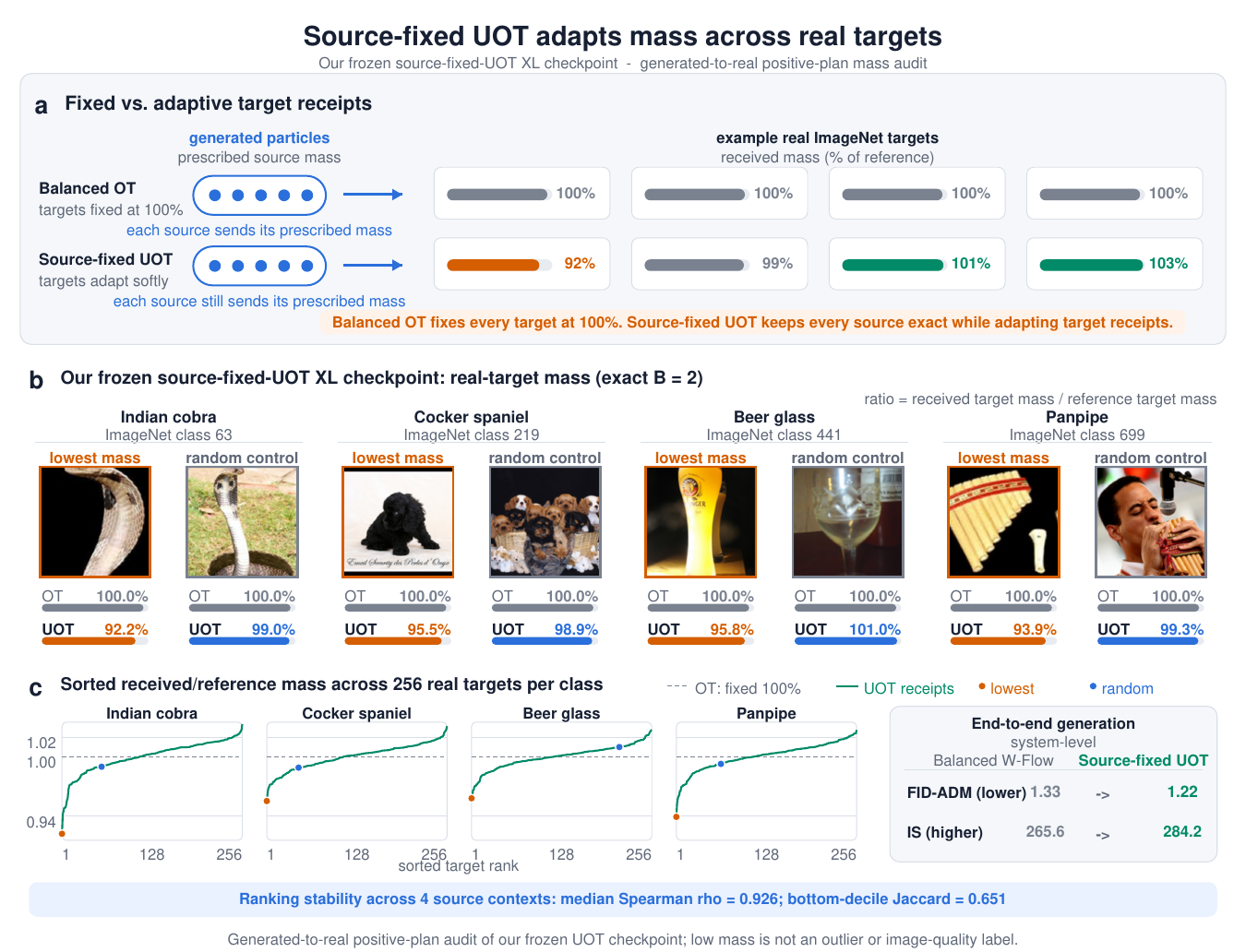}
    \caption{\textbf{Forward source-fixed UOT reweights real targets.}
\textbf{(a)} Balanced OT fixes both marginals; source-fixed UOT
fixes the generated marginal and KL-relaxes the real-data marginal.
\textbf{(b)} Lowest-mass targets and class-matched controls,
where $r_j$ is received mass relative to the uniform target marginal.
\textbf{(c)} Rankings over 256 targets per class are stable across
four contexts ($\rho=0.926$, Jaccard $=0.651$).
Lower mass indicates reduced weight in the forward interaction,
not lower image quality.
Right: final-checkpoint FID and IS.
    }
    \label{fig:source_fixed_uot_mass}
    \vspace{-2mm}
\end{figure*}

\begin{figure}[ht]
    \centering
    \includegraphics[width=\linewidth]{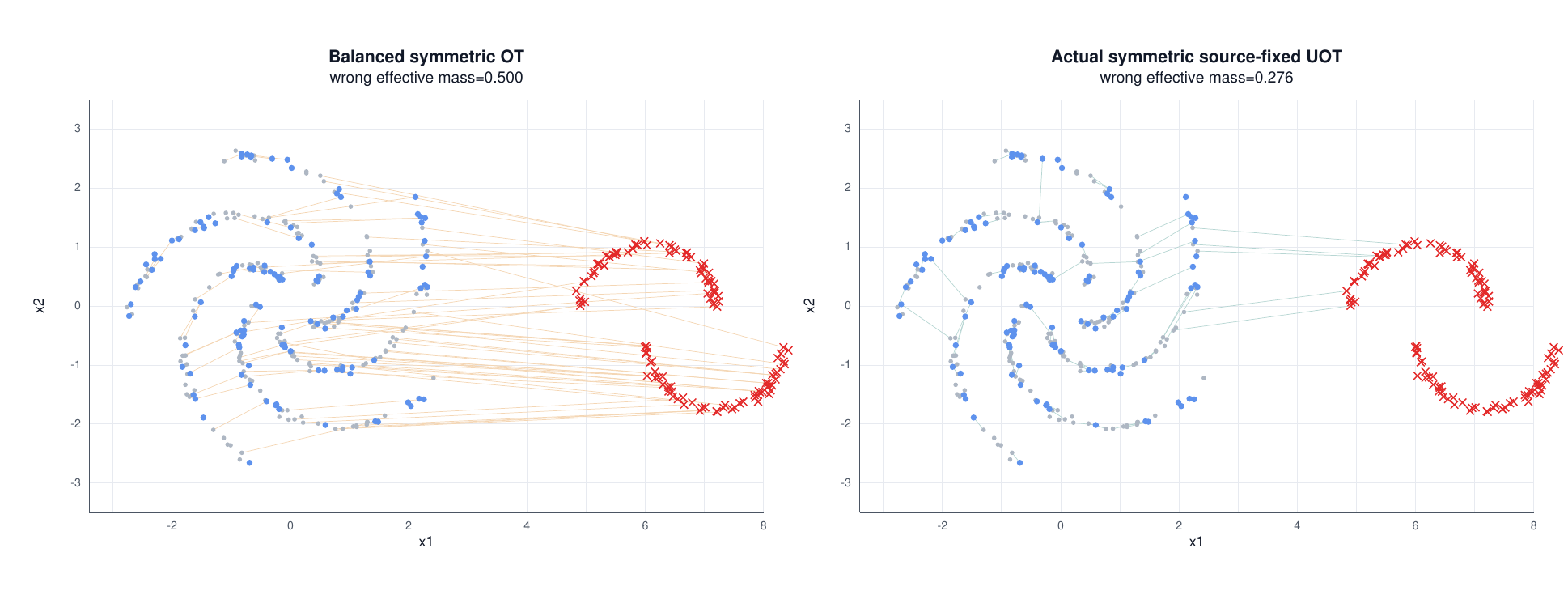}
    \caption{\textbf{Wrong-manifold targets.} Part of the target mass (red crosses)
    lies on a corrupted manifold. Balanced OT (left) must transport full mass to it;
    symmetrized source-fixed UOT (right) reduces the mass sent there. Gray: sources;
    blue: nominal targets.}
    \label{fig:toy_wrong_manifold}
\end{figure}

\begin{figure}[ht]
    \centering
    \includegraphics[width=\linewidth]{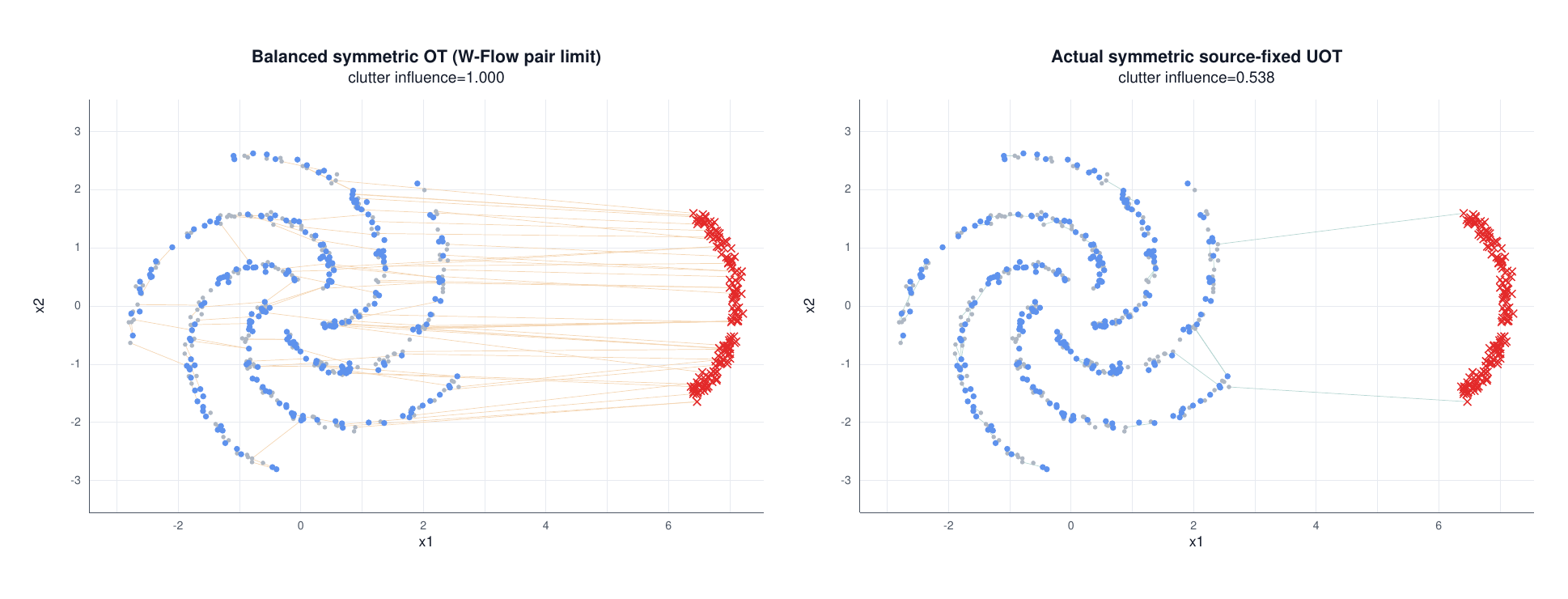}
    \caption{\textbf{Spurious-island targets.} A cluster of spurious targets (red
    crosses) appears away from the data. Balanced OT (left) couples sources to it at
    full weight; symmetrized source-fixed UOT (right) reduces its influence.}
    \label{fig:toy_false_island}
\end{figure}

% \begin{figure}[ht]
%     \centering
%     \includegraphics[width=\linewidth]{figs/03_cluster_proportion_mismatch.pdf}
%     \caption{\textbf{Mode imbalance.} Equal-weight sources (gray dots) against an
%     80/10/10 target (teal crosses), left. Balanced OT (middle) reassigns source mass
%     across clusters to match the proportions; symmetrized source-fixed UOT (right)
%     keeps sources on nearby clusters.}
%     \label{fig:toy_imbalance}
% \end{figure}

\paragraph{Synthetic transport diagnostics.}
We construct four two-dimensional source--target problems to compare
balanced OT with symmetrized source-fixed UOT under target perturbations. Let $\{\mathbf x_i\}_{i=1}^{N}$ and
$\{\mathbf y_j\}_{j=1}^{M}$ denote the source and target samples with uniform weights
$\tfrac1N\mathbf 1_N$ and $\tfrac1M\mathbf 1_M$. In every experiment the ground cost is
$C_{ij}=\tfrac12\|\mathbf x_i-\mathbf y_j\|^2$ and the entropic regularization is
$\varepsilon=0.05$. The balanced baseline minimizes
\[
  \langle\mathbf C,\mathbf P\rangle
  +\varepsilon\,\mathrm{KL}\big(\mathbf P\,\big\|\,\tfrac1{NM}\mathbf 1_N\mathbf 1_M^{\!\top}\big)
  \quad\text{subject to}\quad
  \mathbf P\mathbf 1_M=\tfrac1N\mathbf 1_N,\qquad
  \mathbf P^{\!\top}\mathbf 1_N=\tfrac1M\mathbf 1_M .
\]
The source-fixed unbalanced variant instead solves
\[
  \min_{\mathbf P\ge0,\;\mathbf P\mathbf 1_M=\frac1N\mathbf 1_N}
  \langle\mathbf C,\mathbf P\rangle
  +\varepsilon\,\mathrm{KL}\big(\mathbf P\,\big\|\,\tfrac1{NM}\mathbf 1_N\mathbf 1_M^{\!\top}\big)
  +\rho\,\mathrm{KL}\big(\mathbf P^{\!\top}\mathbf 1_N\,\big\|\,\tfrac1M\mathbf 1_M\big),
\]
where $\tau=\rho/(\rho+\varepsilon)=0.985$, corresponding to $\rho\approx3.283$. The displayed optimization problems define one directed solve.
We solve each problem in both directions, obtaining
$\mathbf P^\to\in\mathbb R_+^{N\times M}$ for source-to-target
transport and $\mathbf P^\leftarrow\in\mathbb R_+^{M\times N}$
for target-to-source transport, with the cost matrix transposed
and the reference weights exchanged in the reverse solve.
The effective cross-interaction plan is
\[
\overline{\mathbf P}
:=\tfrac12\bigl(\mathbf P^\to+
(\mathbf P^\leftarrow)^\top\bigr).
\]
For source-fixed UOT,
$\mathbf P^\to\mathbf 1_M=\frac1N\mathbf 1_N$ and
$\mathbf P^\leftarrow\mathbf 1_N=\frac1M\mathbf 1_M$.
Consequently,
\[
\overline{\mathbf P}\mathbf 1_M
\ge \frac{1}{2N}\mathbf 1_N,
\qquad
\overline{\mathbf P}^{\top}\mathbf 1_N
\ge \frac{1}{2M}\mathbf 1_M,
\]
componentwise.
Thus, each directed solve fixes its own source marginal,
while both effective marginals may be reweighted.
The balanced comparison uses $\tau=1$ in both directions.

The plans are computed with Algorithm~\ref{alg:sinkhorn_sf} in double precision; the
balanced baseline sets $\tau=1$, and no diagonal mask is applied. The four problems
use $L=3000/80$, $3000/150$, $400/400$, and $3000/80$ iterations for the
balanced/source-fixed solvers, respectively. All random samples are reproducible from
the base seed 20260831, with separate deterministic seed offsets for each source,
target, and perturbation.

\paragraph{Wrong-manifold replacement (Figure~\ref{fig:toy_wrong_manifold}).}
We sample 192 source points and 120 clean target points from a five-arm pinwheel. For each point, the radius is uniform on $[0.45,2.75]$, and the angle is $2\pi k/5+1.17r+\eta$, where $k$ is the arm index and $\eta\sim\mathcal{N}(0,0.045^2)$; isotropic Cartesian noise with standard deviation $0.03$ is then added. Another 120 target points are sampled from a shifted, noisy two-moons distribution. These distractors constitute half of the observed target sample.

\paragraph{Spurious target island (Figure~\ref{fig:toy_false_island}).}
The source contains 240 points from a five-arm pinwheel, and the observed target contains 240 independently sampled clean pinwheel points plus 160 remote crescent-shaped distractors. The latter therefore account for 40\% of the 400 target points. The clean pinwheel uses radii uniform on $[0.45,2.85]$, angular coefficient $1.16$, angular noise of standard deviation $0.045$, and Cartesian noise of standard deviation $0.03$. For the distractor crescent, $\theta\sim\mathrm{Uniform}[-1.28,1.28]$ and $R\sim\mathcal{N}(1.02,0.045^2)$, with coordinates
$(6.10+R\cos\theta,\;0.05+1.55R\sin\theta)$ and additional isotropic noise of standard deviation $0.025$.

\paragraph{Cluster-proportion mismatch (Figure~\ref{fig:toy_imbalance}).}
Source and target points are sampled around the same three centers,
$(-2.7,-1.2)$, $(2.7,-1.2)$, and $(0,2.8)$, with isotropic Gaussian noise of standard deviation $0.12$. The 180 source points are divided equally among the three components (60/60/60), whereas the 180 target points have counts 144/18/18, giving an 80/10/10 empirical target proportion.

\paragraph{Oversampled-mode minibatch (Figure~\ref{fig:uot_toy_examples}).}
Eight Gaussian modes are placed on a ring of radius $2.45$, at angles $2\pi k/8+\pi/8$. All points have isotropic Gaussian noise of standard deviation $0.16$. The 192 source points are equally divided among the modes. Of the 240 independently sampled target points, 80\% belong to one mode, with the remainder distributed across the other seven. This is a near-duplicate \emph{mode burst}, not a set of point-identical copies.

The panels display the effective cross-interaction plans
$\overline{\mathbf P}$.
For each method, we connect evenly spaced source points to their
largest-coupling targets,
\[
j^*(i)=\arg\max_j\overline P_{ij}.
\]
The four problems display 96, 100, 54, and 96 such links,
respectively. Colors follow Figure~\ref{fig:uot_toy_examples}, and link width does not encode transported mass.
These plots visualize couplings constructed from both transport directions, rather than integrated particle trajectories.

\section{Additional results}\label{app:addr}

\paragraph{Efficiency evaluation.}
We evaluated all methods at $256\times256$ resolution on NVIDIA H20 GPUs, assigning one GPU to each process. Inference used FP32 tensors with TF32 enabled. We retained each model's high-fidelity sampling configuration: SiT-XL/2 used a 250-step Euler SDE sampler with CFG $1.5$ and the official EMA VAE; RAE-XL used a 50-step Euler ODE sampler with autoguidance $1.42$; BAR-L used guidance $5.3$, temperature $3.0$, token allocation $[2,2,5,7]$, and KV caching; and UOT-GF used its 200k-step XL checkpoint with CFG $1.14$ and one generator evaluation.

For throughput, we used a batch size of 125, discarded 10 warm-up batches, and timed the following 390 batches. Throughput is the number of timed images divided by the elapsed generation-and-decoding time. For single-image latency, we fixed the conditioning label to ImageNet class 0, discarded 20 warm-up requests, and timed 200 batch-one requests individually, synchronizing CUDA immediately before and after each request. We report the mean latency; checkpoint loading and file I/O are excluded from both timing measurements.

We estimated FLOPs per image by running the complete sampling and decoding path under PyTorch's \texttt{FlopCounterMode} and dividing the counted operations by the batch size. These are theoretical estimates for supported operators, not measurements of hardware utilization or energy consumption. The speedup in panel is each model's batch-125 throughput divided by SiT-XL/2's batch-125 throughput.

\begin{figure*}[t]
  \centering
  \includegraphics[width=\textwidth]{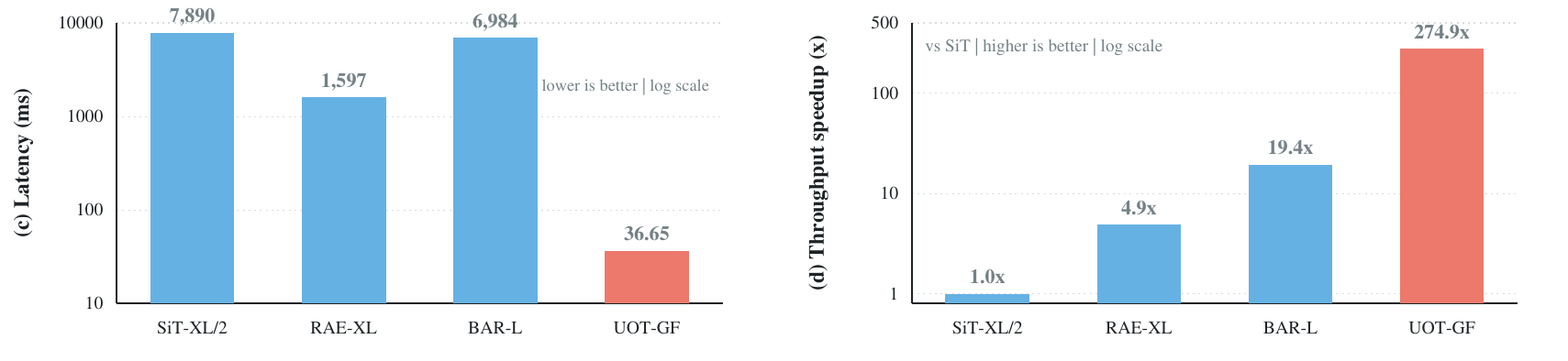}
  \caption{Single-H20 inference efficiency at $256\times256$ resolution.
  (c) Mean batch-one generation-and-decoding latency over 200 timed requests.
  (d) Batch-125 throughput speedup relative to SiT-XL/2.
  Both panels use logarithmic vertical axes and model-specific high-fidelity
  sampling settings. The speedup in (d) is not computed from batch-one latency.}
  \label{fig:efficiency_cd}
\end{figure*}

\paragraph{Controlled visualization of target-marginal relaxation.}

Figure~\ref{fig:controlled_target_mass} isolates the transport mechanism by applying balanced OT and source-fixed UOT to the \emph{same} generated sources and real ImageNet targets. The 128 source particles are generated with a fixed seed from our frozen EMA W-Flow-XL/2 checkpoint at 180k steps (class 178, CFG 1.14). The 128 targets consist of 96 distinct class-178 images, 16 exact copies of four of those images, and 16 foreign-class images (class 0). We encode the targets as 4096-dimensional global SD-VAE latents and use the training-time input scaling and quadratic cost for both operators, with $\epsilon=0.05$ ($\epsilon_{\mathrm{eff}}=204.8$) and 10 Sinkhorn iterations; source-fixed UOT additionally uses $\tau=0.985$. Only the positive-forward coupling is shown. The upper plots report each target's received mass relative to its uniform reference mass: balanced OT keeps this ratio near one, whereas source-fixed UOT preserves the source marginal while allowing target-specific deviations through a soft KL penalty. The lower heatmaps show row-normalized coupling density relative to uniform, with target columns grouped into clean images, duplicate families, and foreign-class images. This controlled probe visualizes marginal adaptation, rather than comparing separately trained generators or establishing a difference in image quality.

\begin{figure*}[t]
    \centering
    \includegraphics[width=0.95\textwidth]{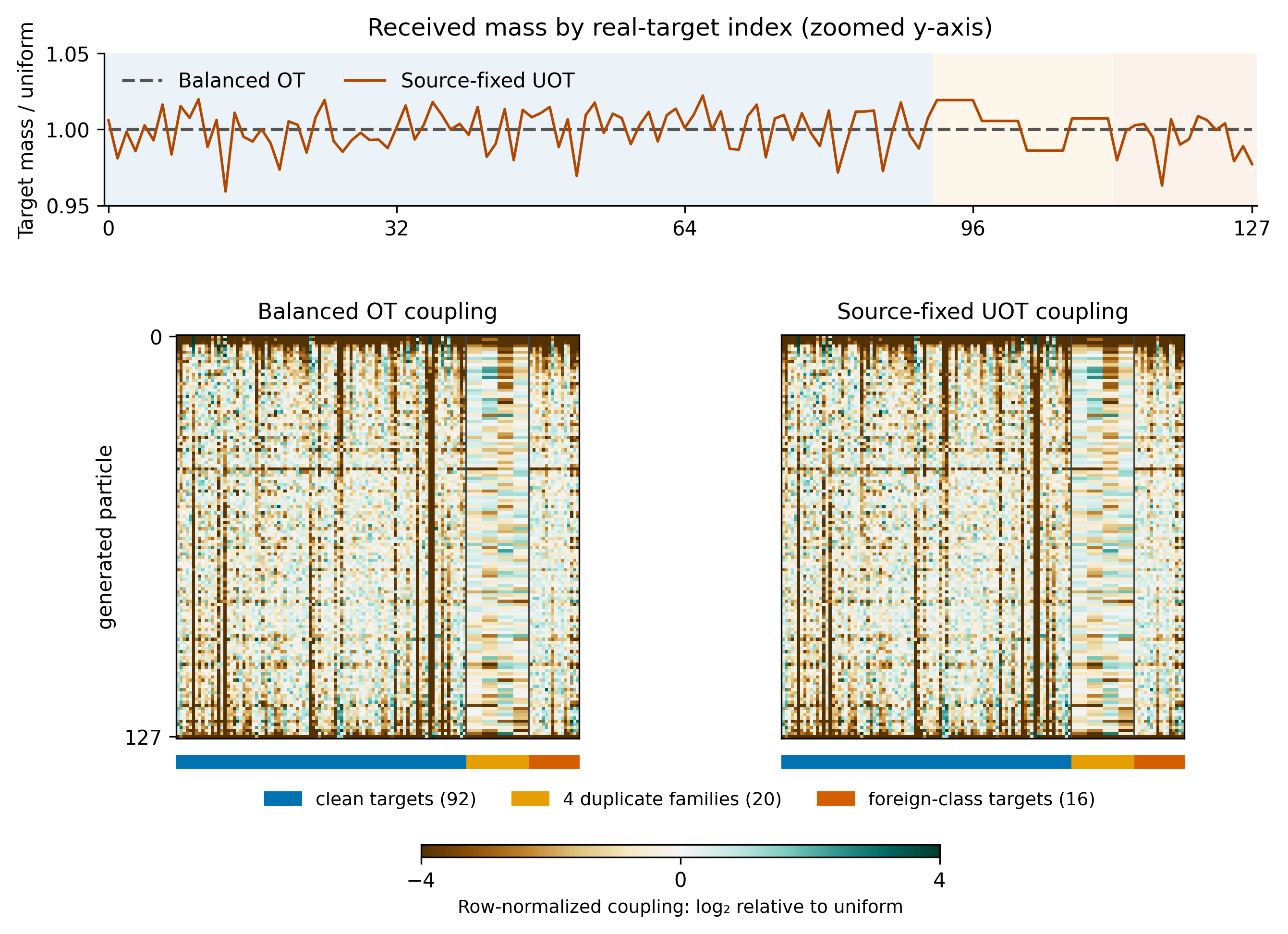}
    \caption{Controlled comparison of balanced OT and source-fixed UOT on identical frozen XL sources and real ImageNet targets. Top: received target mass relative to the uniform reference. Bottom: row-normalized positive-forward couplings; the colored strips distinguish 92 clean targets, 20 members of four duplicate families, and 16 foreign-class targets.}
    \label{fig:controlled_target_mass}
\end{figure*}

\paragraph{Representation-space updates and couplings.}
We additionally visualize a frozen-checkpoint replay of the source-fixed UOT training update for ImageNet classes 178 and 208 (Figs.~\ref{fig:uot_coupling_178} and~\ref{fig:uot_coupling_208}). The shared PCA projection shows generated particles, reference features, and the particles after one Euler update; arrows indicate the projected update direction. The accompanying heatmaps display the coupling components used to construct that update. Transport and velocity are computed in the full 4096-dimensional global latent space; PCA is used only for visualization. These figures are diagnostics of the update mechanism, not balanced-OT comparisons or measures of generation quality.

\begin{figure*}[t]
    \centering
    \includegraphics[width=0.95\textwidth]{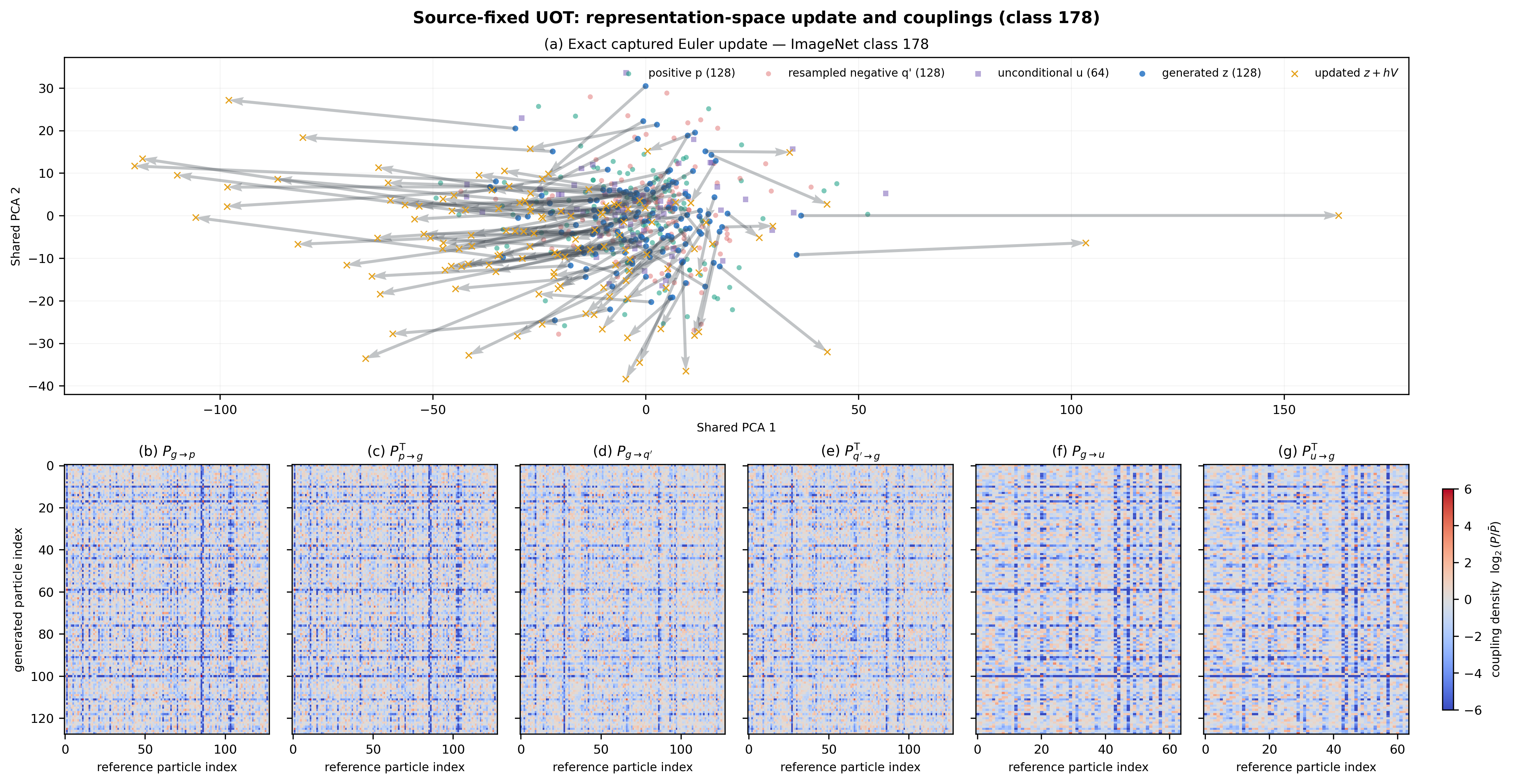}
    \caption{Source-fixed UOT update for ImageNet class 178 using the frozen XL checkpoint. (a) Shared PCA projection of the latent particles and one Euler update; arrows connect generated particles to their updated positions. (b--g) Coupling components used to form the update. Computation uses 4096-dimensional global latents, $\epsilon_{\mathrm{eff}}=204.8$, $\tau=0.985$, and 10 Sinkhorn iterations.}
    \label{fig:uot_coupling_178}
\end{figure*}

\begin{figure*}[t]
    \centering
    \includegraphics[width=0.95\textwidth]{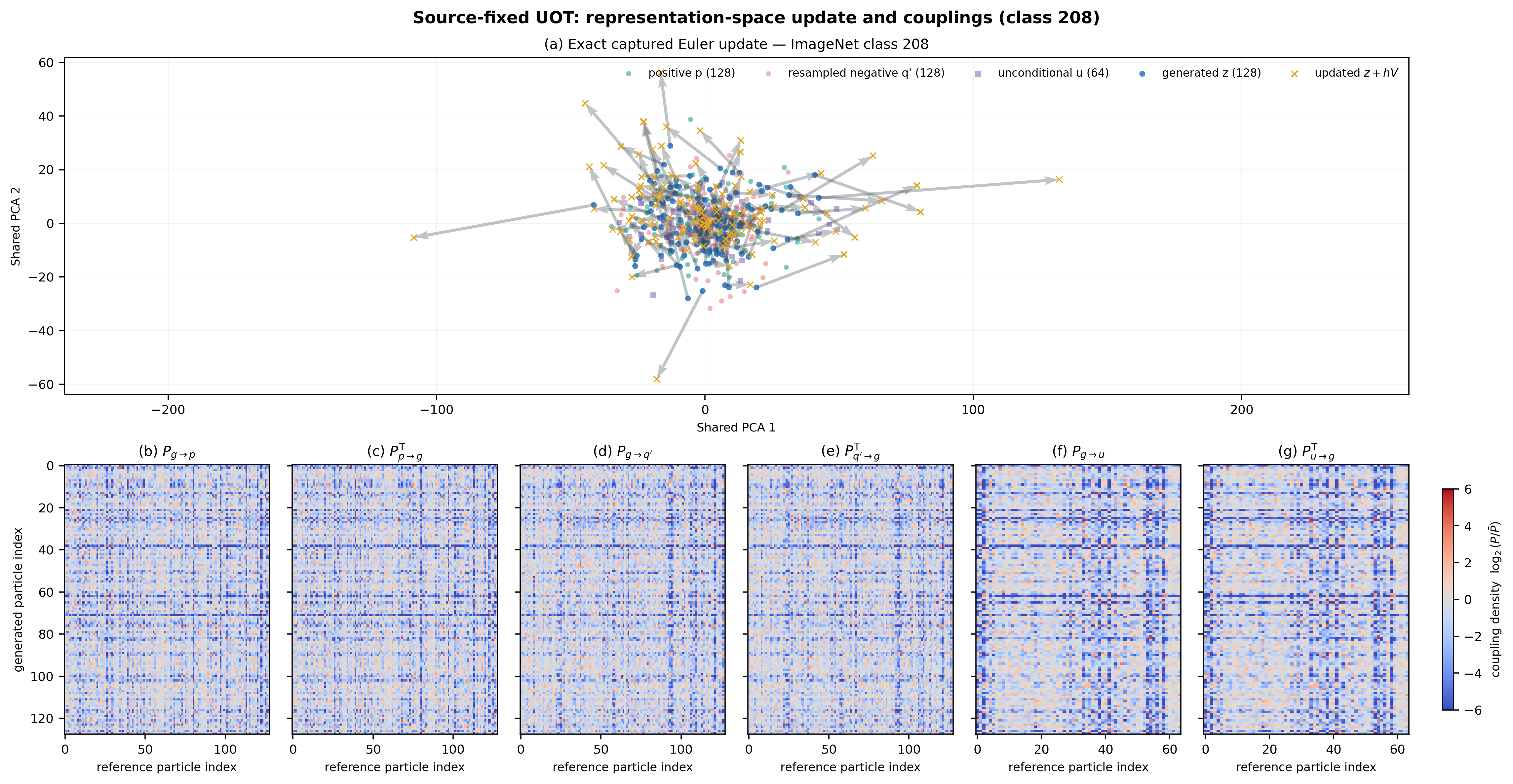}
    \caption{The same frozen-checkpoint source-fixed UOT probe for ImageNet class 208. (a) Shared PCA projection and one Euler update. (b--g) Corresponding coupling components. The PCA coordinates are for display only; all transport computations use the full latent representation.}
    \label{fig:uot_coupling_208}
\end{figure*}

\begin{figure*}[!htbp]
    \centering

    \begin{minipage}[c]{0.47\textwidth}
    \vspace{0pt}
        \centering
        \includegraphics[width=\linewidth]
        {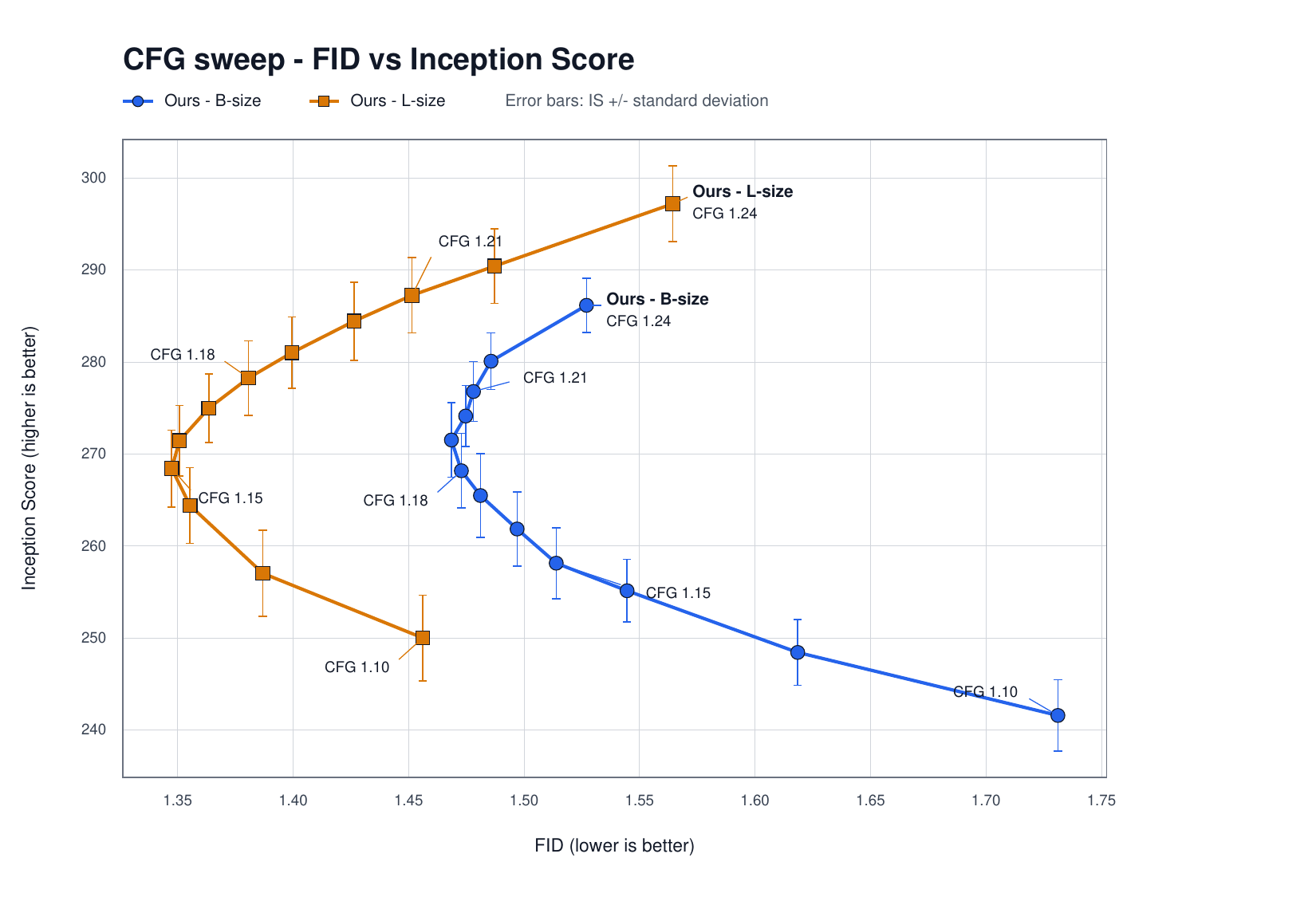}
    \end{minipage}
    \hspace{0.005\textwidth}%
    % ==================== Right: CFG tables ====================
    % ==================== Right: CFG tables ====================
\begin{minipage}[c]{0.515\textwidth}
\vspace{0pt}
    \centering
    \setlength{\tabcolsep}{3pt}
    \renewcommand{\arraystretch}{1.16}

    {\footnotesize\bfseries CFG scales: 1.10--1.17}

    \vspace{0.6mm}

    \resizebox{\linewidth}{!}{%
    \begin{tabular}{@{}llcccccc@{}}
        \toprule
        \rowcolor{gray!12}
        \textbf{Model}
        & \textbf{Metric}
        & \textbf{1.10}
        & \textbf{1.12}
        & \textbf{1.14}
        & \textbf{1.15}
        & \textbf{1.16}
        & \textbf{1.17} \\
        \midrule

        \rowcolor{gray!5}
        \cellcolor{red!15} {\textbf{B/2}} &  \cellcolor{red!15} \textbf{FID} $\downarrow$
        & 1.73
        & 1.61
        & 1.54
        & 1.51
        & 1.49
        & 1.48\\

        \rowcolor{gray!5}
        \cellcolor{red!15}
        & \cellcolor{red!15} \textbf{IS} $\uparrow$
        & $241.59{\pm}3.87$
        & $248.44{\pm}3.56$
        & $255.14{\pm}3.40$
        & $258.14{\pm}3.85$
        & $261.85{\pm}4.05$
        & $265.49{\pm}4.54$ \\

        \addlinespace[1pt]

        \rowcolor{gray!5}
        \cellcolor{red!15} {\textbf{L/2}}  & \cellcolor{red!15} \textbf{FID} $\downarrow$
        & 1.45
        & 1.38
        & 1.35
        & \bestcfg{1.34}
        & 1.35
        & 1.36 \\

        \rowcolor{gray!5}
        \cellcolor{red!15}
        &  \cellcolor{red!15} \textbf{IS} $\uparrow$
        & $250.01{\pm}4.66$
        & $257.03{\pm}4.69$
        & $264.41{\pm}4.12$
        & $268.43{\pm}4.18$
        & $271.45{\pm}3.83$
        & $274.98{\pm}3.74$ \\

        \bottomrule
    \end{tabular}%
    }

    \vspace{1.5mm}

    {\footnotesize\bfseries CFG scales: 1.18--1.24}

    \vspace{0.6mm}

    \resizebox{\linewidth}{!}{%
    \begin{tabular}{@{}llcccccc@{}}
        \toprule
        \rowcolor{gray!12}
        \textbf{Model}
        & \textbf{Metric}
        & \textbf{1.18}
        & \textbf{1.19}
        & \textbf{1.20}
        & \textbf{1.21}
        & \textbf{1.22}
        & \textbf{1.24} \\
        \midrule

        \rowcolor{gray!5}
        \cellcolor{red!15} {\textbf{B/2}}  & \cellcolor{red!15} \textbf{FID} $\downarrow$
        & 1.47
        & \bestcfg{1.46}
        & 1.47
        & 1.47
        & 1.48
        & 1.52 \\

        \rowcolor{gray!5}
        \cellcolor{red!15}
        &  \cellcolor{red!15} \textbf{IS} $\uparrow$
        & $268.18{\pm}4.06$
        & $271.52{\pm}4.05$
        & $274.13{\pm}3.32$
        & $276.79{\pm}3.25$
        & $280.09{\pm}3.09$
        & $286.16{\pm}2.94$ \\

        \addlinespace[1pt]

        \rowcolor{gray!5}
        \cellcolor{red!15} {\textbf{L/2}}  &  \cellcolor{red!15} \textbf{FID} $\downarrow$
        & 1.38
        & 1.39
        & 1.42
        & 1.45
        & 1.48
        & 1.56 \\

        \rowcolor{gray!5}
        \cellcolor{red!15}
        & \cellcolor{red!15} \textbf{IS}$\uparrow$
        & $278.26{\pm}4.05$
        & $281.02{\pm}3.89$
        & $284.44{\pm}4.24$
        & $287.25{\pm}4.09$
        & $290.42{\pm}4.04$
        & $297.19{\pm}4.13$ \\

        \bottomrule
    \end{tabular}%
    }

\end{minipage}
    \caption{
        CFG-scale analysis for our B/2 and L/2 models.
        The left panel shows the FID--IS trade-off, while the right panel    reports the corresponding numerical results.
    }
    \label{fig:cfg_sweep_with_table}
    \par\vspace{2mm}

    % Figure 10 remains separately numbered within the shared float.
    \includegraphics[width=0.49\textwidth]
    {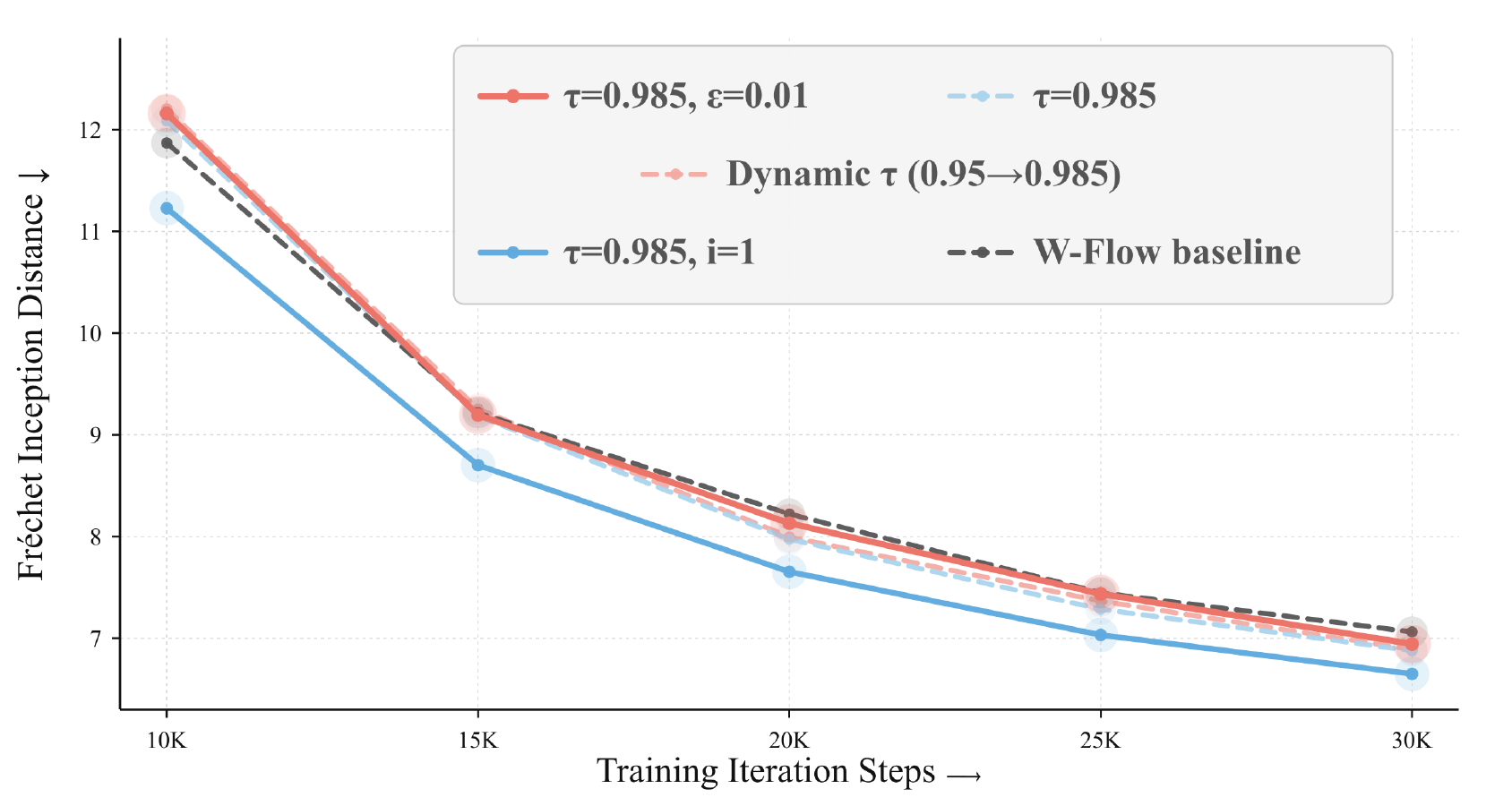}
    \hfill
    \includegraphics[width=0.49\textwidth]
    {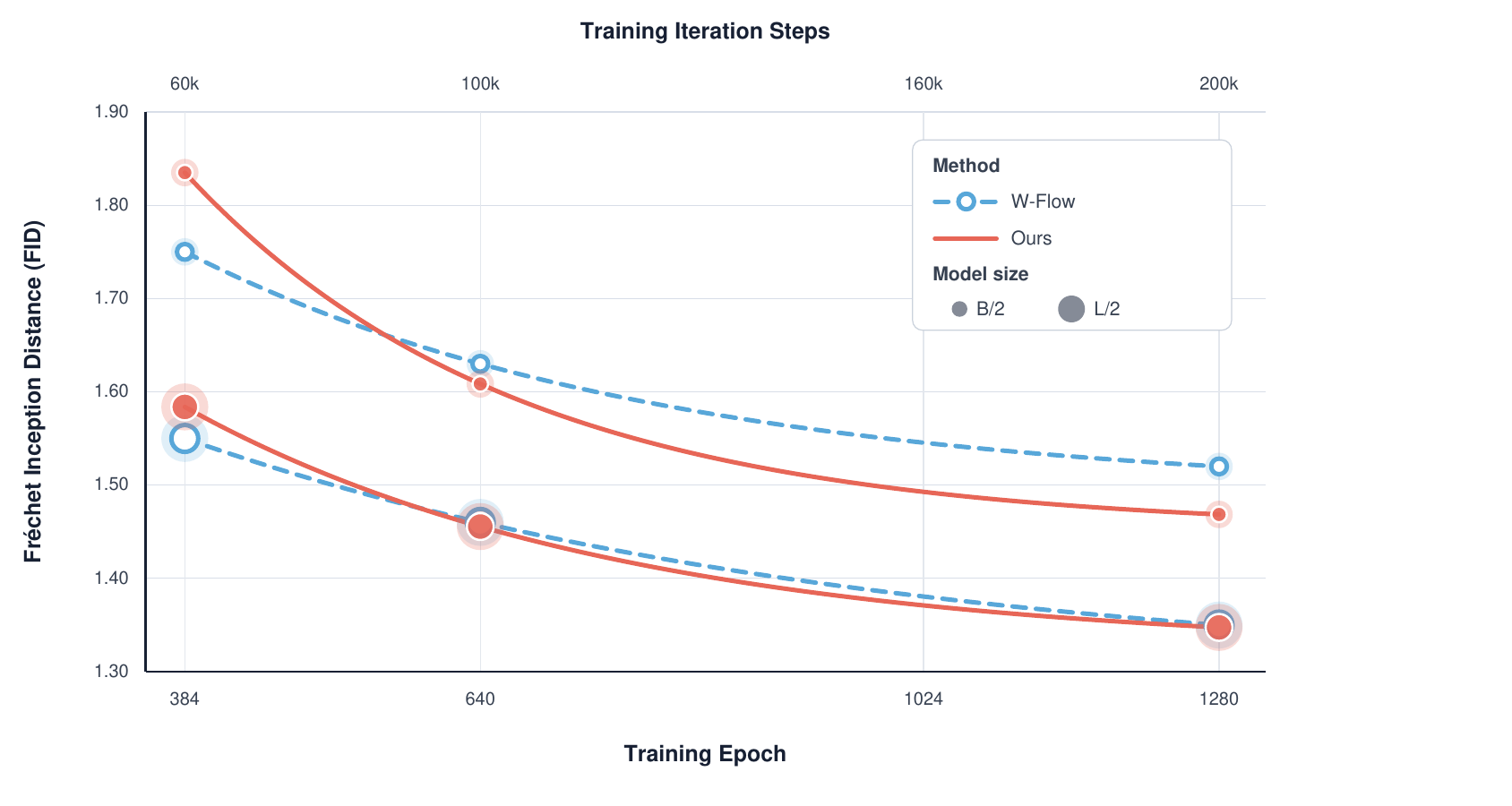}

    \caption{
    Training dynamics and FID convergence.
    (Left) Comparison of our ablation variants with baseline methods.
    (Right) Comparison of our B/2 and L/2 models with baseline methods.
}
    \label{fig:training_convergence}
\end{figure*}

\begin{figure*}[ht]
\centering
\captionsetup[subfigure]{font=small,skip=2pt}

% Row 1
\begin{subfigure}{0.49\textwidth}
    \centering
    \includegraphics[width=\linewidth]{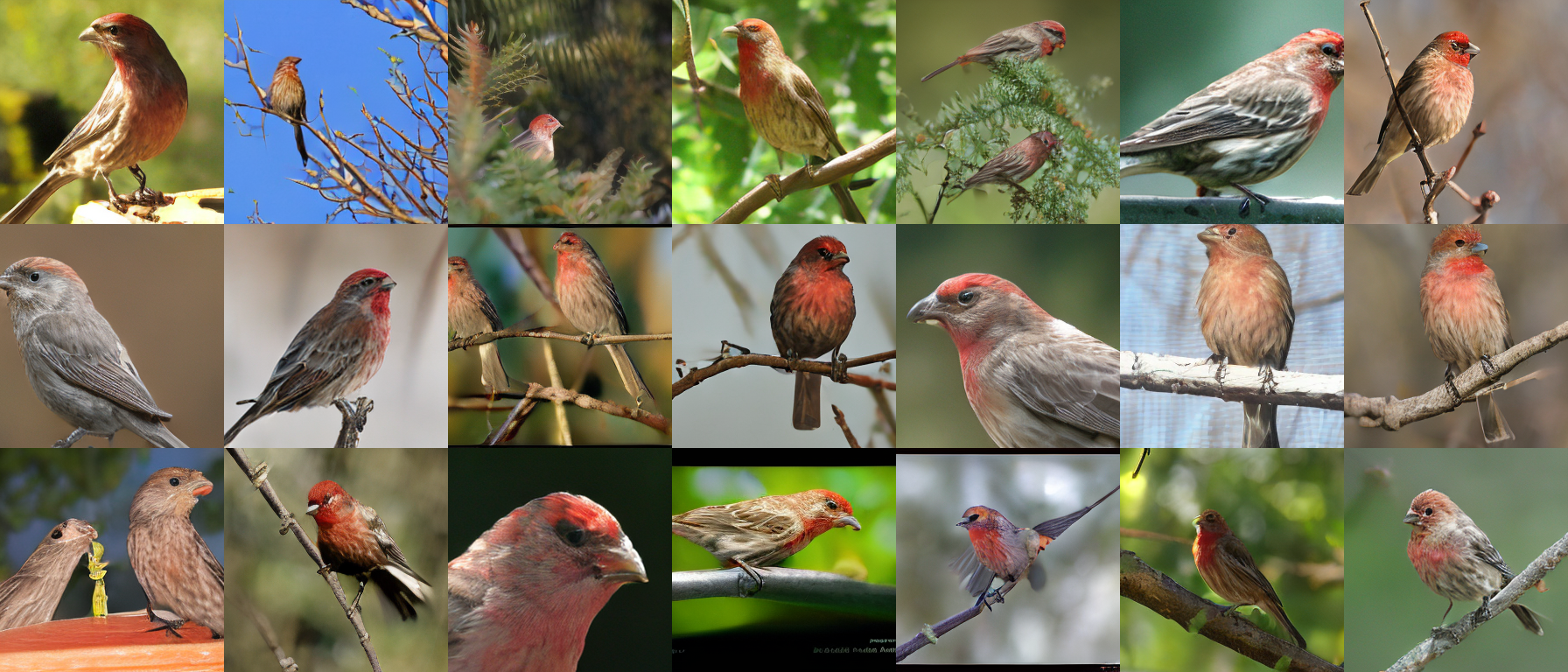}
    \caption{House finch (Class 012)}
\end{subfigure}\hfill
\begin{subfigure}{0.49\textwidth}
    \centering
    \includegraphics[width=\linewidth]{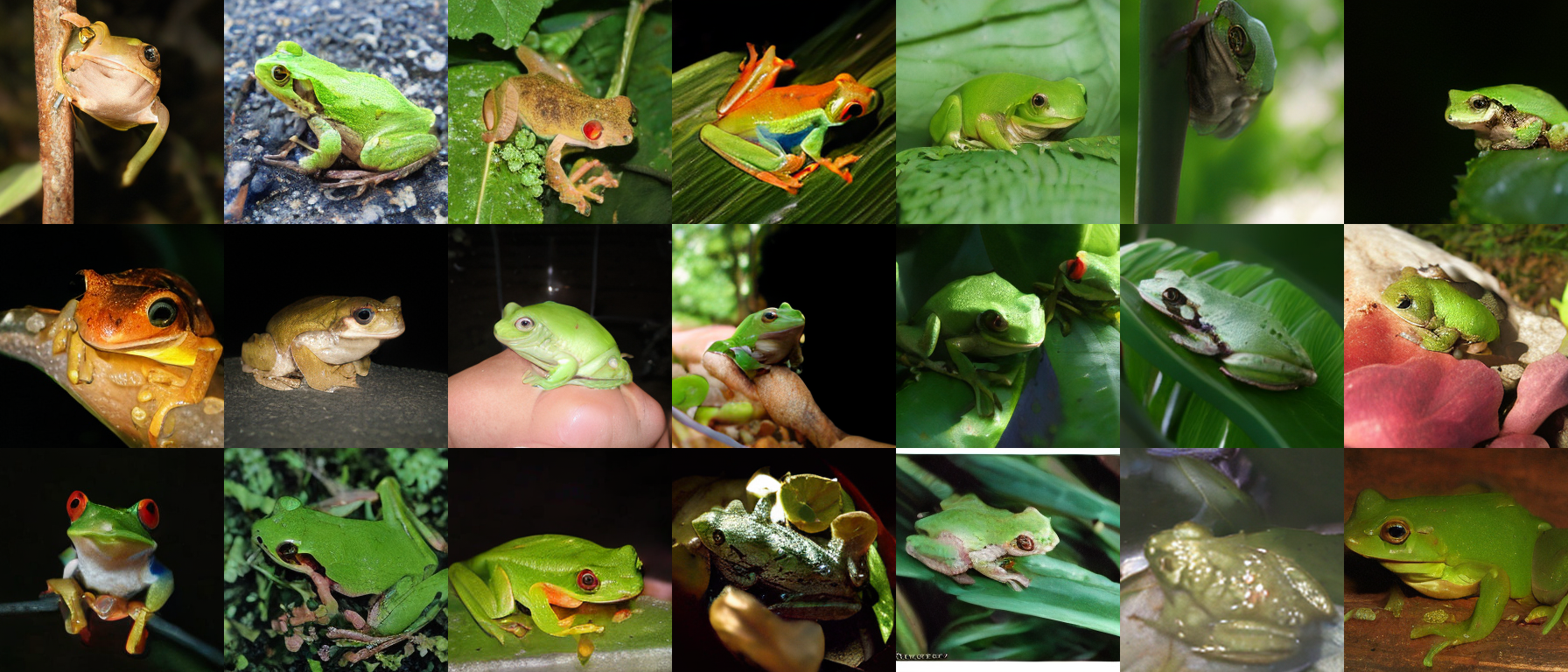}
    \caption{Tree frog (Class 031)}
\end{subfigure}

\par\medskip
% Row 2
\begin{subfigure}{0.49\textwidth}
    \centering
    \includegraphics[width=\linewidth]{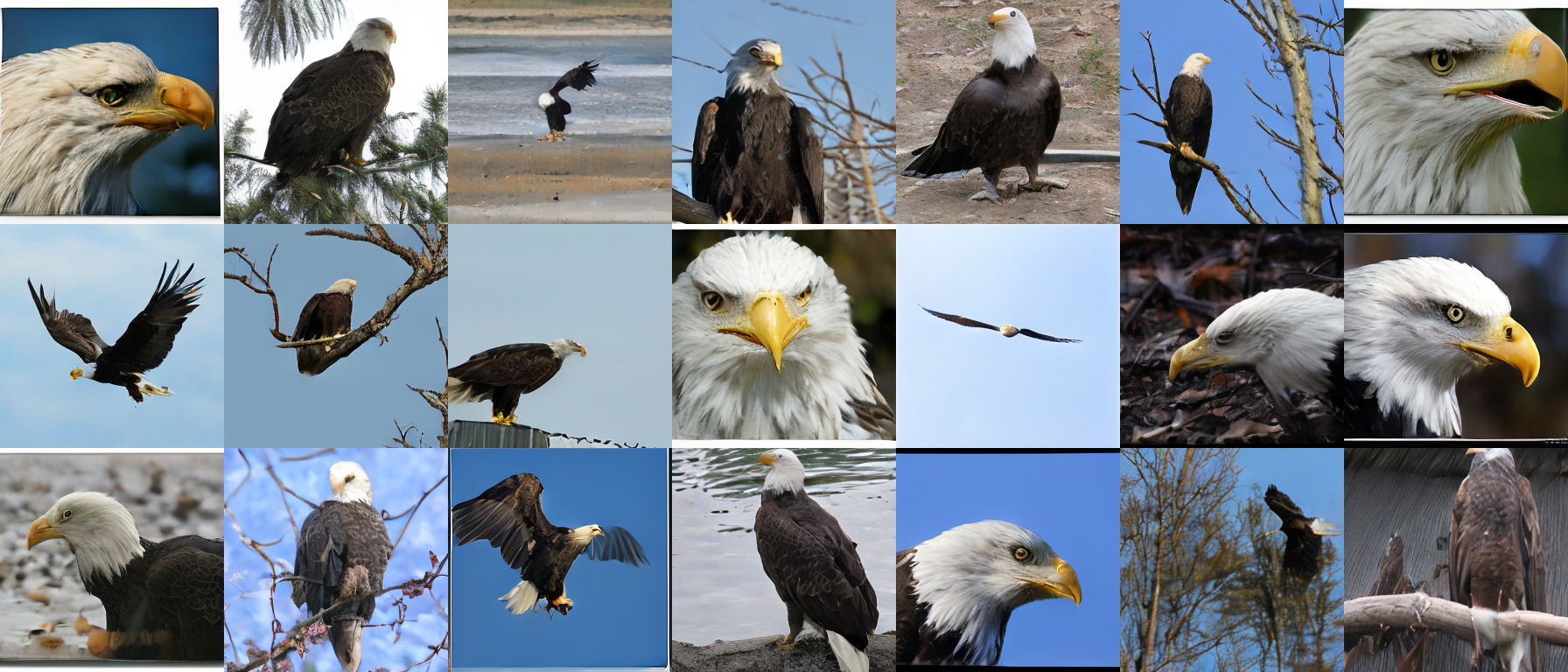}
    \caption{Bald eagle (Class 022)}
\end{subfigure}\hfill
\begin{subfigure}{0.49\textwidth}
    \centering
    \includegraphics[width=\linewidth]{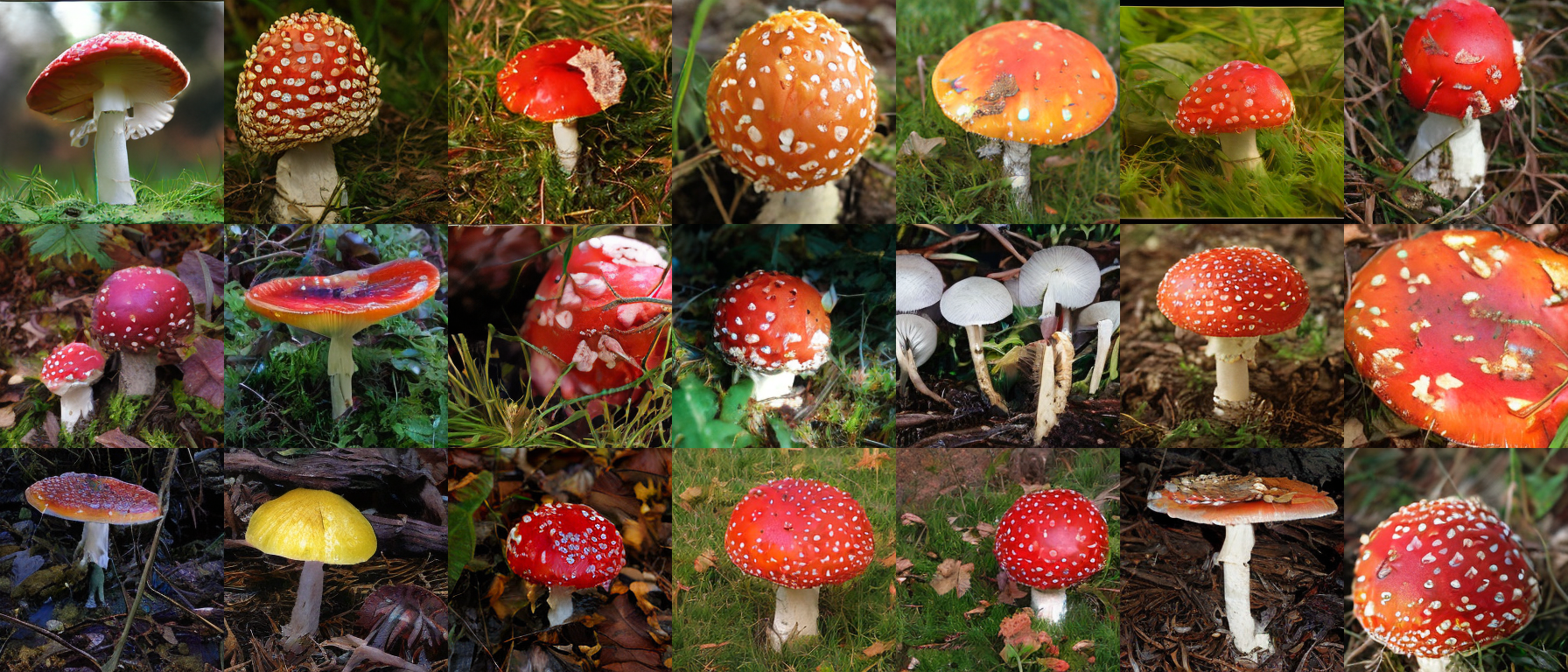}
    \caption{Agaric (Class 992)}
\end{subfigure}

\par\medskip
% Row 3
\begin{subfigure}{0.49\textwidth}
    \centering
    \includegraphics[width=\linewidth]{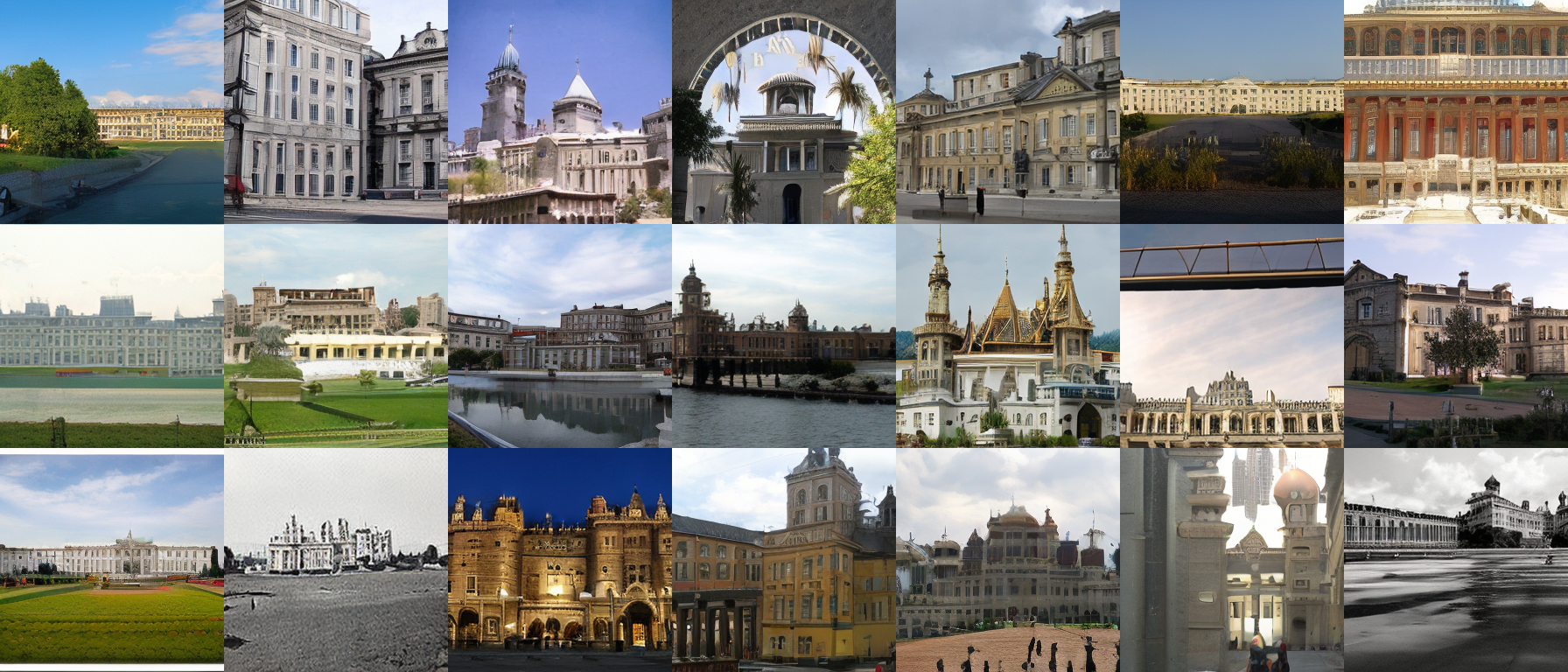}
    \caption{Palace (Class 698)}
\end{subfigure}\hfill
\begin{subfigure}{0.49\textwidth}
    \centering
    \includegraphics[width=\linewidth]{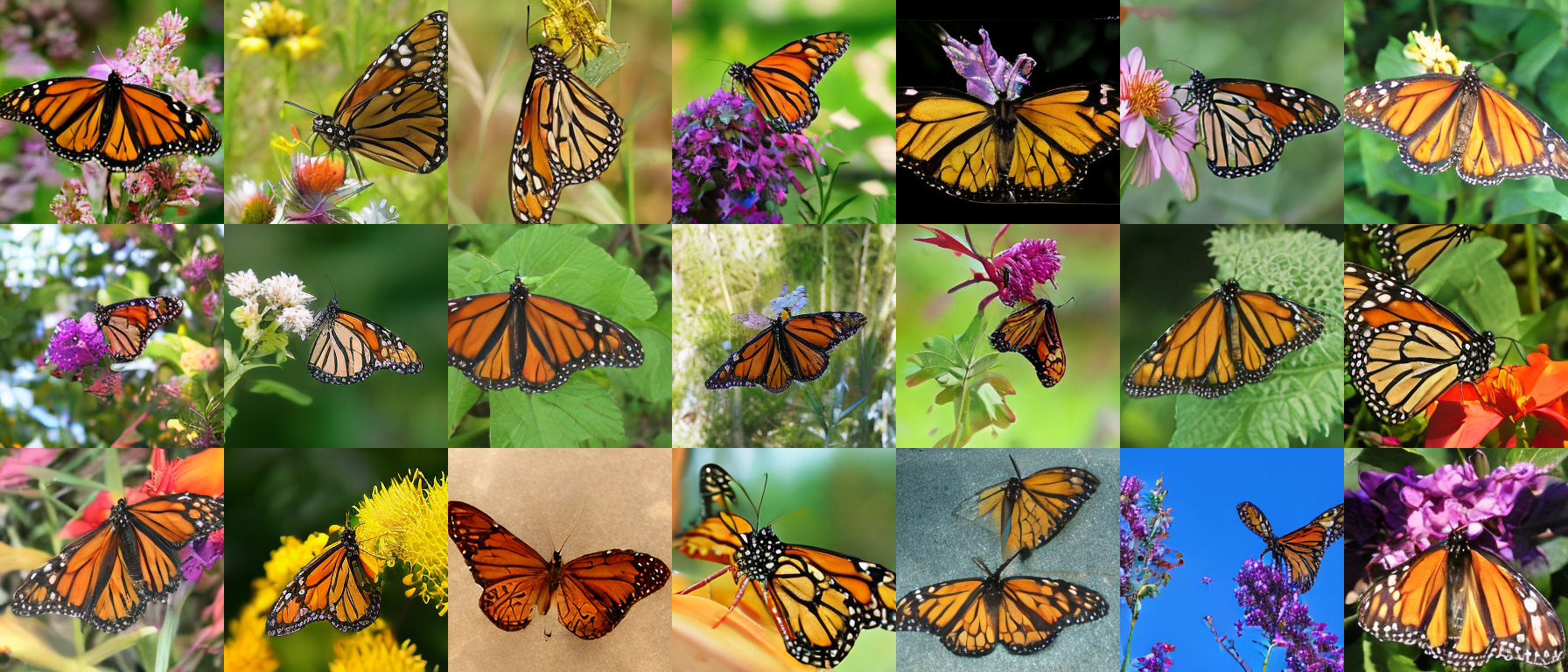}
    \caption{Monarch butterfly (Class 323)}
\end{subfigure}

\par\medskip
% Row 4
\begin{subfigure}{0.49\textwidth}
    \centering
    \includegraphics[width=\linewidth]{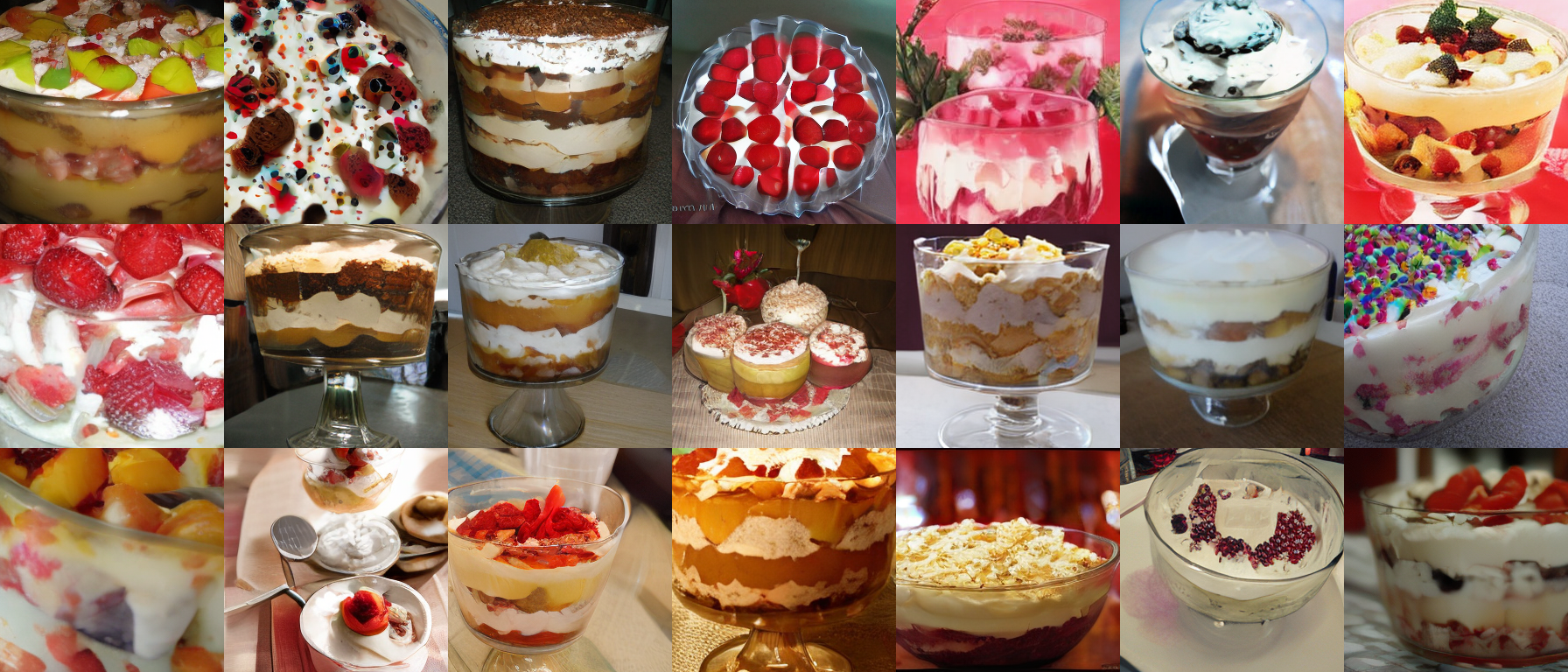}
    \caption{Trifle (Class 927)}
\end{subfigure}\hfill
\begin{subfigure}{0.49\textwidth}
    \centering
    \includegraphics[width=\linewidth]{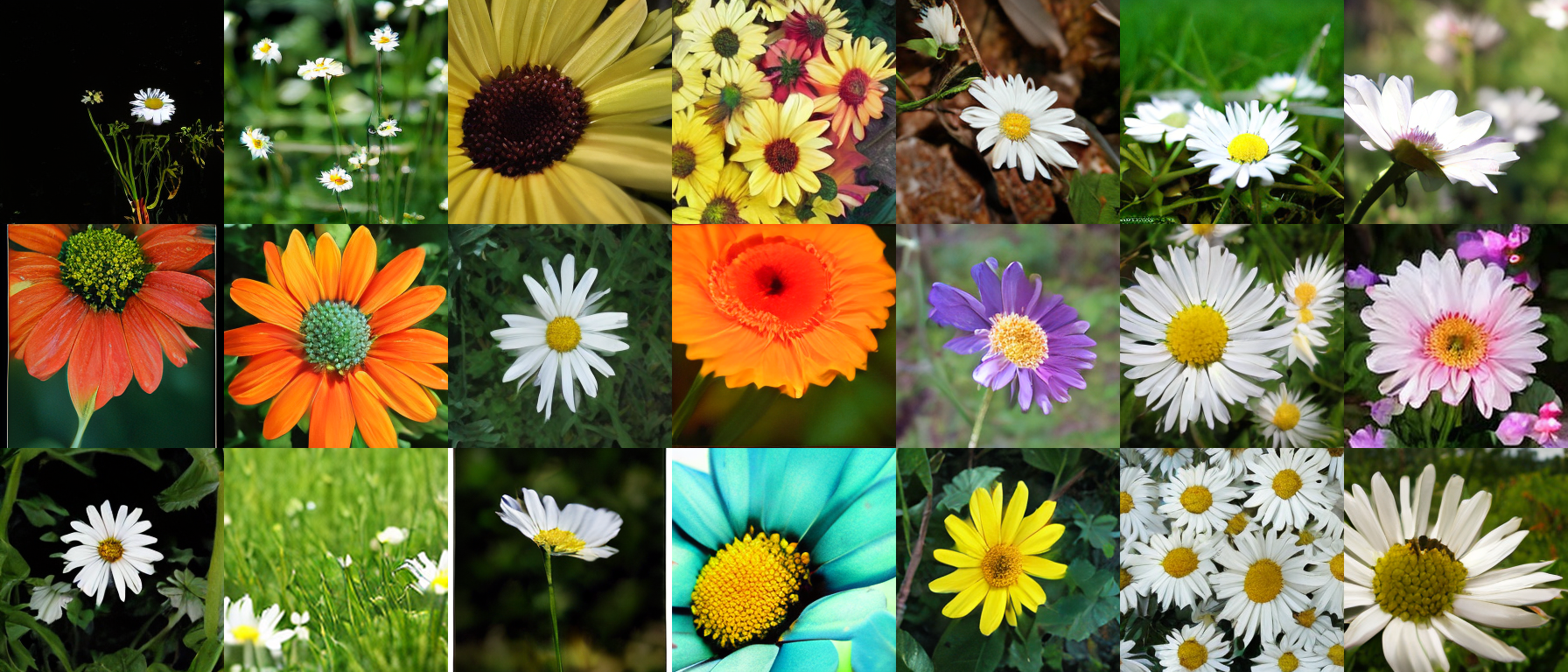}
    \caption{Daisy (Class 985)}
\end{subfigure}

\par\medskip
% Row 5
\begin{subfigure}{0.49\textwidth}
    \centering
    \includegraphics[width=\linewidth]{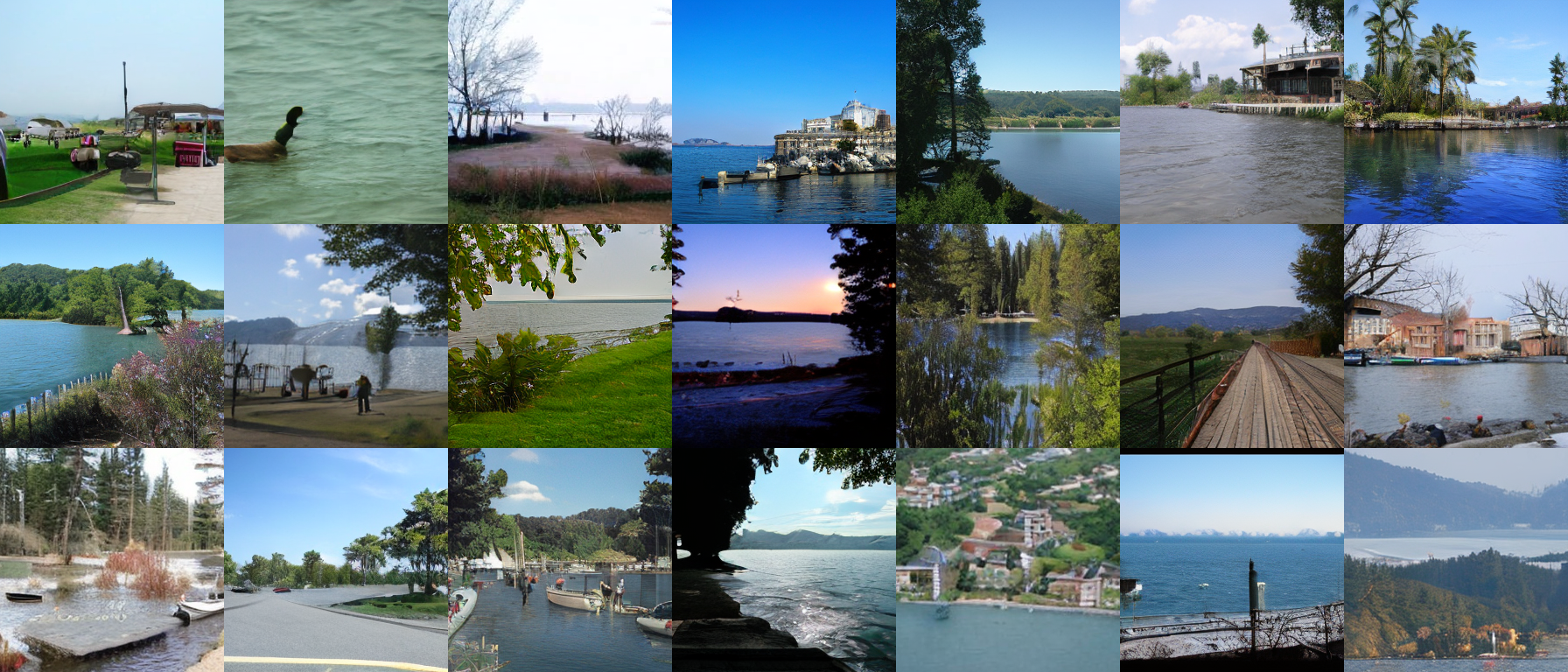}
    \caption{Lakeside (Class 975)}
\end{subfigure}\hfill
\begin{subfigure}{0.49\textwidth}
    \centering
    \includegraphics[width=\linewidth]{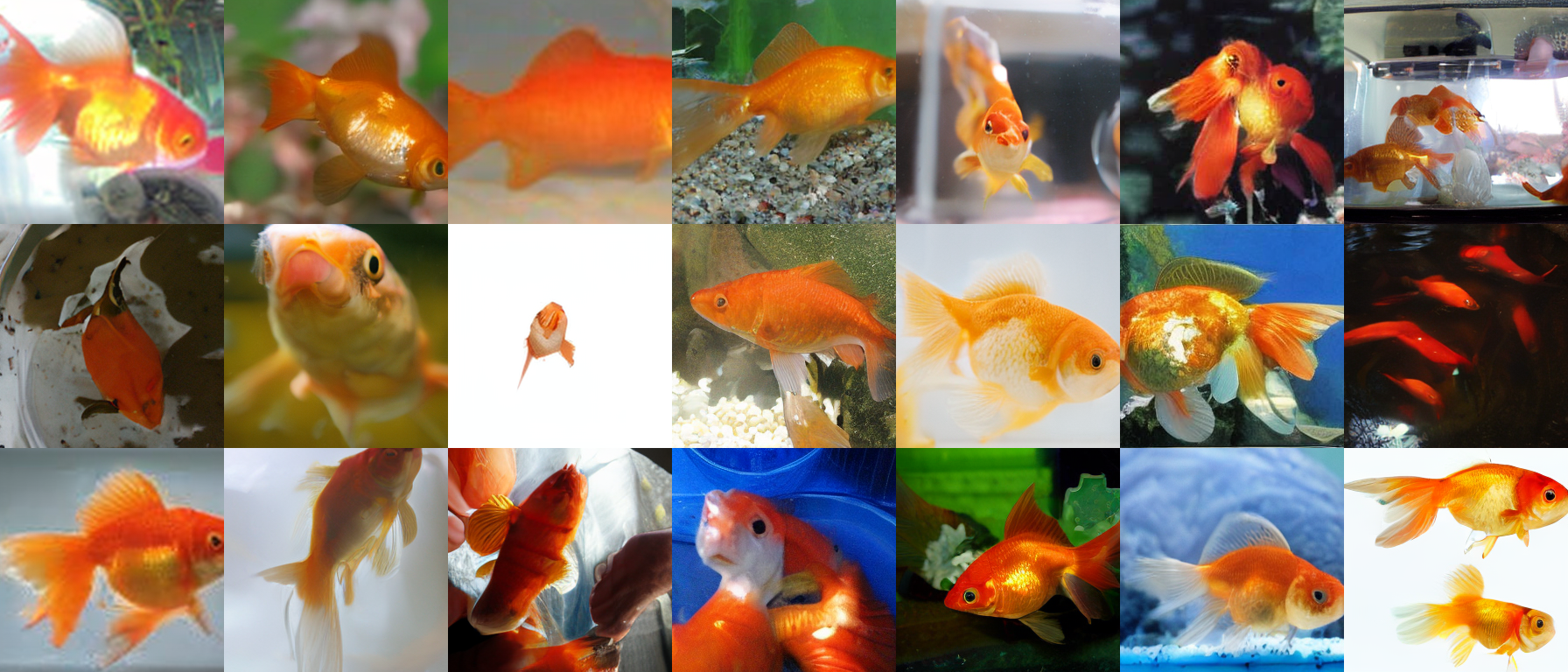}
    \caption{Goldfish (Class 001)}
\end{subfigure}

\caption{Random class-conditional ImageNet samples generated by our UOT-XL model at 200K training steps with CFG $=1.14$. Class indices follow the zero-based ImageNet-1K label mapping.}
\label{fig:samples-five-by-two}
\vspace{-2mm}
\end{figure*}

\begin{figure*}[t]
\centering
\captionsetup[subfigure]{font=small,skip=2pt}

% Row 1
\begin{subfigure}{0.49\textwidth}
    \centering
    \includegraphics[width=\linewidth]{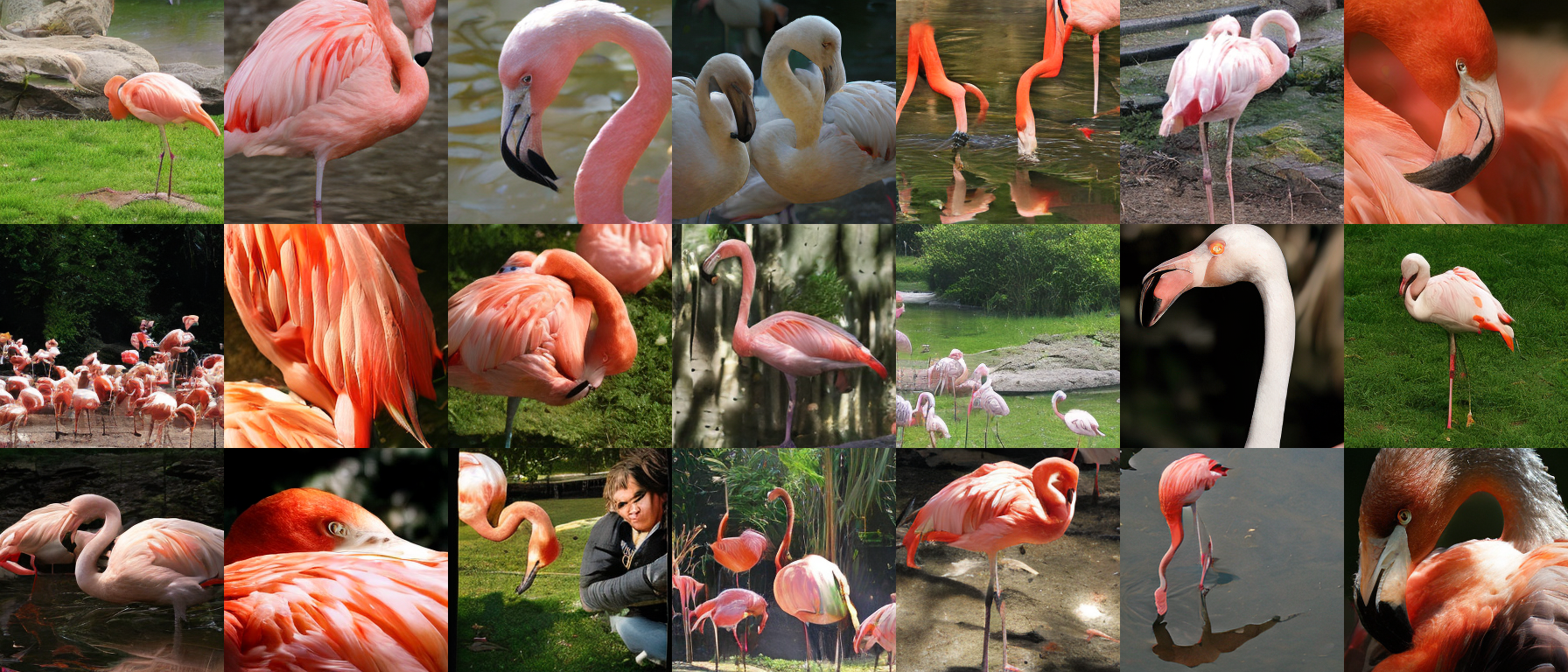}
    \caption{Flamingo (Class 130)}
\end{subfigure}\hfill
\begin{subfigure}{0.49\textwidth}
    \centering
    \includegraphics[width=\linewidth]{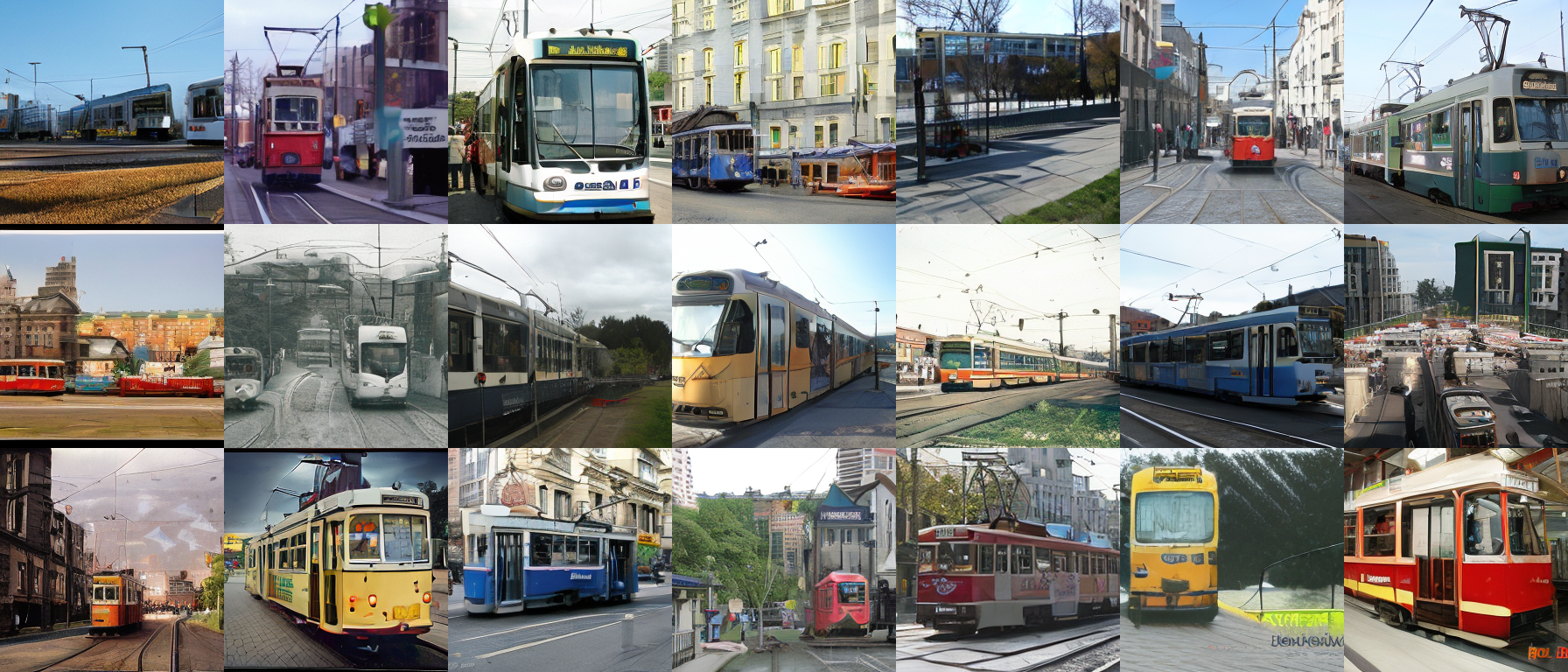}
    \caption{Streetcar (Class 829)}
\end{subfigure}

\par\medskip
% Row 2
\begin{subfigure}{0.49\textwidth}
    \centering
    \includegraphics[width=\linewidth]{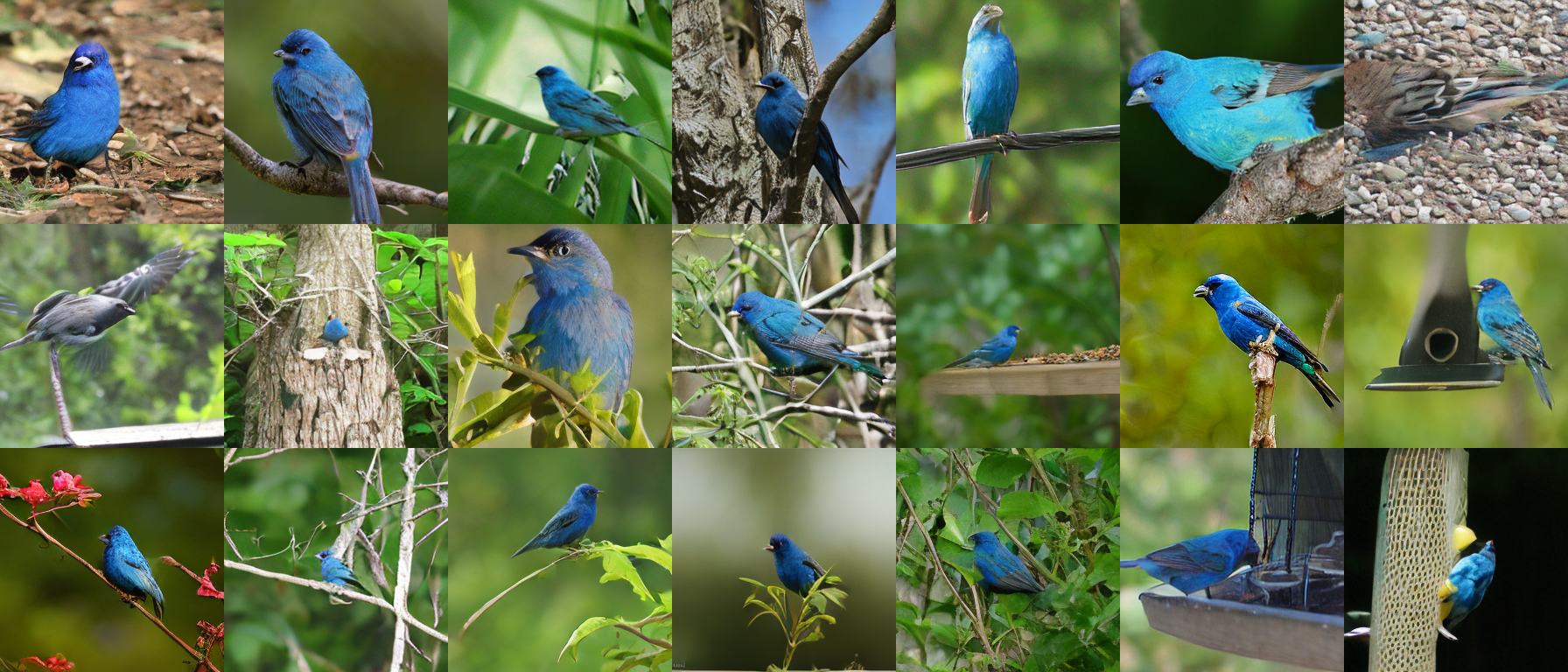}
    \caption{Indigo bunting (Class 014)}
\end{subfigure}\hfill
\begin{subfigure}{0.49\textwidth}
    \centering
    \includegraphics[width=\linewidth]{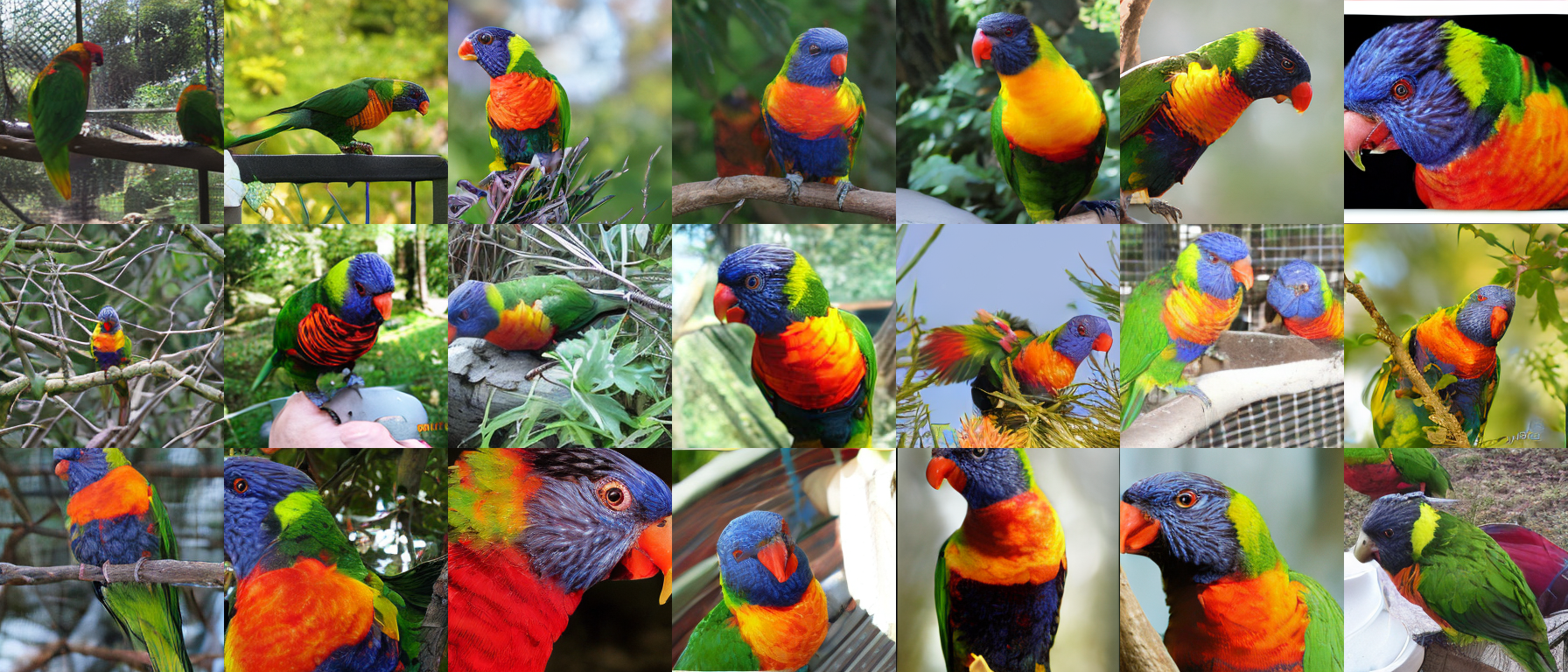}
    \caption{Lorikeet (Class 090)}
\end{subfigure}

\par\medskip
% Row 3
\begin{subfigure}{0.49\textwidth}
    \centering
    \includegraphics[width=\linewidth]{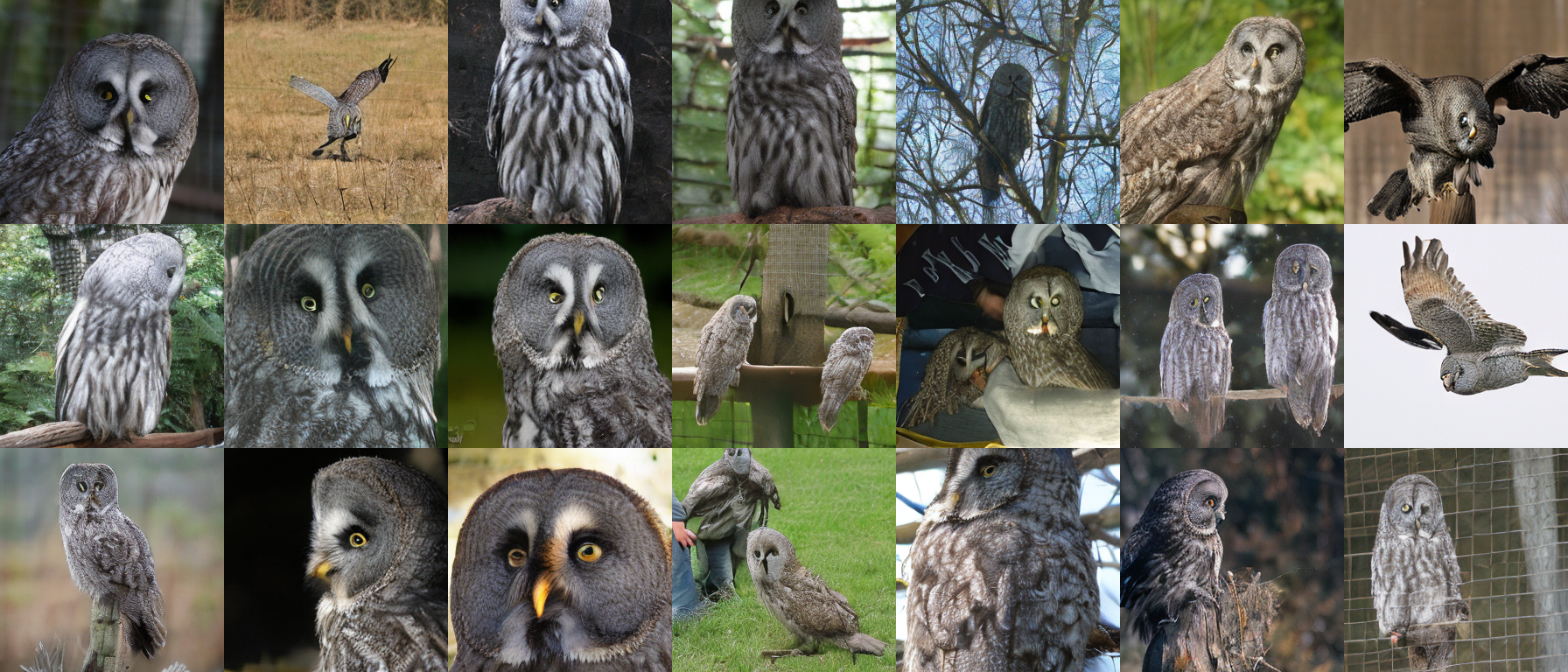}
    \caption{Great grey owl (Class 024)}
\end{subfigure}\hfill
\begin{subfigure}{0.49\textwidth}
    \centering
    \includegraphics[width=\linewidth]{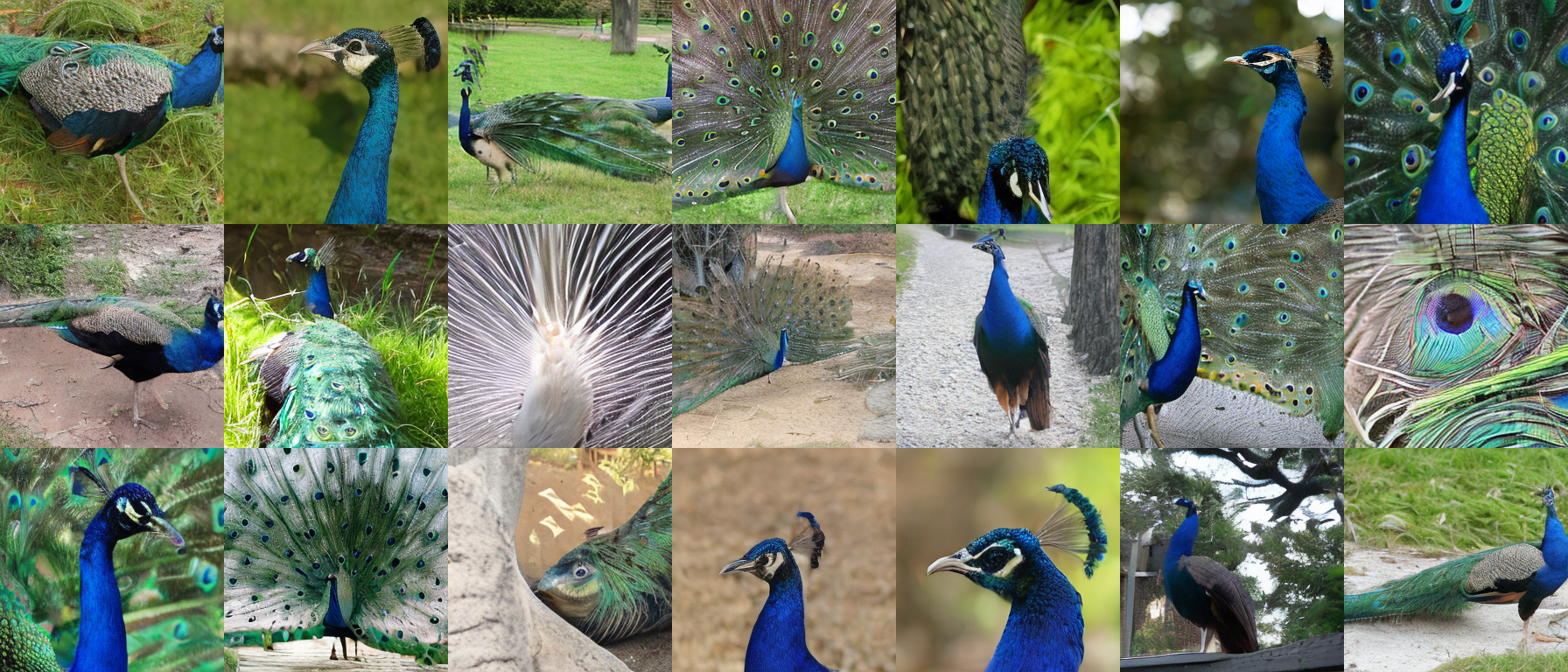}
    \caption{Peacock (Class 084)}
\end{subfigure}

\par\medskip
% Row 4
\begin{subfigure}{0.49\textwidth}
    \centering
    \includegraphics[width=\linewidth]{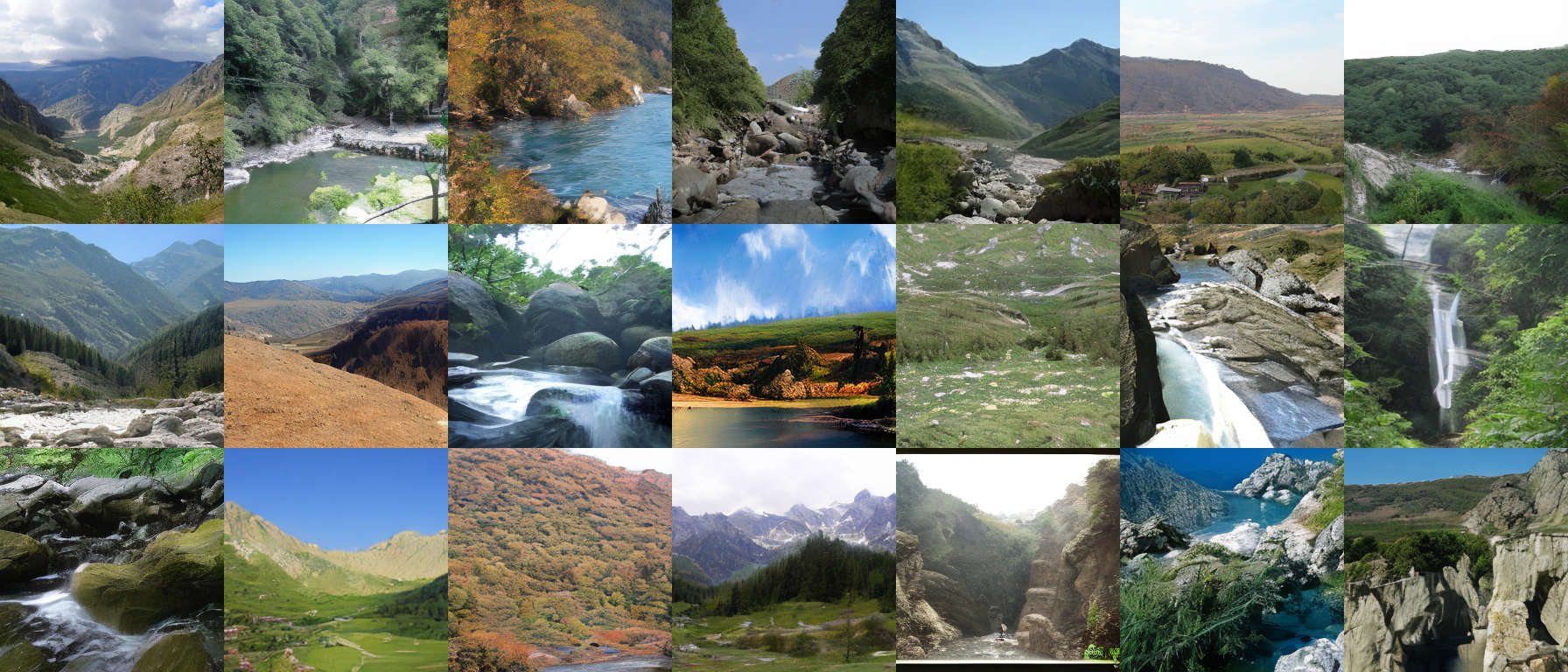}
    \caption{Valley (Class 979)}
\end{subfigure}\hfill
\begin{subfigure}{0.49\textwidth}
    \centering
    \includegraphics[width=\linewidth]{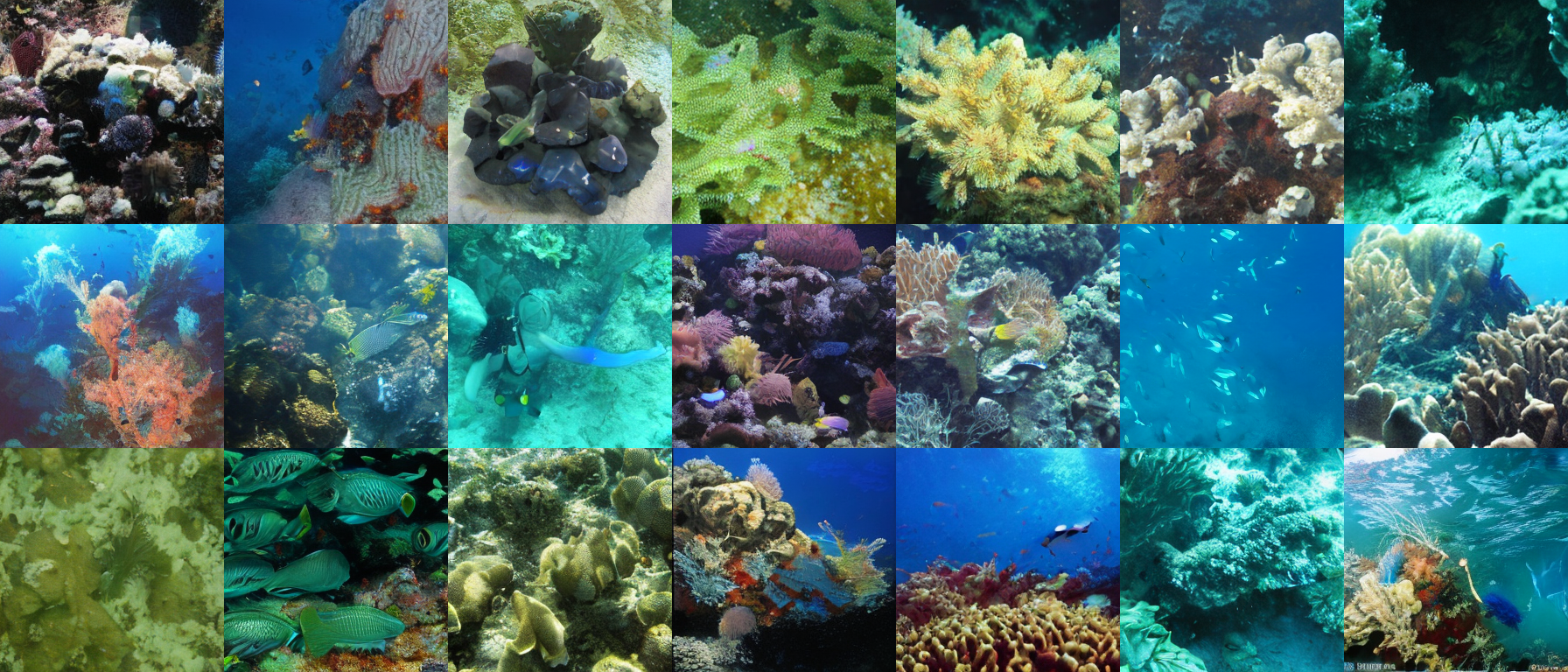}
    \caption{Coral reef (Class 973)}
\end{subfigure}

\par\medskip
% Row 5
\begin{subfigure}{0.49\textwidth}
    \centering
    \includegraphics[width=\linewidth]{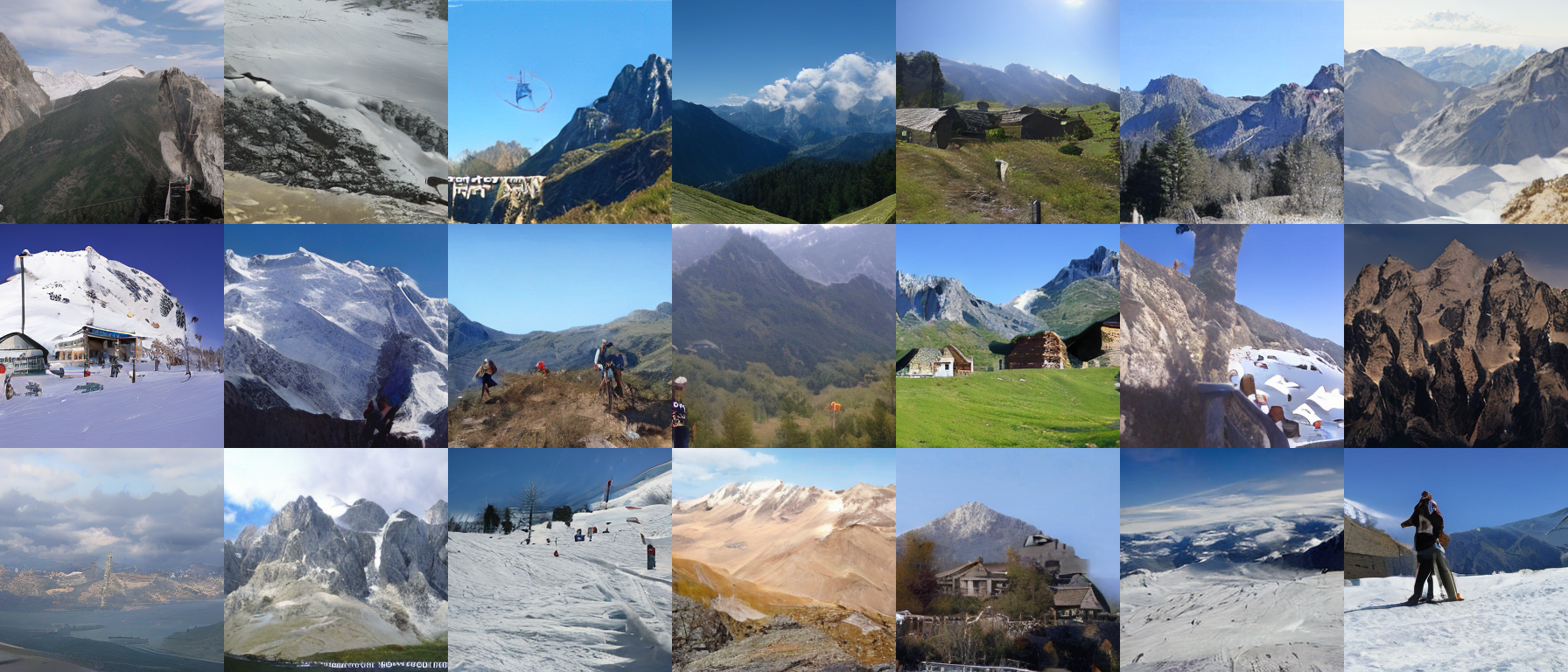}
    \caption{Alp (Class 970)}
\end{subfigure}\hfill
\begin{subfigure}{0.49\textwidth}
    \centering
    \includegraphics[width=\linewidth]{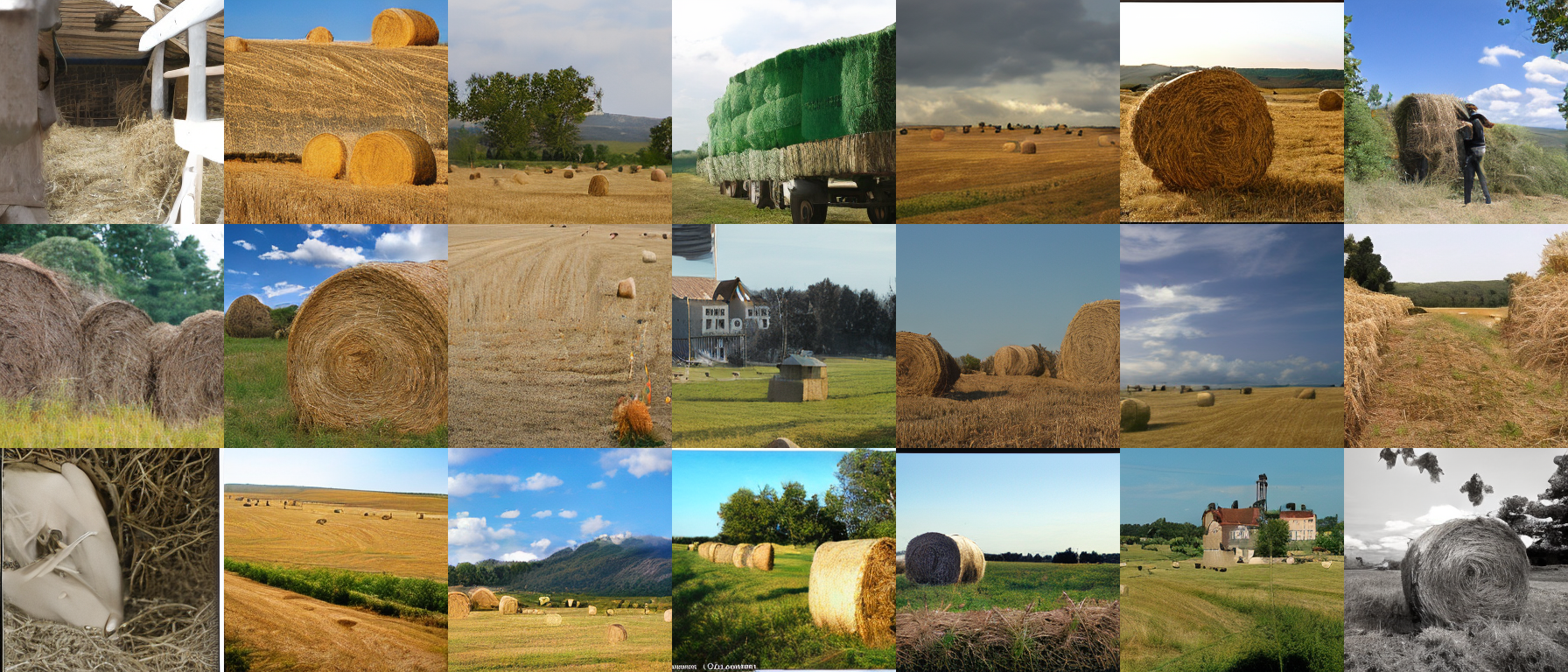}
    \caption{Hay (Class 958)}
\end{subfigure}

\par\medskip
% Row 6
\begin{subfigure}{0.49\textwidth}
    \centering
    \includegraphics[width=\linewidth]{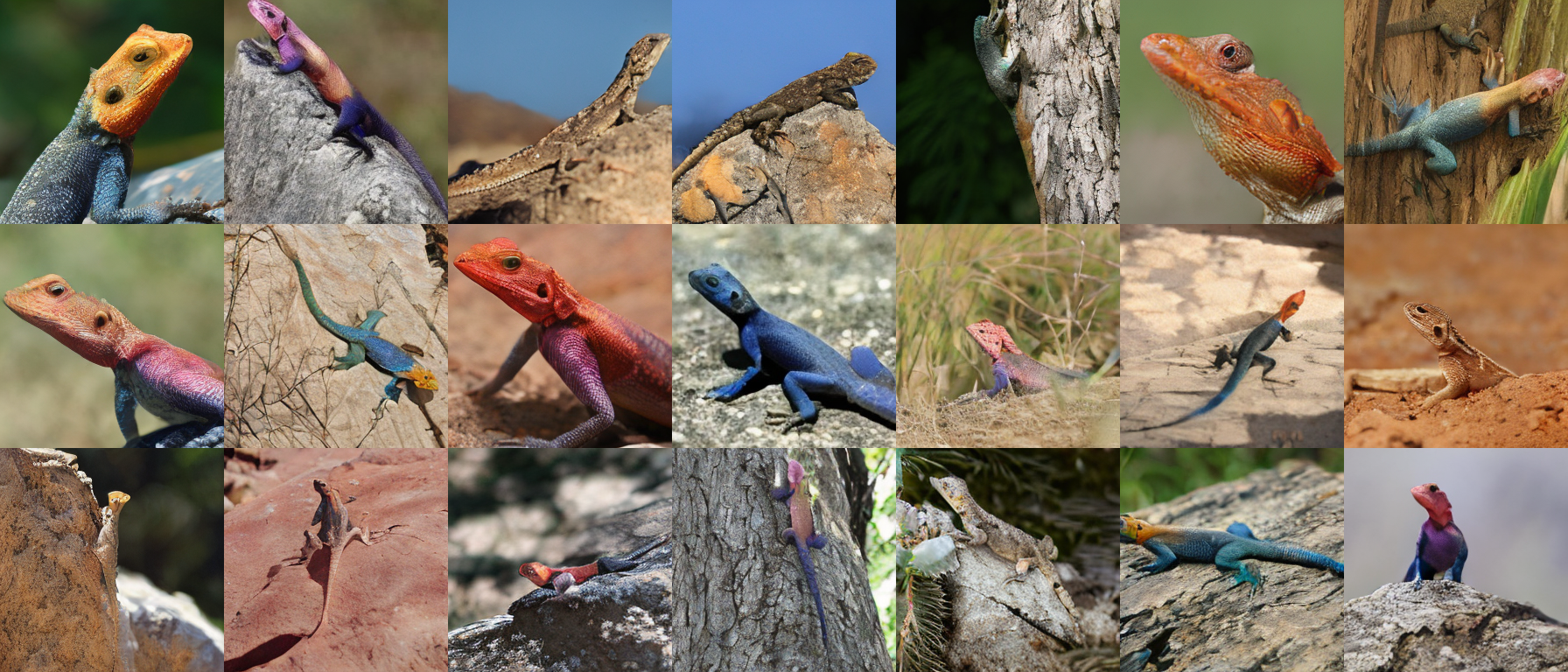}
    \caption{Agama (Class 042)}
\end{subfigure}\hfill
\begin{subfigure}{0.49\textwidth}
    \centering
    \includegraphics[width=\linewidth]{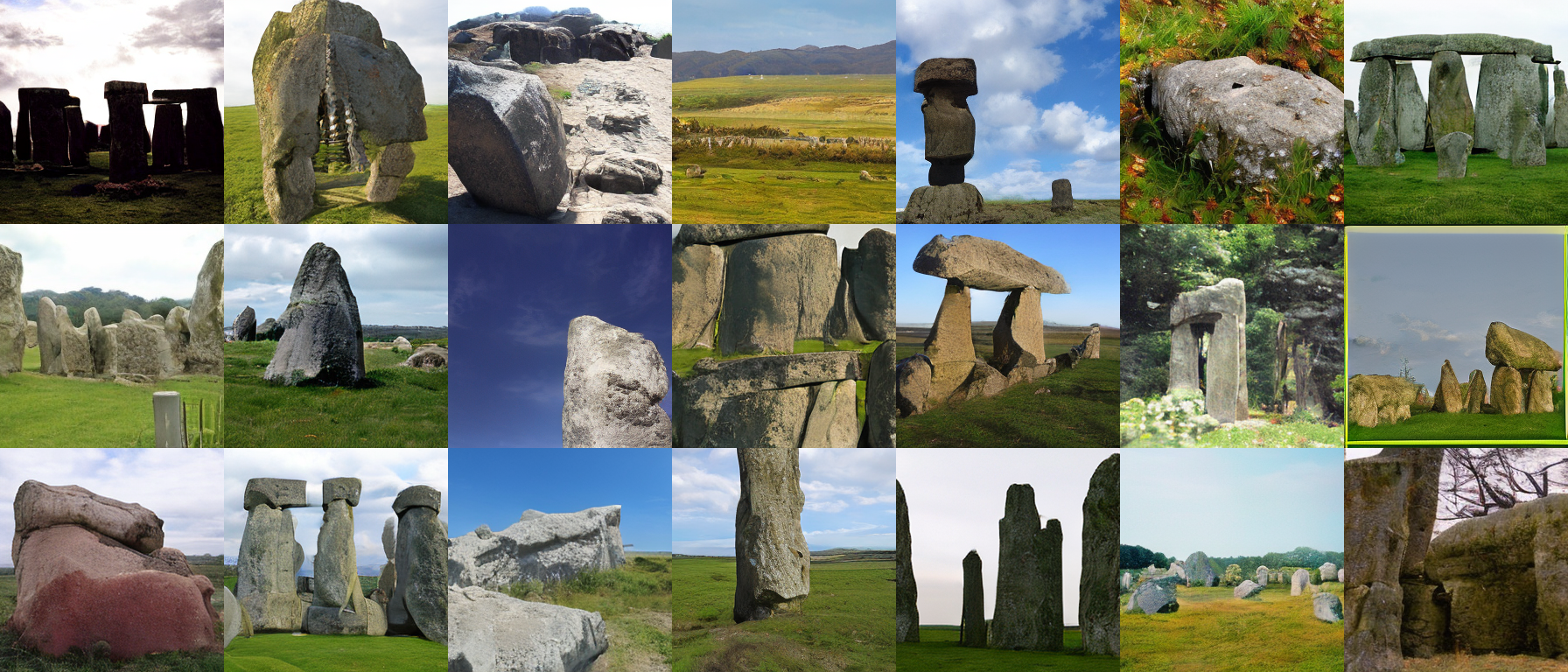}
    \caption{Megalith (Class 649)}
\end{subfigure}

\caption{
Random class-conditional ImageNet samples generated by our UOT-XL model
at 200K training steps with CFG $=1.14$.
Class indices follow the zero-based ImageNet-1K label mapping.
}
\label{fig:samples-six-by-two}
\end{figure*}

% ==================== Lion ====================
\begin{figure*}[t]
\centering
\captionsetup[subfigure]{font=small,skip=2pt}

\begin{subfigure}{0.95\textwidth}
    \centering
    \includegraphics[width=\linewidth]
        {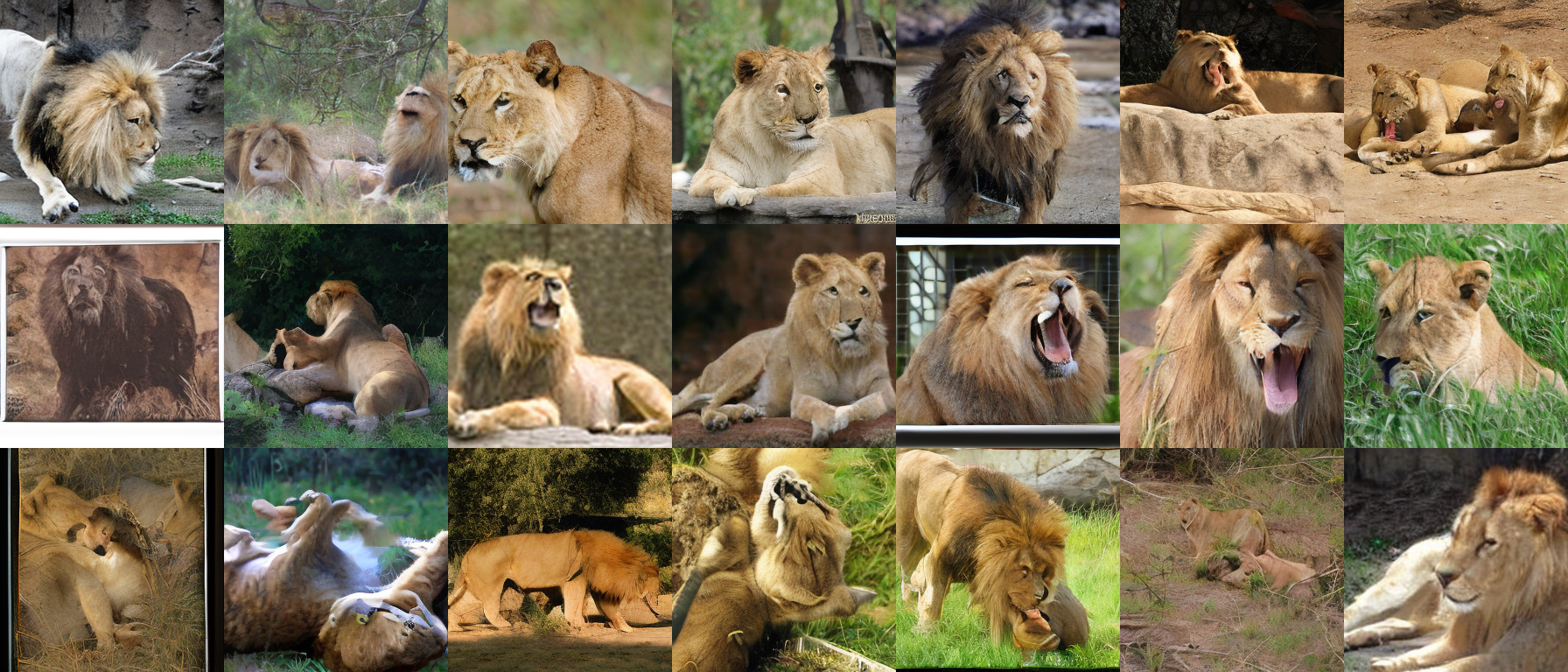}
    \caption{W-Flow-B (balanced OT)}
\end{subfigure}

\par\medskip
\begin{subfigure}{0.95\textwidth}
    \centering
    \includegraphics[width=\linewidth]
        {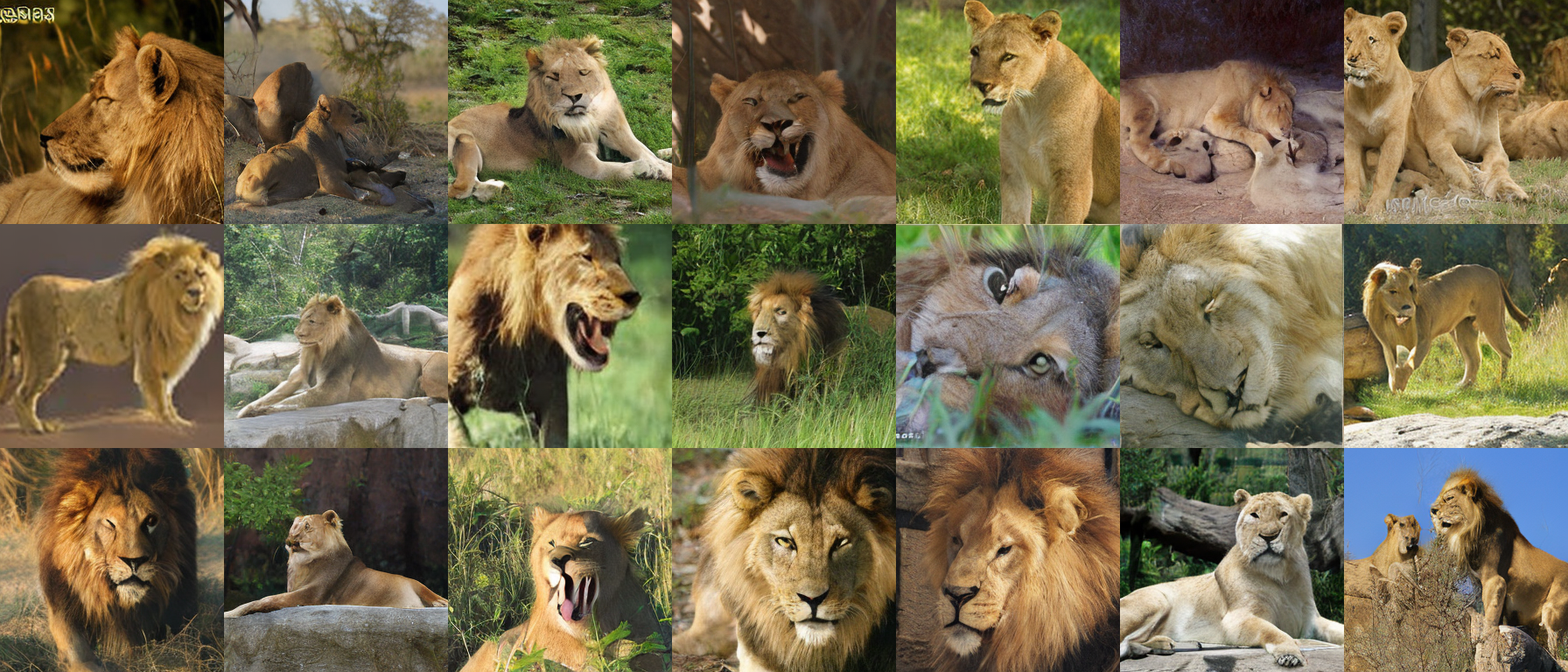}
    \caption{Ours-B (source-fixed UOT)}
\end{subfigure}

\caption{
Qualitative comparison for lion (ImageNet Class 291):
official W-Flow-B \citep{han2026one} (top) and our UOT-B (bottom).
Both models use 200K-step EMA checkpoints, CFG $=1.19$,
and sampling seed $=42$.
Corresponding positions use the same initial noise and class label.
All 21 samples per model from this seed are shown without filtering.
}
\label{fig:b-comparison-lion-seed42}
\end{figure*}

% ==================== Pembroke Welsh corgi ====================
\begin{figure*}[t]
\centering
\captionsetup[subfigure]{font=small,skip=2pt}

\begin{subfigure}{0.95\textwidth}
    \centering
    \includegraphics[width=\linewidth]
        {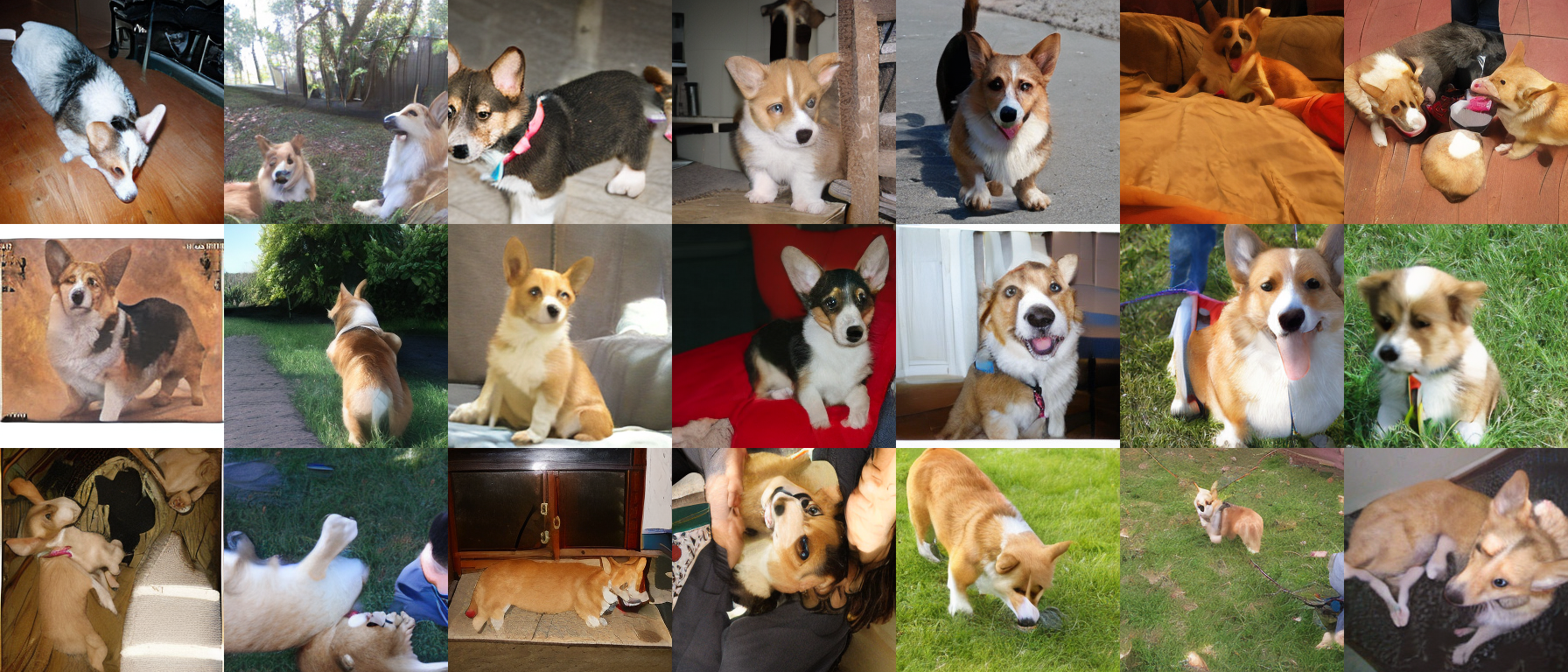}
    \caption{W-Flow-B (balanced OT)}
\end{subfigure}

\par\medskip
\begin{subfigure}{0.95\textwidth}
    \centering
    \includegraphics[width=\linewidth]
        {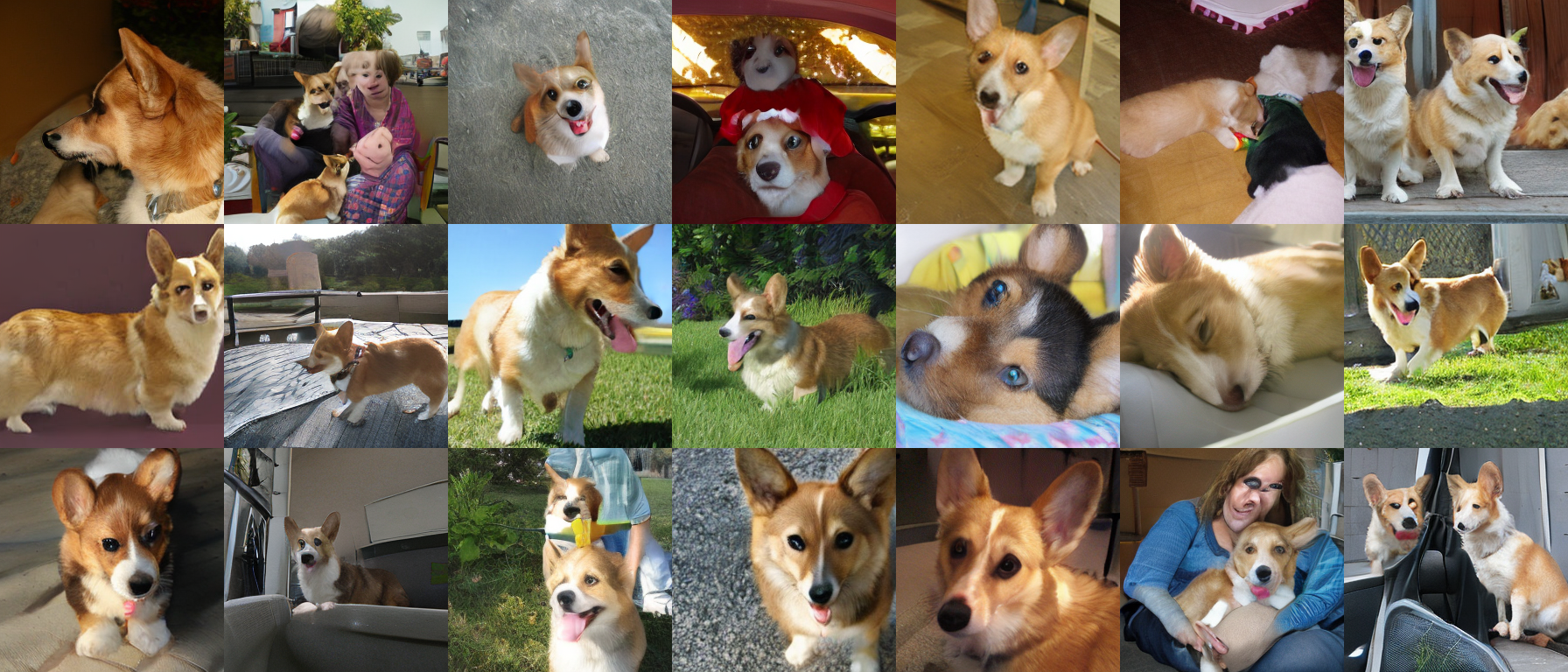}
    \caption{Ours-B (source-fixed UOT)}
\end{subfigure}

\caption{
Qualitative comparison for Pembroke Welsh corgi (ImageNet Class 263):
official W-Flow-B \citep{han2026one} (top) and our UOT-B (bottom).
Both models use 200K-step EMA checkpoints, CFG $=1.19$,
and sampling seed $=42$.
Corresponding positions use the same initial noise and class label.
All 21 samples per model from this seed are shown without filtering.
}
\label{fig:b-comparison-corgi-seed42}
\end{figure*}

% ==================== King penguin ====================
\begin{figure*}[t]
\centering
\captionsetup[subfigure]{font=small,skip=2pt}

\begin{subfigure}{0.95\textwidth}
    \centering
    \includegraphics[width=\linewidth]
        {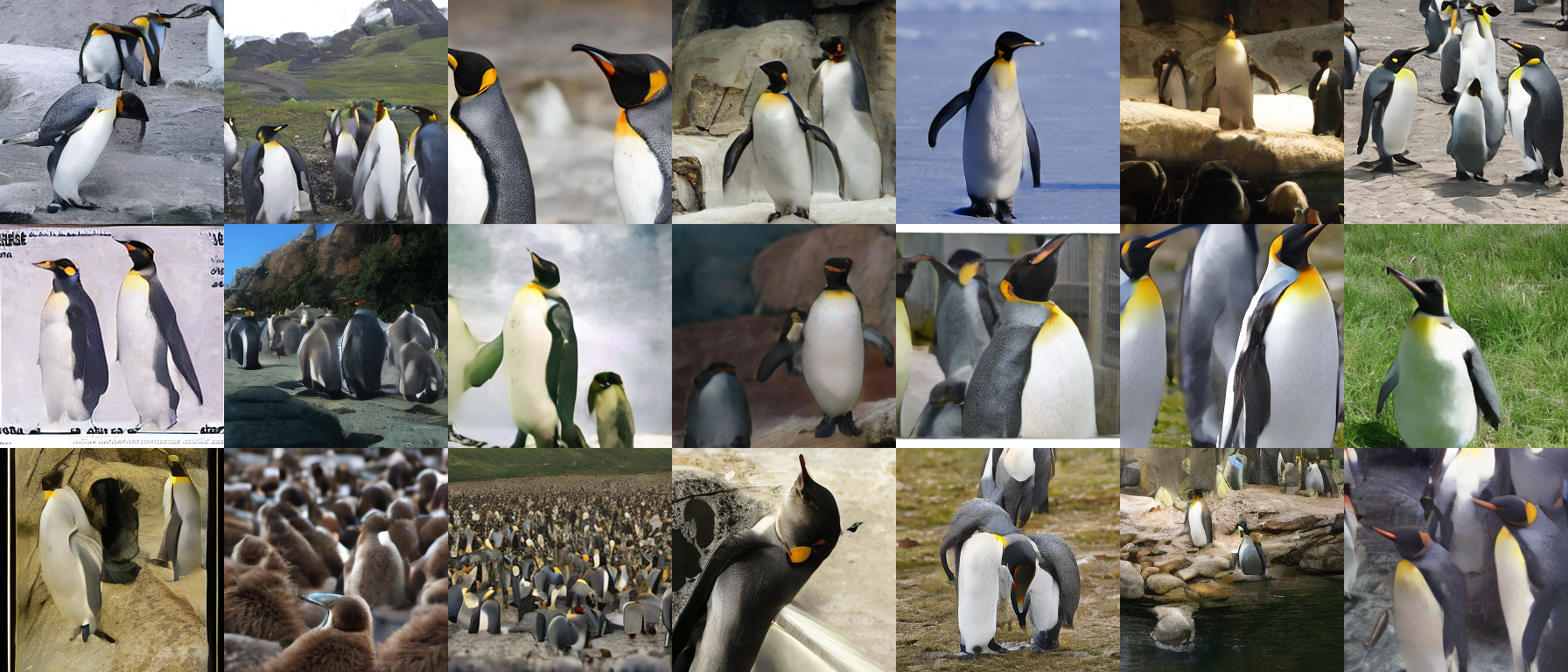}
    \caption{W-Flow-B (balanced OT)}
\end{subfigure}

\par\medskip
\begin{subfigure}{0.95\textwidth}
    \centering
    \includegraphics[width=\linewidth]
        {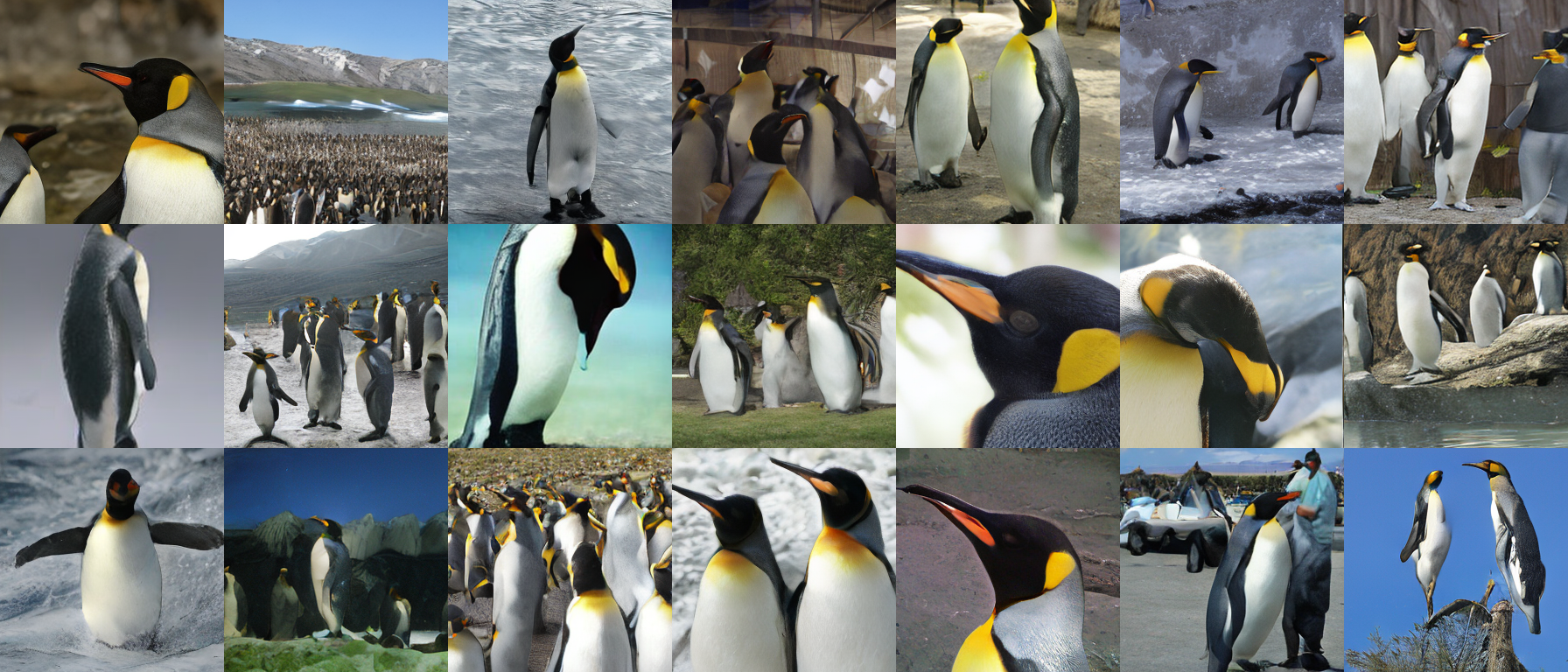}
    \caption{Ours-B (source-fixed UOT)}
\end{subfigure}

\caption{
Qualitative comparison for king penguin (ImageNet Class 145): official W-Flow-B \citep{han2026one} (top) and our UOT-B (bottom). Both models use 200K-step EMA checkpoints, CFG $=1.19$, and sampling seed $=42$. Corresponding positions use the same initial noise and class label.
All 21 samples per model from this seed are shown without filtering.}
\label{fig:b-comparison-penguin-seed42}
\end{figure*}

\begin{figure*}[t]
\centering
\captionsetup[subfigure]{font=small,skip=2pt}

\begin{subfigure}{0.95\textwidth}
    \centering
    \includegraphics[width=\linewidth]{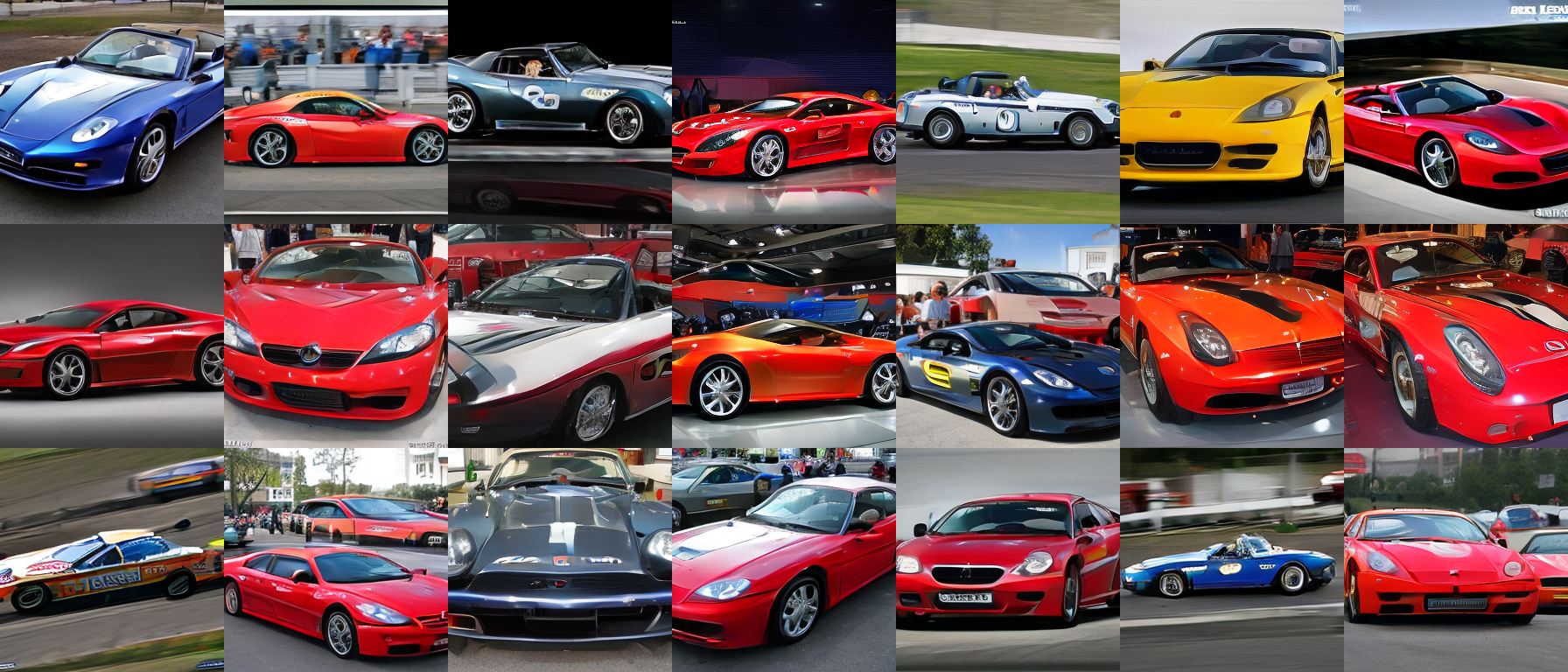}
    \caption{CFG $=2.0$}
\end{subfigure}

\par\medskip

\begin{subfigure}{0.95\textwidth}
    \centering
    \includegraphics[width=\linewidth]{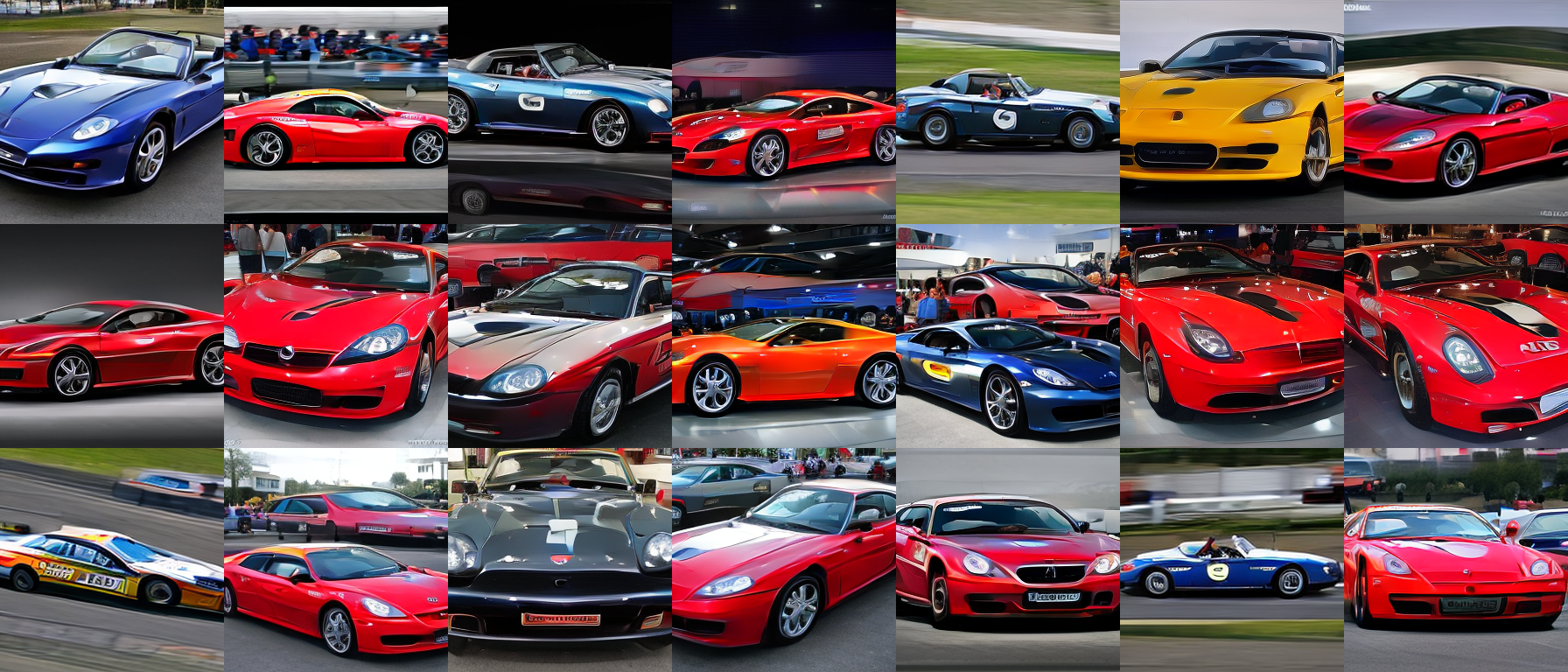}
    \caption{CFG $=3.0$}
\end{subfigure}

\caption{
Effect of classifier-free guidance scale on samples from our XL model. Both panels show sports cars (ImageNet Class 817), generated using the same 200K-step EMA checkpoint and sampling seed $=0$. Each panel contains 21 samples arranged in a $3\times7$ grid; corresponding positions use the same initial noise.
}
\label{fig:uot-xl-cfg-comparison}
\end{figure*}

\section{Limitations and Broader Impact}
\label{sec:limitations}

% \paragraph{Limitations.}
% Our analysis establishes force identities, kinetic lifting,
% and stationary-state examples under the stated assumptions. Positive definiteness is known for symmetric generators under the conditions of \citet{Sjourn2019SinkhornDF}, but remains open for the source-fixed energy, as do the trajectory assumptions required for convergence. The convergence statements concern the population flow; the mini-batch update with finitely many Sinkhorn iterations is the next object to analyze. Empirically, why the mild forward-plan reweighting observed on ImageNet (Figure~\ref{fig:source_fixed_uot_mass}) improves generation remains open, and the controlled comparison with balanced transport is at B/2; L/2 and XL/2 are system-level results with a fixed transport recipe. Experiments are on ImageNet-256 in the SD-VAE latent space with DiT generators; other datasets and latent spaces remain to be explored.

\paragraph{Limitations.}
Our analysis establishes force identities, kinetic lifting,
and stationary-state examples under the stated assumptions. Positive definiteness is known for symmetric generators under the conditions of \citet{Sjourn2019SinkhornDF}, but remains open for the source-fixed energy, as do the trajectory assumptions required for convergence. The convergence statements concern the population flow; the mini-batch update with finitely many Sinkhorn iterations is the next object to analyze. Empirically, why the mild forward-plan reweighting observed on ImageNet (Figure~\ref{fig:source_fixed_uot_mass}) improves generation remains open, and the controlled comparison with balanced transport is at B/2; L/2 and XL/2 are system-level results with a fixed transport recipe. Experiments are on ImageNet-256 in the SD-VAE latent space with DiT generators; other datasets and latent spaces remain to be explored.

\paragraph{Broader impact.}
One-step generation produces an image in a single forward pass, making high-quality synthesis cheaper to serve, easier to embed in interactive settings, and less energy-intensive at inference than diffusion and flow models. It carries the familiar hazards of generative imagery: fabricated or misleading content, reproduction of training-data stereotypes, and likenesses generated without consent. Our method adds one consideration: marginal reweighting may affect the representation of underrepresented samples, and aggregate precision and recall do not establish subgroup coverage, so subgroup-level evaluation remains necessary. We support standard release practices: content provenance or watermarking, documented training data and intended use, and access controls or output filtering outside research contexts.

\section{Additional generated samples}\label{app:samples}

\end{document}